\documentclass[a4paper]{article}
\usepackage{amsmath}
\usepackage{amsthm}
\usepackage{amssymb}
\usepackage{latexsym}
\usepackage{proof}
\usepackage{enumerate}
\usepackage[all]{xy}
\usepackage[dvipdfmx]{graphicx}
\usepackage{ascmac}
\usepackage{ulem}
\usepackage{enumitem}

\newtheorem{Theorem}{Theorem}[section]
\newtheorem{Lemma}[Theorem]{Lemma}
\newtheorem{Corollary}[Theorem]{Corollary}
\newtheorem{Definition}[Theorem]{Definition}
\newtheorem{Remark}[Theorem]{Remark}
\newtheorem{Example}[Theorem]{Example}
\newcommand{\jump}[1]{\ensuremath{[\![#1]\!]}}
\newcommand{\expa}[1]{\ensuremath{\langle\![#1]\!\rangle}}

\newlist{steps}{enumerate}{1}
\setlist[steps, 1]{label = Step \arabic*:}

\begin{document}
\title{\bf
  Completeness of Kozen's Axiomatization for the Modal $\mu$-Calculus:
  A Simple Proof
}
\author{
  Kuniaki Tamura\\
  15-9-103, Takasago 3-chome, Katsushika, Tokyo 125-0054, Japan\\
  E-mail: \texttt{kuniaki.tamura@gmail.com}
}
\maketitle

\begin{abstract}
The modal $\mu$-calculus, introduced by Dexter Kozen, is an extension of modal logic with fixpoint
operators. Its axiomatization, $\mathsf{Koz}$, was introduced at the same time and is an extension of
the minimal modal logic $\mathsf{K}$ with the so-called Park fixpoint induction principle. It took
more than a decade for the completeness of $\mathsf{Koz}$ to be proven, finally achieved by Igor
Walukiewicz. However, his proof is fairly involved.

In this article, we present an improved proof for the completeness of $\mathsf{Koz}$ which, although
similar to the original, is simpler and easier to understand.

\begin{flushleft}{\bf Keywords:}
  The modal $\mu$-calculus, completeness, $\omega$-automata. 
\end{flushleft}
\end{abstract}

\section{Introduction}\label{sec: introduction}

The \textit{modal $\mu$-calculus} originated with Scott and De Bakker \cite{bak-sco69} and was further
developed by Dexter Kozen \cite{DBLP:journals/tcs/Kozen83} into the main version currently used. It is used to describe and verify properties of labeled transition systems (Kripke models).
Many modal and temporal logics can be encoded into the modal $\mu$-calculus, including
$\mathsf{CTL}^{\ast}$ and its widely used fragments -- the linear temporal logic $\mathsf{LTL}$ and the
computational tree logic $\mathsf{CTL}$. The modal $\mu$-calculus also provides one of the strongest examples of the
connections between modal and temporal logics, automata theory and game theory (for example, see
\cite{DBLP:conf/dagstuhl/2001automata}). As such, the modal $\mu$-calculus is a very
active research area in both theoretical and practical computer science. We refer the reader to
Bradfield and Stirling's tutorial article \cite{Bradfield07modalmu-calculi} for a thorough introduction
to this formal system.

The difference between the modal $\mu$-calculus and modal logic is that the former has the
\textit{least fixpoint operator} $\mu$ and the \textit{greatest fixpoint operator} $\nu$ which
represent the least and greatest fixpoint solution to the equation $\alpha(x) = x$, where $\alpha(x)$ is a
monotonic function mapping some power set of possible worlds into itself.\footnote{
  In the modal $\mu$-calculus, the term \textit{state} is preferred to \textit{possible world} since it originated in
  the area of verification of computer systems.
  However, we do not use this terminology since it is reserved for \textit{state of automata} in this article. 
}
In Kozen's initial work \cite{DBLP:journals/tcs/Kozen83}, he proposed an axiomatization $\mathsf{Koz}$,
which was an extension of the minimal modal logic $\mathsf{K}$ with a further axiom and inference rule
-- the so-called Park fixpoint induction principle: 
$$
\infer[(\textsf{Prefix})]
    {\alpha(\mu x.\alpha(x)) \vdash \mu x.\alpha(x)}
    {}
    \qquad
\infer[(\textsf{Ind})]
    {\mu x.\alpha(x) \vdash \beta}
    {\alpha(\beta) \vdash \beta}
$$
The system $\mathsf{Koz}$ is very simple and natural; nevertheless, Kozen himself could not prove
completeness for the full language, but only for the negations of formulas of a special kind called the \textit{aconjunctive formula}. Completeness for the full language turned out to be a knotty problem and
remained open for more than a decade. Finally, Walukiewicz \cite{Walukiewicz2000142} solved this
problem positively, but his proof is quite involved.\footnote{
The difficulties of the proof have been pointed out, e.g., see
\cite{DBLP:conf/fossacs/CateF10,Lenzi05,Bradfield07modalmu-calculi,DBLP:journals/logcom/Alberucci09,Bezhanishvili20121}
}
The aim of this article is to provide an improved proof that is easier to understand. First, we outline Walukiewicz's proof and explain its difficulties, and then present our improvement.

The completeness theorem considered here is sometimes called weak completeness and requires that the
validity follows the provability; that is:
\begin{enumerate}
\item[\bf (a)]
  For any formula $\varphi$, if $\varphi$ is not satisfiable, then $\sim\!\varphi$ is provable in
  $\mathsf{Koz}$.            
\end{enumerate}
Here, $\sim\!\varphi$ denotes the negation of $\varphi$. Note that strong completeness cannot
be applied to the modal $\mu$-calculus since it lacks compactness. The first step of the proof is based
on the results of Janin and Walukiewicz \cite{conf/mfcs/JaninW95}, in which they introduced the class of
formulas called \textit{automaton normal form},\footnote{
  In the original article \cite{conf/mfcs/JaninW95}, this class of formulas was called
  the \textit{disjunctive formula}; however, the term \textit{automaton normal form}
  is the currently used terminology, to the author's knowledge.
}
and showed the following two theorems:
\begin{enumerate}
\item[\bf (b)]
  For any formula $\varphi$, we can construct an automaton normal form $\mathsf{anf}(\varphi)$ which is
  semantically equivalent to $\varphi$.
\item[\bf (c)]
  For any automaton normal form $\widehat{\varphi}$, if $\widehat{\varphi}$ is not satisfiable, then
  $\sim\!\widehat{\varphi}$ is provable in $\mathsf{Koz}$; that is, $\mathsf{Koz}$ is complete for
  the negations of the automaton normal form.
\end{enumerate}
The above theorems lead to the following Claim (d) for proving:
\begin{enumerate}
  \item[\bf (d)]
  For any formula $\varphi$, there exists a semantically equivalent automaton normal form
  $\widehat{\varphi}$ such that $\varphi\rightarrow\widehat{\varphi}$ is provable in $\mathsf{Koz}$.            
\end{enumerate}
Indeed, for any unsatisfiable formula $\varphi$, Claim (d) tells us that
$\sim\!\widehat{\varphi}\rightarrow\sim\!\varphi$ is provable; on the other hand, from Theorem (c) we
obtain that $\sim\!\widehat{\varphi}$ is provable; therefore $\sim\!\varphi$ is provable in
$\mathsf{Koz}$ as required. Hence, our target (a) is reduced to Claim (d). 

Another important tool is the concept of a \textit{tableau},
which is a tree structure that is labeled by some
subformulas of the primary formula $\varphi$ and is related to the satisfiability problem for $\varphi$. 
Niwinski and Walukiewicz
\cite{Niwinski199699} introduced a game played by two adversaries on a tableau and,
by analyzing these games, showed that:
\begin{enumerate}
\item[\bf (e)]
  For any unsatisfiable formula $\varphi$, there exists a structure called the \textit{refutation} for
  $\varphi$ which is a substructure of tableau.           
\end{enumerate}
Importantly, a refutation for $\varphi$ is very similar to a proof diagram for $\varphi$;
roughly speaking, the difference between them is that the former can have infinite
branches while the latter can not.
Walukiewicz shows that if the refutation for $\varphi$ satisfies a special \textit{thin} condition,
it can be transformed into a proof diagram for $\varphi$. In other words,
\begin{enumerate}
\item[\bf (f)]
  For any unsatisfiable formula $\varphi$ such that there exists a thin refutation for
  $\varphi$, $\sim\!\varphi$ is provable in $\mathsf{Koz}$.           
\end{enumerate}
Note that Claim (f) is a slight generalization of the completeness for the negations of the aconjunctive
formula in the sense that the refutation for an unsatisfiable aconjunctive formula is always thin, and
Claim (f) can be shown by the same method as Kozen's original argument.

The proof is based on confirming Claim (d) by induction on the length of $\varphi$, using (b) and (f).
The hardest step of induction is the case $\varphi = \mu x.\alpha(x)$.
Suppose $\varphi = \mu x.\alpha(x)$ and
that we could assume, by inductive hypothesis,
$\alpha(x)\rightarrow\widehat{\alpha}(x)$ is provable in $\mathsf{Koz}$ where $\widehat{\alpha}(x)$ is
an automaton normal form equivalent to $\alpha(x)$. For the inductive step, we want to discover
an automaton normal form $\widehat{\varphi}$ equivalent to $\mu x.\alpha(x)$ such that
$\mu x.\alpha(x)\rightarrow\widehat{\varphi}$ is provable. Note that since
$\alpha(x)\rightarrow\widehat{\alpha}(x)$ is provable,
$\mu x.\alpha(x)\rightarrow\mu x.\widehat{\alpha}(x)$ is also provable. Furthermore, $\mu x.\alpha(x)$
and $\mu x.\widehat{\alpha}(x)$ are equivalent to each other. Set
$\widehat{\varphi} := \mathsf{anf}(\mu x.\widehat{\alpha}(x))$. Then, it is sufficient to show that
$\mu x.\widehat{\alpha}(x)\rightarrow\widehat{\varphi}$ is provable, and thus, from the induction rule
$(\mathsf{Ind})$,
$\widehat{\alpha}(\widehat{\varphi})\rightarrow\widehat{\varphi}$ is provable. To show this, 
Walukiewicz developed
a new utility called \textit{tableau consequence}, which is a binary relation on the tableau and is
characterized using game theoretical notations.
The following two facts were then shown:
\begin{enumerate}
\item[\bf (g)]
  Let $\widehat{\alpha}(x)$ and $\widehat{\varphi}$ be formulas denoted above. Then the tableau for
  $\widehat{\varphi}$ is a
  consequence of the tableau for $\widehat{\alpha}(\widehat{\varphi})$.
\item[\bf (h)]
  For any automaton normal forms $\widehat{\beta}(y)$ and $\widehat{\psi}$, if the
  tableau for $\widehat{\psi}$ is a consequence of the tableau for
  $\widehat{\beta}(\widehat{\psi})$,
  then we can construct a thin refutation for
  $\sim\!(\widehat{\beta}(\widehat{\psi})\rightarrow\widehat{\psi})$.\footnote{
    More precisely, this assertion must be stated more generally to be applicable
    in other cases of an inductive step, see Lemma \ref{lem: completeness for tableau consequence}.
  }
\end{enumerate}
The real difficulty appeared when proving Claim (g). To establish this claim, Walukiewicz introduced
complicated functions across some tableaux and analyzed the properties of these functions very
carefully. Finally, Claims (f), (g) and (h) together immediately establish that
$\widehat{\alpha}(\widehat{\varphi})\rightarrow\widehat{\varphi}$ is provable in $\mathsf{Koz}$.
Thus, he obtained a proof for Claim (d), confirming completeness.
The following figure summarizes the Walukiewicz's proof strategy described above.
\begin{figure}[htbp]
  \centering
  \[
  \xymatrix@R=7pt{
    (b) \ar@{=>}[r] & (g) \ar@{=>}[dr] & (c) \ar@{=>}[dr]   &    \\
                         & (h) \ar@{=>}[r]   & (d) \ar@{=>}[r]    & (a) \\
    (e) \ar@{=>}[r] & (f) \ar@{=>}[ur] &                          &     \\
    }
  \]
  \caption{The outline of the Walukiewicz's proof.}
  \label{fig: conversion to proof diagram}
\end{figure}

This article's main contribution is the simplification of the proof of Claim (g) and (f).
For this purpose, we will apply the $\omega$-automaton conversion method introduced
by Safra \cite{safra1988} and K{\v r}et\`ins\`y et al. \cite{Kret2017}.
It is shown that the proof of claim (g) and (f) are much more visible by using the mechanism called
\textit{index appearance record} provided by those automata.
In addition, we make improvements to some terms and concepts.
For example, the concept of the tableau consequence
will be redefined as a concept similar to the concept of \textit{bisimulation}
(instead of the game theoretical notations), which is one of the most fundamental and standard notions
in the model theory of modal and its extensional logics.
As a result, the proof of (h) is a little easier to understand.
Consequently, although our proof of
completeness does not include any innovative concepts, it is far more concise than the original proof.

The author hopes that the method given in this article may assist investigation of the modal
$\mu$-calculus and related topics.

\subsection{Outline of the article}\label{subsec: outline of the article}
The remainder of this article is organized as follows: in the following subsection
\ref{subsec: notation},
we will define some terminologies used within the article. Section \ref{sec: the modal mu-calculus}
gives basic definitions of the syntax and semantics of the modal $\mu$-calculus.
Section \ref{sec: automata} 
introduce well known results concerning $\omega$-automata.
The automaton mechanism used in the main proof will be
introduced in this section.
Section \ref{sec: application} is an application of Section \ref{sec: automata}.
We will prove claims (b) and (f) using the theory of $\omega$-automata.
Section \ref{sec: completeness} is the final section and 
contains the principle part of this article -- the proof of Claim (g).
Finally, we prove the completeness of $\mathsf{Koz}$ by showing Claim (d).

\subsection{Notation}\label{subsec: notation}

\noindent
\begin{description}
\item[Sets:]
Let $X$ be an arbitrary set. The \textit{cardinality of $X$} is denoted $|X|$. The
\textit{power set of $X$} is denoted $\mathcal{P}(X)$. $\omega$ denotes the set of natural numbers.

\item[Sequences:]
A finite sequence over some set $X$ is a function $\pi: \{1, \dots, n\} \rightarrow X$ where
$1 \leq n$. An infinite sequence over $X$ is a function $\pi: \omega\setminus \{0\} \rightarrow X$.
Here, a sequence can refer to either a finite or infinite sequence. The length of a sequence $\pi$ is
denoted $|\pi|$. Let $\pi$ be a sequence over $X$. The set of $x \in X$ which appears infinitely often
in $\pi$ is denoted $\mathsf{Infinite}(\pi)$. We
denote the $n$-th element in $\pi$ by $\pi[n]$ and the fragment of $\pi$ from the $n$-th element to
the $m$-th element by $\pi[n, m]$. For example, if $\pi = \mathsf{aabbcddd}$, then
$\pi[5] = \mathsf{c}$ and $\pi[2, 6] = \mathsf{abbcd}$. Note that when $\pi$ is a finite non-empty
sequence, $\pi[|\pi|]$ denotes the tail of $\pi$.

\item[Alphabets:]
Suppose that $\Sigma$ is a non-empty finite set. Then we may call $\Sigma$ an \textit{alphabet}
  and its element $a \in \Sigma$ a \textit{letter}. We denote the set of finite sequences over $\Sigma$
  by $\Sigma^{\ast}$,
  the set of non-empty finite sequences over $\Sigma$ by $\Sigma^{+}$,
  and the set of infinite sequences over $\Sigma$ by $\Sigma^{\omega}$. As usual, we
  call an element of $\Sigma^{\ast}$ a \textit{word}, an element of $\Sigma^{\omega}$
  an \textit{$\omega$-word}, a set of finite words $\mathcal{L} \subseteq \Sigma^{\ast}$ a
  \textit{language} and, a set of $\omega$-words $\mathcal{L'} \subseteq \Sigma^{\omega}$ an
  \textit{$\omega$-language}. The notion of the \textit{factor} on words is defined as usual:
  for two words $u, v \in \Sigma^{\ast} \cup \Sigma^{\omega}$, $u$ is a
  factor of $v$ if $v = xuy$ for some $x, y \in \Sigma^{\ast} \cup \Sigma^{\omega}$.

  \item[Graphs:]
  In this article, the term \textit{graph} refers to a directed graph. That is, a graph is a pair
  $\mathcal{G} = (V, E)$ where $V$ is an arbitrary set of \textit{vertices} and $E$ is an arbitrary
  binary relation over $V$, i.e., $E \subseteq V \times V$. A vertex $u$ is said to be an $E$-successor
  (or simply a successor) of a vertex $v$ in $\mathcal{G}$ if $(v, u) \in E$.  For any vertex
  $v$, we denote the set of all $E$-successors of $v$ by $E(v)$. The sequence
  $\pi \in V^{\ast} \cup V^{\omega}$ is called an $E$-\textit{sequence} if $\pi[n+1] \in E(\pi[n])$ for
  any $n < |\pi|$. $E^{\ast}$ denotes the reflexive transitive closure of $E$ and $E^{+}$ denotes the
  transitive closure of $E$.

  \item[Trees:]
  The term \textit{tree} is used to mean a \textit{rooted direct tree}. More precisely, a tree is a
  triple $\mathcal{T} = (T, C, r)$ where $T$ is a set of \textit{nodes}, $r \in T$ is a \textit{root} of
  the tree and, $C$ is a \textit{child relation}, i.e., $C \subseteq T \times T$ such that for any
  $t \in T \setminus \{r\}$, there is exactly one $C$-sequence starting at $r$ and ending at $t$.
  An unique $C$-sequence that starts at $r$ and ends at $t$ is denoted by $\vec{rt}$.
  As
  usual, we say that $u$ is a child of $t$ (or $t$ is a parent of $u$) if $(t, u) \in C$. A node
  $t \in T$ is
  a \textit{leaf} if $C(t) = \emptyset$. A \textit{branch} of
  $\mathcal{T}$ is either a finite $C$-sequence starting at $r$ and ending at a leaf or an
  infinite $C$-sequence starting at $r$.

  \item[Unwinding:]
  Let $\mathcal{G} = (V, E)$ be a graph. An \textit{unwinding} of $\mathcal{G}$ on $v \in V$ is the
  tree structure $\mathsf{UNW}_{v}(\mathcal{G}) = (T, C, r)$ where:
  \begin{itemize}
    \item $T$ consists of all finite non-empty $E$-sequences that start at $v$,
    \item $(\pi, \pi') \in C$ if and only if; $|\pi|+1 = |\pi'|$,
    $\pi = \pi'[1, |\pi|]$ and $(\pi[|\pi|], \pi'[|\pi'|]) \in E$, and
    \item $r := v$.
  \end{itemize}
  This concept can be extended naturally into a graph with some additional relations
  or functions. 
  For example, let $\mathcal{S} = (V, E, f)$ be a structure where $\mathcal{G} = (V, E)$ is a graph
  and $f$ is a function with domain $V$. Then we define
  $\mathsf{UNW}_{v}(\mathcal{S}) := (\mathsf{UNW}_{v}(\mathcal{G}), f')$ as $f'(\pi) := f(\pi[|\pi|])$
  for any $\pi \in V^{+}$. Note that we use the same symbol $f$ instead of $f'$ in $\mathsf{UNW}_{v}(\mathcal{S})$
  if there is no danger of confusion.

  \item[Functions:]
  Let $f$ be a function from some set $X$ to some set $Y$. We define the new function $\vec{f}$ from
  $X^{+} \cup X^{\omega}$ to $Y^{+} \cup Y^{\omega}$ as:
  \[ \vec{f}(\pi) := f(\pi[1])f(\pi[2])\cdots \]
  where $\pi \in X^{+} \cup X^{\omega}$. It is obvious that for any $\pi \in X^{+} \cup X^{\omega}$,
  we have $|\pi| = |\vec{f}(\pi)|$.
\end{description}

\section{The modal $\mu$-calculus}\label{sec: the modal mu-calculus}
We will now introduce the syntax, semantics and axiomatization $\mathsf{Koz}$ of the modal
$\mu$-calculus, and then present some additional concepts and results for use in the following sections.

\subsection{Syntax}\label{subsec: syntax}

\begin{Definition}[\bf Formula]\label{def: formula}\normalfont
Let $\mathsf{Prop} = \{ p, q, r, x, y, z, \dots \}$ be an infinite countable set of
\textit{propositional variables}. Then the collection of the
\textit{modal $\mu$-formulas} is defined as follows:
\[ 
    \varphi ::= (\top), (\bot), (p) \mid (\neg p) \mid (\varphi \vee \varphi) \mid
                (\varphi \wedge \varphi) \mid (\Diamond \varphi) \mid (\square \varphi) \mid
                (\mu x.\varphi) \mid (\nu x.\varphi)
\]
where $p, x \in \textsf{Prop}$.
Moreover, for formulas of the form $(\eta x.\varphi)$ with $\eta \in \{ \mu, \nu \}$,
we require that each occurrence of $x$ in $\varphi$ is positive; that is, $\neg x$ is not a
subformula of $\varphi$. Henceforth in this article, we will use $\eta$ to denote $\mu$ or $\nu$.
A formula of the form $p$ or $\neg p$ for $p \in \mathsf{Prop}$,
$\top$ and $\bot$ is called \textit{literal}. We
use the term $\mathsf{Lit}$ to refer to the set of all literals, i.e.,
$\mathsf{Lit} := \{ p, \neg p, \bot, \top \mid p \in \mathsf{Prop} \}$.
We call $\mu$ and $\nu$
\textit{the least fixpoint operator} and \textit{the greatest fixpoint operator}, respectively.
\end{Definition}

\begin{Remark}\normalfont
In Definition \ref{def: formula}, we confined the formula to a \textit{negation normal form}; that is, the
negation symbol may only be applied to propositional variables. However, this restriction can be
inconvenient, and so we extend the concept of the negation to an arbitrary formula $\varphi$
(denoted by $\sim\! \varphi$) inductively as follows: 
\begin{itemize}
    \item $\sim\! \top := \bot$,\; $\sim\! \bot := \top$. 
    \item $\sim\! p := \neg p$,\; $\sim\! \neg p := p$ for $p \in \mathsf{Prop}$. 
    \item $\sim\! (\varphi \vee \psi) := ((\sim\! \varphi)\wedge(\sim\! \psi))$,\;
          $\sim\! (\varphi \wedge \psi) := ((\sim\! \varphi) \vee (\sim\! \psi))$.
    \item $\sim\! (\Diamond \varphi) := (\square (\sim\! \varphi))$,\;
          $\sim\! (\square \varphi) := (\Diamond (\sim\! \varphi))$.
    \item $\sim\! (\mu x.\varphi(x)) := (\nu x.(\sim\! \varphi(\neg x)))$,\;
          $\sim\! (\nu x.\varphi(x)) := (\mu x.(\sim\! \varphi(\neg x)))$.
\end{itemize}
We introduce \textit{implication} $(\varphi \rightarrow \psi)$ as $((\sim\! \varphi) \vee \psi)$ and
\textit{equivalence} $(\varphi \leftrightarrow \psi)$ as
$((\varphi \rightarrow \psi) \wedge (\psi \rightarrow \varphi))$ as per the usual notation.
To minimize
the use of parentheses, we assume the following precedence of operators from highest to lowest:
$\neg$, $\sim$, $\Diamond$, $\square$, $\eta x$, $\vee$, $\wedge$, $\rightarrow$ and
$\leftrightarrow$. Moreover, we often abbreviate the outermost parentheses. 
For example, we write $\Diamond p \rightarrow q$ for $((\Diamond p) \rightarrow q)$
but not for $(\Diamond (p \rightarrow q))$.
\end{Remark}

As fixpoint operators $\mu$ and $\nu$ can be viewed as quantifiers, we use the standard
terminology and notations for quantifiers. We denote the set of all propositional variables
appearing free in $\varphi$ by $\mathsf{Free}(\varphi)$, and those appearing bound by
$\mathsf{Bound}(\varphi)$.
If $\psi$ is a subformula of $\varphi$, we write $\psi \leq \varphi$. We write $\psi < \varphi$ when
$\psi$ is a proper subformula.
$\mathsf{Sub}(\varphi)$ is the set of all subformulas of $\varphi$ and
$\mathsf{Lit}(\varphi)$ denotes the set of all literals which are subformulas of $\varphi$.
Let $\varphi(x)$ and $\psi$ be two formulas. The \textit{substitution} of all free appearances of $x$
with $\psi$ into $\varphi$ is denoted $\varphi(x)[x/\psi]$ or sometimes simply $\varphi(\psi)$.
As with predicate logic, we prohibit substitution when a new binding
relation will occur by that substitution.

The following two definitions regarding formulas will be used frequently
in the remainder of the article.
\begin{Definition}[\bf Well-named formula]\label{def: well-named formula}\normalfont
The set of \textit{well-named formulas} $\mathsf{WNF}$ is defined inductively as follows:
\begin{enumerate}
  \item $\mathsf{Lit} \subseteq \mathsf{WNF}$.
  \item Let $\alpha, \beta \in \mathsf{WNF}$ where
    $\mathsf{Bound}(\alpha) \cap \mathsf{Free}(\beta) = \emptyset$ and
    $\mathsf{Free}(\alpha) \cap \mathsf{Bound}(\beta) = \emptyset$. Then
    $\alpha \vee \beta, \alpha \wedge \beta \in \mathsf{WNF}$.
  \item Let $\alpha \in \mathsf{WNF}$. Then $\Diamond \alpha, \square \alpha \in \mathsf{WNF}$.
  \item Let $\alpha(x) \in \mathsf{WNF}$ where $x \in \mathsf{Free}(\alpha(x))$ occurs at once,
    positively, moreover, $x$ is in the scope of some modal operators.
    Then $\eta x.\alpha(x) \in \mathsf{WNF}$.
\end{enumerate}
If $\varphi$ is well-named and $x$ is bound in $\varphi$, then there is exactly one subformula which
binds $x$; this formula is denoted $\eta_{x}x.\varphi_{x}(x)$.
\end{Definition}

\begin{Definition}[\bf Alternation depth]\label{def: alternation depth}\normalfont
Given a formula $\varphi$,
\begin{enumerate}
  \item Let $\preceq^{-}_{\varphi}$ be a binary relation on $\mathsf{Bound}(\varphi)$ such that
        $x \preceq^{-}_{\varphi} y$ if and only if $x \in \mathsf{Free}(\varphi_{y}(y))$. The
        \textit{dependency order} $\preceq_{\varphi}$ is defined as the transitive closure of
        $\preceq^{-}_{\varphi}$.

  \item A sequence $\langle x_{1}, x_{2}, \dots, x_{K}\rangle \in \mathsf{Bound}(\varphi)^{+}$ is
        said to be an alternating chain if:
        \[
          x_{1}\preceq^{-}_{\varphi}x_{2}\preceq^{-}_{\varphi}\dots\preceq^{-}_{\varphi}x_{K}
        \]
        and $\eta_{x_k} \neq \eta_{x_{k+1}}$ for every $k \in \omega$ such that $1 \leq k \leq K-1$. 
        The \textit{alternation depth} of $\alpha$ (denoted $\mathsf{alt}(\alpha)$) is the maximal
        length of alternating chains such that $x_{1} \leq \alpha$.
        That is, the alternation depth of $\alpha$ is
        the maximal number of alternations between $\mu$- and $\nu$-operators in $\alpha$.
  \item A \textit{priority function} $\Omega_{\varphi}: \mathsf{Sub}(\varphi) \rightarrow \omega$ is defined as follows:  
  \begin{eqnarray}\label{eq: priority of formulas}
  \Omega_{\varphi}(\psi) := \left\{
  \begin{array}{ll}
    \mathsf{alt}(x)&
    \text{if $\psi = x$, $\eta_{x} = \mu$ and $\mathsf{alt}(x) \equiv 0 \pmod 2$,}\\

    \mathsf{alt}(x)-1&
    \text{if $\psi = x$, $\eta_{x} = \mu$ and $\mathsf{alt}(x) \equiv 1 \pmod 2$,}\\

    \mathsf{alt}(x)&
    \text{if $\psi = x$, $\eta_{x} = \nu$ and $\mathsf{alt}(x) \equiv 1 \pmod 2$,}\\

    \mathsf{alt}(x)-1&
    \text{if $\psi = x$, $\eta_{x} = \nu$ and $\mathsf{alt}(x) \equiv 0 \pmod 2$,}\\

    0&
    \text{otherwise.}
  \end{array}\right.
  \end{eqnarray}
  The number $\Omega_{\varphi}(\psi)$ is called the priority of $\psi$.
\end{enumerate}
\end{Definition}

\begin{Example}\normalfont
For a formula $\varphi = \mu x.\nu y.(\Diamond x \vee (\mu z.(\Diamond z \wedge \square y)))$, we have
$\mathsf{alt}(\varphi) = 3$ since $x\preceq^{-}_{\varphi}y\preceq^{-}_{\varphi}z$ with
$\eta_{x} \neq \eta_{y}$ and $\eta_{y} \neq \eta_{z}$. Note that although
$x \notin \mathsf{Free}(\varphi_{z}(z))$, we have $x\preceq_{\varphi}z$.
\end{Example}

\subsection{Semantics}\label{subsec: semantics}
\begin{Definition}[\bf Kripke model]\normalfont
A \textit{Kripke model} for the modal $\mu$-calculus is a structure $\mathcal{S}=(S, R, \lambda)$ such
that:
\begin{itemize}
  \item $S = \{ s, t, u, \dots \}$ is a non-empty set of \textit{possible worlds}.
  \item $R$ is a binary relation over $S$ called the \textit{accessibility relation}.
  \item $\lambda:\mathsf{Prop}\rightarrow \mathcal{P}(S)$ is a \textit{valuation}.   
\end{itemize}
\end{Definition}

\begin{Definition}[\bf Denotation]\normalfont
Let $\mathcal{S} = (S, R, \lambda)$ be a Kripke model and let $x$ be a propositional variable. Then for
any set of possible worlds $T \in \mathcal{P}(S)$, we can define a new valuation $\lambda[x \mapsto T]$
on $S$ as follows: 
\begin{eqnarray*}
  \lambda[x \mapsto T](p):=
  \left\{\begin{array}{ll}
      T&\text{if $p = x$,}\\
      \lambda(p)&\text{otherwise.}\\
  \end{array}\right.
\end{eqnarray*}
Moreover, $\mathcal{S}[x \mapsto T]$ denotes the Kripke model $(S, R, \lambda[x \mapsto T])$. A
\textit{denotation} $\jump{\varphi}_{\mathcal{S}} \in \mathcal{P}(S)$ of a formula
$\varphi$ on $\mathcal{S}$ is defined inductively on the structure
of $\varphi$ as follows: 
\begin{itemize}
  \item $\jump{\bot}_{\mathcal{S}} := \emptyset$ and
        $\jump{\top}_{\mathcal{S}} := S$.
  \item $\jump{p}_{\mathcal{S}} := \lambda(p)$ and
        $\jump{\neg p}_{\mathcal{S}} := S \setminus \lambda(p)$ for any $p \in \mathsf{Prop}$.
  \item $\jump{\varphi \vee \psi}_{\mathcal{S}} :=
        \jump{\varphi}_{\mathcal{S}} \cup \jump{\psi}_{\mathcal{S}}$ and
        $\jump{\varphi \wedge \psi}_{\mathcal{S}} :=
        \jump{\varphi}_{\mathcal{S}} \cap \jump{\psi}_{\mathcal{S}}.$
  \item $\jump{\Diamond \varphi}_{\mathcal{S}} :=
        \{ s \mid \exists t \in S, (s, t) \in R \wedge t \in \jump{\varphi}_{\mathcal{S}}\}.$
  \item $\jump{\square \varphi}_{\mathcal{S}} := \{ s \mid \forall t \in S, (s, t) \in R
        \Longrightarrow t \in \jump{\varphi}_{\mathcal{S}}\}.$
  \item $\jump{\mu x.\varphi(x)}_{\mathcal{S}} := \bigcap\{T \in \mathcal{P}(S) \mid
        \jump{\varphi(x)}_{\mathcal{S}[x \mapsto T]} \subseteq T\}.$
  \item $\jump{\nu x.\varphi(x)}_{\mathcal{S}} := \bigcup\{ T \in \mathcal{P}(S) \mid
        T \subseteq \jump{\varphi(x)}_{\mathcal{S}[x \mapsto T]}\}.$
\end{itemize}
In accordance with the usual terminology, we say that a formula
$\varphi$ is \textit{true} or \textit{satisfied} at
a possible world $s \in S$ (denoted $\mathcal{S}, s \models \varphi$) if
$s \in \jump{\varphi}_{\mathcal{S}}$. A formula $\varphi$ is \textit{valid} (denoted
$\models \varphi$) if $\varphi$ is true at every world in any model. 
\end{Definition}

\begin{Example}\normalfont
Let $\mathcal{S} = (S, R, \lambda)$ be a Kripke model. A formula $\varphi(x)$ such that
$x \in \mathsf{Free}(\varphi(x))$
can be naturally seen as the following function:
\[
\begin{xy}
  (0,8) *{\mathcal{P}(S)},(20,8)*{\mathcal{P}(S)}, 
  (0,4) *{\rotatebox[origin=c]{90}{$\in$}},(20,4) *{\rotatebox[origin=c]{90}{$\in$}},
  (0,0) *{T}="T",(26,0) *{\jump{\varphi(x)}_{\mathcal{S}[x \mapsto T]}.}="[T]",
	\ar (5,8);(15,8)
	\ar @{|->} (3,0);(15,0)
\end{xy}
\]
This function is \textit{monotone} if $x$ is positive in $\varphi(x)$. Thus, by the Knaster-Tarski Theorem
\cite{tarski1955}, $\jump{\mu x.\varphi(x)}_{\mathcal{S}}$ and
$\jump{\nu x.\varphi(x)}_{\mathcal{S}}$ are the least and greatest fixpoint of the function $\varphi(x)$,
respectively.

Under this characterization of fixpoint operators,
we find that
many interesting properties of the Kripke model can be
represented by modal $\mu$-formulas. For
example, consider the formula $\varphi_{1} = \mu x.(\Diamond x \vee p)$. For every Kripke model
$\mathcal{S}$ and its possible world $s$, we have
$\mathcal{S}, s \models \varphi_{1}$ if and only if
there is some possible world reachable from $s$ in which $p$ is true. Consider the formula
$\varphi_{2} = \nu y.\mu x.((\Diamond y \wedge p)\vee(\Diamond x \wedge \neg p))$.
Then $\mathcal{S}, s \models \varphi_{2}$ if and only if
there is some path from $s$ on which $p$ is true infinitely often.
\end{Example}

\subsection{Axiomatization}\label{subsec: axiomatization}
We give the Kozen's axiomatization $\mathsf{Koz}$ for the modal $\mu$-calculus in the Tait-style
calculus.\footnote{
In Kozen's original article \cite{DBLP:journals/tcs/Kozen83}, the system
$\mathsf{Koz}$
was defined as the axiomatization of the equational theory. Nevertheless we present
$\mathsf{Koz}$ as an equivalent Tait-style calculus due to the calculus'
affinity with the tableaux discussed in the sequel.
}
Hereafter, we will write $\Gamma$, $\Gamma'$, $\dots$ for a finite set of formulas. Moreover, the
standard abbreviation in the Tait-style calculus are used. That is, we write $\alpha, \Gamma$ for
$\{ \alpha \} \cup \Gamma$; $\Gamma, \Gamma'$ for $\Gamma \cup \Gamma'$; and $\sim\!\Gamma$ for
$\{ \sim\!\gamma \mid \gamma \in \Gamma \}$ and so forth.

\fbox{\bf Axioms}
$\mathsf{Koz}$ contains basic tautologies of classical propositional calculus and the
\textit{pre-fixpoint axioms}:
$$
\infer[(\textsf{Bot})]
    {\bot \vdash}
    {}
    \qquad
\infer[(\textsf{Tau})]
    {\varphi, \sim\!\varphi \vdash}
    {}
    \qquad
\infer[(\textsf{Prefix})]
    {\alpha(\mu x.\alpha(x)), \sim\!\mu x.\alpha(x) \vdash}
    {}
$$

\fbox{\bf Inference Rules}
In addition to the classical inference rules from propositional modal logic, for any formula $\varphi(x)$
such that $x$ appears only positively, we have the \textit{induction rule} $(\mathsf{Ind})$ to handle
fixpoints: 
$$
    \infer[(\vee)]{\alpha\vee\beta, \Gamma \vdash}
    {\alpha, \Gamma \vdash \quad \beta, \Gamma \vdash} \qquad
    \infer[(\wedge)]{\alpha\wedge\beta, \Gamma \vdash}
    {\alpha, \beta, \Gamma \vdash} \qquad
$$
$$
    \infer[(\mathsf{Weak})]{\alpha, \Gamma \vdash}
    {\Gamma \vdash} \qquad
    \infer[(\Diamond)]{\Diamond \psi, \Gamma \vdash}
    {\psi, \{ \alpha \mid \square \alpha \in \Gamma\} \vdash}
$$
$$
    \infer[(\mathsf{Cut})]{\Gamma, \Gamma' \vdash}
    {\Gamma, \sim\!\alpha \vdash \quad \alpha, \Gamma' \vdash} \qquad
    \infer[(\mathsf{Ind})]{\mu x.\varphi(x), \sim\!\psi \vdash}
    {\varphi(\psi), \sim\!\psi \vdash}
$$
Of course, the condition of substitution is satisfied in the $(\mathsf{Ind})$-rule; namely,
no new binding relation occurs by applying the substitution $\varphi(\psi)$.
As usual, we say that a formula $\sim\!\bigwedge\Gamma$ is \textit{provable} in $\mathsf{Koz}$
(denoted $\Gamma \vdash$) if there exists a proof diagram of $\Gamma$. We frequently use notation
such as
$\Gamma \vdash \Gamma'$ to mean $\Gamma, \sim\!\Gamma' \vdash$.

The following two lemmas state basic properties of $\mathsf{Koz}$. We leave the proofs of these
statements as an exercise to the reader.

\begin{Lemma}\label{lem: basic properties of KOZ 01}
Let $\varphi$ be a modal $\mu$-formula and let $\alpha(x)$ and $\beta(x, x)$ be modal
$\mu$-formulas where $x$ appears only positively. Then, the following holds:

\begin{enumerate}
\item
  $\vdash \eta x.\alpha(x) \leftrightarrow \eta y.\alpha(y)$ where $y \notin \mathsf{Free}(\alpha(x))$.

\item
  $\vdash \eta x.\beta(x, x) \leftrightarrow \eta x.\eta y.\beta(x, y)$ where
  $y \notin \mathsf{Free}(\beta(x, x))$.

\item
  $\vdash \mu x. \alpha(x) \leftrightarrow \alpha(\bot)$, if no appearances of $x$ are
  in the scope of any modal operators.

\item
  $\vdash \nu x. \alpha(x) \leftrightarrow \alpha(\top)$, if no appearances of $x$ are
  in the scope of any modal operators.

\item
  We can construct a well-named formula $\mathsf{wnf}(\varphi) \in \mathsf{WNF}$ such that
  $\vdash \varphi \leftrightarrow \mathsf{wnf}(\varphi)$.
\end{enumerate}
\end{Lemma}

\begin{Lemma}\label{lem: basic properties of KOZ 02}
Let $\alpha$, $\beta$, $\varphi(x)$, $\psi(x)$, $\chi_{1}(x)$ and $\chi_{2}(x)$ be modal $\mu$-formulas
where $x$ appears only positively in $\varphi(x)$ and $\psi(x)$. Further,
suppose that
$\chi_{1}(\alpha)$, $\chi_{1}(\beta)$ and $\chi_{2}(\alpha)$ are legal substitution;
namely, a new binding relation does not occur by such substitutions. Then, the
following holds:

\begin{enumerate}

\item
  If $\vdash \varphi(x) \rightarrow \psi(x)$ then
  $\vdash \eta x.\varphi(x) \rightarrow \eta x.\psi(x)$.

\item
  If $\vdash \alpha \leftrightarrow \beta$ then $\vdash \chi_{1}(\alpha) \leftrightarrow \chi_{1}(\beta)$.

\item
  If $\vdash \chi_{1}(x) \leftrightarrow \chi_{2}(x)$ then
  $\vdash \chi_{1}(\alpha) \leftrightarrow \chi_{2}(\alpha)$.

\end{enumerate}
\end{Lemma}
The following lemma is essentially used when proving claim (f).

\begin{Lemma}\label{lem: basic properties of KOZ 03}
Let $\varphi$ and $\alpha(x)$ be modal
$\mu$-formulas with $x$ appearing only positively in $\alpha(x)$ and $x \notin \mathsf{Free}(\varphi)$. Then we have that
if $\alpha\big(\mu x.(\varphi \wedge \alpha(x))\big) \vdash \varphi$ then
$\mu x.\alpha(x) \vdash \varphi$.
\end{Lemma}

\begin{proof}
Suppose that
\[
    \alpha\big(\mu x.(\varphi \wedge \alpha(x))\big) \vdash \varphi
\]
By propositional principal we have
\begin{equation}\label{eq: extension of rule 1}
    \alpha\big(\mu x.(\varphi \wedge \alpha(x))\big) \vdash
    \varphi \wedge \alpha\big(\mu x.(\varphi \wedge \alpha(x))\big) 
\end{equation}
On the other hand, by $(\mathsf{Prefix})$ rule we have
\begin{equation}\label{eq: extension of rule 2}
    \varphi \wedge \alpha\big(\mu x.(\varphi \wedge \alpha(x))\big) \vdash
    \mu x.(\varphi \wedge \alpha(x))
\end{equation}
By combining Statements $(\ref{eq: extension of rule 1})$ and $(\ref{eq: extension of rule 2})$ we get
\begin{equation}\label{eq: extension of rule 3}
    \alpha\big(\mu x.(\varphi \wedge \alpha(x))\big) \vdash
    \mu x.(\varphi \wedge \alpha(x))
\end{equation}
Therefore by applying $(\mathsf{Ind})$ to $(\ref{eq: extension of rule 3})$ we have
\begin{equation}\label{eq: extension of rule 4}
    \mu x.\alpha(x) \vdash
    \mu x.(\varphi \wedge \alpha(x))
\end{equation}
The following statement $(\ref{eq: extension of rule 5})$ is easily provable in $\mathsf{Koz}$:
\begin{equation}\label{eq: extension of rule 5}
    \mu x.(\varphi \wedge \alpha(x)) \vdash
    \varphi
\end{equation}
Finally apply $(\mathsf{Cut})$ rule to Statement $(\ref{eq: extension of rule 4})$ and
$(\ref{eq: extension of rule 5})$, then we have
\[
    \mu x.\alpha(x) \vdash \varphi
\]
This completes the proof.
\end{proof}

\begin{Remark}\label{rem: remember}\normalfont
Let $\varphi(x)$ be a formula where $x$ appears only positively in  $\varphi(x)$.
By Lemma \ref{lem: basic properties of KOZ 02} and \ref{lem: basic properties of KOZ 03},
we can assume that $\mathsf{Koz}$ can simulate the following two inference rules:
$$
    \infer[(\mathsf{Record})]{\mu x.\alpha(x), \Gamma \vdash}
    {\alpha\big(\mu x.(\sim\!\bigwedge\Gamma\wedge\alpha(x))\big), \Gamma \vdash} \qquad
    \infer[(\mathsf{Forget})] {\chi\big(\mu x.(\sim\!\bigwedge\Gamma\wedge\alpha(x)\big), \Gamma' \vdash}
    {\chi\big(\mu x.\alpha(x)\big), \Gamma' \vdash}
$$
where $x \notin \mathsf{Free}(\Gamma)$.
In the following, we will discuss the $\mathsf{Koz}$ as having the above inference rules from the beginning.
\end{Remark}

\subsection{Tableau}\label{subsec: tableau}

\begin{Definition}[\bf Cover modality]\normalfont
Let $\Phi$ be a finite set of formulas. Then $\triangledown \Phi$ denotes an abbreviation of the
following formula:
\[
    \big(\bigwedge\Diamond \Phi\big) \wedge \big(\square \bigvee\Phi\big).
\]
Here, $\Diamond \Phi$ denotes the set $\{ \Diamond \varphi \mid \varphi \in \Phi\}$, and as always, we
use the convention that $\bigvee\emptyset := \bot$ and $\bigwedge\emptyset := \top$. The symbol
$\triangledown$ is called the \textit{cover modality}.
\end{Definition}

\begin{Remark}\normalfont
Note that both the ordinary diamond $\Diamond$ and the ordinary box $\square$ can be expressed in
term of cover modality and the disjunction:
\begin{eqnarray*}
  \Diamond \varphi \equiv \triangledown\{\varphi, \top\},\\
  \square \varphi \equiv \triangledown\emptyset \vee \triangledown\{\varphi\}.
\end{eqnarray*}
Therefore, without loss of generality we restrict ourselves to using only $\triangledown$ instead of
$\Diamond$ and $\square$. Hereafter, we exclusively use cover modality notation instead of
ordinary modal notation; thus \textit{if not otherwise mentioned, all formulas are assumed to be using
this new constructor}.
Moreover, syntactic concepts such as
the well-named formula and the alternation depth extend to formulas using this modality.
\end{Remark}

\begin{Definition}\normalfont
Let $\Gamma$ be a set of formulas. We will say that $\Gamma$ is
\textit{locally consistent} if $\Gamma$ does not contain $\bot$ nor any propositional variable $p$ and
its negation $\neg p$ simultaneously. On the other hand, $\Gamma$ is said to be \textit{modal}
(under $\varphi$) if $\Gamma$ does not contain formulas of the forms $\alpha\vee\beta$,
$\alpha\wedge\beta$, $\eta x.\alpha(x)$, or $x \in \mathsf{Bound}(\varphi)$. In
other words, if $\Gamma$ is modal, then $\Gamma$ can possess only literals and formulas of the form
$\triangledown\Phi$.
\end{Definition}

\begin{Definition}[\bf Tableau]\label{def: tableau}\normalfont
Let $\varphi$ be a well-named formula. A set of \textit{tableau rules} for $\varphi$ is defined as
follows:
$$
  \infer[(\vee)]{\alpha\vee\beta, \Gamma}
  {\alpha, \Gamma \;\mid\; \beta, \Gamma} \qquad
  \infer[(\wedge)]{\alpha\wedge\beta, \Gamma}
  {\alpha, \beta, \Gamma}
$$
$$
  \infer[(\eta)]{\eta_{x}x.\varphi_{x}(x), \Gamma}
  {\varphi_{x}(x), \Gamma} \qquad
  \infer[(\mathsf{Regeneration})]{x, \Gamma}
  {\varphi_{x}(x), \Gamma}
$$
$$
  \infer[(\triangledown)]
  {\triangledown\Psi_{1}, \dots, \triangledown\Psi_{i}, l_{1}, \dots, l_{j}}
  {\{\psi_{k}\}\cup\{\bigvee\Psi_{n}\mid n \in N_{\psi_{k}}\}\:\mid\: \text{For every $k \in \omega$ with
  $1 \leq k \leq i$ and $\psi_{k} \in \Psi_{k}$.}}
$$
where in the $(\triangledown)$-rule,
$l_{1}, \dots, l_{j} \in \mathsf{Lit}(\varphi)$ and
$N_{\psi_{k}} := \{ n \in \omega \mid 1 \leq n \leq i,\; n \neq k\}$. Therefore, the premises of
a $(\triangledown)$-rule is equal to $\sum_{1 \leq k\leq i}|\Psi_{k}|$. 

A \textit{tableau} for $\varphi$ is a structure $\mathcal{T}_{\varphi} = (T, C, r, L)$ where $(T, C, r)$
is a tree structure and $L: T \rightarrow \mathcal{P}(\mathsf{Sub}(\varphi))$ is a
\textit{label function} satisfying the following clauses:
\begin{enumerate}
\item
  $L(r) = \{\varphi\}$.
\item
  Let $t \in T$. If $L(t)$ is modal and inconsistent then $t$ has no child. Otherwise, if $t$ is labeled
  by a set of formulas which fulfills the form of the conclusion of some tableau rules, then $t$ has
  children which are labeled by the sets of
  formulas of premises of one of those tableau rules, e.g., if $L(t) = \{ \alpha \vee \beta\}$, then $t$
  must have two children $u$ and $v$ with $L(u) = \{\alpha\}$ and $L(v) = \{\beta\}$. 
\item
  The rule $(\triangledown)$ can be applied in $t$ only if $L(t)$ is modal.
\end{enumerate}
We call a node $t$ a $(\triangledown)$-node if the rule $(\triangledown)$ is applied between $t$ and its
children. The notions of $(\vee)$-, $(\wedge)$-,
$(\eta)$- and $(\mathsf{Regeneration})$-node are
defined similarly.
\end{Definition}

\begin{Definition}[\bf Modal and choice nodes]\label{def: modal and choice nodes}\normalfont
Leaves and $(\triangledown)$-nodes are called \textit{modal nodes}. The root of the tableau and
children of modal nodes are called \textit{choice nodes}. We say that a modal node $t$ and
choice node $u$ are \textit{near} to each other if $t$ is a descendant of $u$ and between the $C$-sequence
from $u$ to $t$, there is no node in which the rule $(\triangledown)$ is applied. Similarly, we say
that a modal node $t'$ is a \textit{next modal node} of a modal node $t$ if $t'$ is a descendant of $t$
and between the $C$-sequence from $t$ to $t'$, rule $(\triangledown)$ is applied exactly once
between $t$ and
its child.
\end{Definition}

\begin{Definition}[\bf Trace]\label{def: trace}\normalfont
Let $\Gamma$ and $\Gamma'$ be finite sets of formulas. We define the \textit{trace function}
$\mathsf{TR}_{\Gamma, \Gamma'}: \Gamma\rightarrow\mathcal{P}(\Gamma')$ as follows:
\begin{itemize}
\item
If $\Gamma$ and $\Gamma'$ can be the lower label and
one of the upper label of a tableau inference rule respectively,
then $\mathsf{TR}_{\Gamma, \Gamma'}$ is a function which
outputs set of formulas of the result of reduction of $\gamma$ where $\gamma \in \Gamma$ as input.
For instance, if $\Gamma = \{ \alpha\wedge\beta, \gamma\}$ and $\Gamma' = \{ \alpha, \beta, \gamma\}$,
then these are the labels of the following $(\wedge)$-rule: 
$$
  \infer[(\wedge)]
  {\alpha\wedge\beta, \gamma}
  {\alpha, \beta, \gamma}
$$
Hence, we have $\mathsf{TR}_{\Gamma, \Gamma'}(\alpha\wedge\beta) := \{ \alpha, \beta\}$
and $\mathsf{TR}_{\Gamma, \Gamma'}(\gamma) := \{ \gamma \}$.

\item
Otherwise, we set $\mathsf{TR}_{\Gamma, \Gamma'}(\gamma) := \emptyset$ for every $\gamma \in \Gamma$.
\end{itemize}
Take a finite or infinite sequence $\vec{\Gamma} = \Gamma_{1}\Gamma_{2}\dots$
of finite sets of formulas.
A \textit{trace} $\mathsf{tr}$ on $\vec{\Gamma}$ is a finite or infinite sequence of formulas
satisfying the following two conditions;
\begin{itemize}
\item
  $\mathsf{tr}[1] \in \Gamma_{1}$.
\item
  For any $n \in \omega\setminus\{0\}$, if $\mathsf{tr}[n]$ is defined and satisfies  
  $\mathsf{TR}_{\Gamma_{n}, \Gamma_{n+1}}(\mathsf{tr}[n]) \neq \emptyset$,
then $\mathsf{tr}[n+1]$ is also defined
and satisfies $\mathsf{tr}[n+1] \in \mathsf{TR}_{\Gamma_{n}, \Gamma_{n+1}}(\mathsf{tr}[n])$.
\end{itemize}
The infinite trace $\mathsf{tr}$ is said to be \textit{even} if
\[
    \max \mathsf{Infinite}(\vec{\Omega}_{\varphi}(\mathsf{tr})) = 0 \pmod 2.
\]
$\vec{\Gamma}$ is said to be even if there exists an even trace on $\vec{\Gamma}$.

Let $\mathcal{T}_{\varphi} = (T, C, r, L)$ be a tableau for a well-named formula $\varphi$.
Let $\xi$ be an infinite branch of $\mathcal{T}_{\varphi}$. Then we say $\xi$ is even
if $\vec{L}(\xi)$ is even.
\end{Definition}

\begin{Definition}[\bf $\mu$-trace]\label{def: mu-trace}\normalfont
Let $\mathsf{tr}$ be an infinite trace. Then we call $\mathsf{tr}$ \textit{$\mu$-trace} if
the smallest variable (with respect to dependency order $\preceq_{\varphi}$) regenerated infinitely often
is $\mu$-variable. Similarly, we call a trace $\mathsf{tr}$ a \textit{$\nu$-trace} if
the smallest variable regenerated infinitely often is $\nu$-variable.
Note that every infinite trace $\mathsf{tr}$ is either a $\mu$-trace or $\nu$-trace since
all the rules except $(\mathsf{regeneration})$-rule decrease the size of formulas
and formulas are eventually reduced since every bound variable is in the scope of some modal operator.

\textit{Based on the above definition, the fact that $\vec{\Gamma}$ is even can be rephrased that
$\vec{\Gamma}$ contains a $\mu$-trace.}
\end{Definition}

\subsection{Refutation}\label{subsec: refutatiuon}

\begin{Definition}[\bf Refutation]\label{def: refutation}\normalfont
A well-named formula $\varphi$ is given. \textit{Refutation rules} for $\varphi$ are defined as the
rules of tableau, but this time, we modify the set of rules by adding an explicit weakening rule:
$$
  \infer[(\mathsf{Weak})]{\alpha, \Gamma}
  {\Gamma}
$$
and, instead of the $(\triangledown)$-rule, we take the following $(\triangledown_{r})$-rule:
$$
\infer[(\triangledown_{r})]
  {\triangledown\Psi_{1}, \dots, \triangledown\Psi_{i}, l_{1}, \dots, l_{j}}
  {\{\psi_{k}\}\cup\{\bigvee\Psi_{n}\mid n \in N_{\psi_{k}}\}}
$$
where in the $(\vee_{r})$-rule, we have $1 \leq k \leq i$, $\psi_{k} \in \Psi_{k}$,
$N_{\psi_{k}} = \{ n \in \omega \mid 1 \leq n \leq i, \; n \neq k \}$ and
$l_{1}, \dots, l_{j} \in \mathsf{Lit}(\varphi)$.
Therefore the $(\triangledown_{r})$-rule has one premise.

A \textit{refutation} for $\varphi$ is a structure
$\mathcal{R}_{\varphi} = (T, C, r, L)$ where $(T, C, r)$ is a tree structure and
$L: T \rightarrow \mathcal{P}(\mathsf{Sub}(\varphi))$ is a \textit{label function} satisfying the
following clauses:
\begin{enumerate}
  \item $L(r) = \{\varphi\}$.
  \item Every leaf is labeled by some inconsistent set of formulas.
  \item Let $t \in T$. If $L(t)$ is modal and inconsistent, then $t$ has no child. Otherwise, if $t$ is
        labeled by the set of formulas which fulfils the form of the conclusion of some refutation
        rules, then $t$ has children which are labeled by the sets of formulas of premises of those
        refutation rules. 
  \item The rule $(\triangledown_{r})$ can be applied to $t$ only if $L(t)$ is modal.
  \item For any infinite branch $\xi$, $\xi$ is even in the sense of Definition
        \ref{def: trace}. In other words, $\vec{L}(\xi)$ contains some $\mu$-traces.
\end{enumerate}
\end{Definition}

The following theorem is proved by
Niwinski and Walukiewicz \cite{Niwinski199699}
 (see also the literature \cite{DBLP:conf/dagstuhl/2001automata}).

\begin{Theorem}\label{the: refutation}
Let $\varphi$ be a well-named formula. If $\varphi$ is not satisfiable, then there exists a
refutation for $\varphi$.
\end{Theorem}

\begin{Remark}\label{rem: thin}\normalfont
The refutation is very similar to the proof diagram.
Indeed, it is easy to see that among the rules of refutation, $(\vee)$, $(\wedge)$,
and $(\mathsf{Weak})$ are the same as the rules of $\mathsf{Koz}$, and $(\triangledown_{r})$
corresponds to $(\Diamond)$.
The remarkable difference is that the proof diagram is a finite tree, whereas the refutation may contain infinite branches.
When trying to convert a refutation to a proof diagram,
the whole problem lies in ``cutting" these infinite branches.

The condition that the infinite branch contains a $\mu$-trace and inference rule
$(\mathsf{Record})$ are the keys to solving this cutting problem.
Consider an infinite branch $\xi$ of refutation, as shown on the left side of Figure \ref{fig: conversion to proof diagram}.
  \begin{figure}[htbp]
    \centering
    \includegraphics[width=14cm]{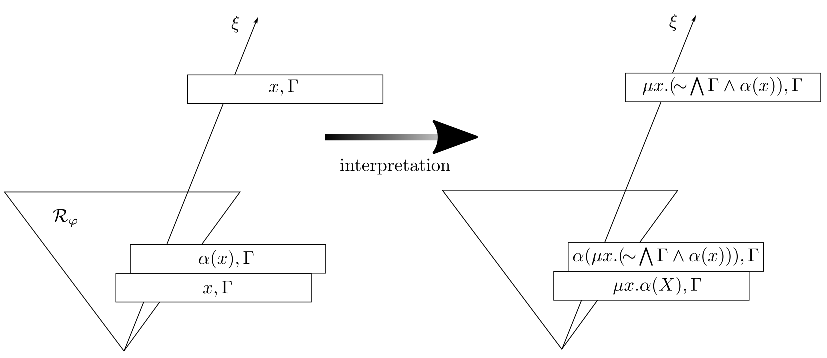}
    \caption{The conversion to proof diagram.}
    \label{fig: conversion to proof diagram}
  \end{figure}
Since $\xi$ contains a $\mu$-trace, there exists $\mu$-variables $x$ that will be regenerated
infinitely often.
The right side of Figure \ref{fig: conversion to proof diagram}
is constructed so that the interpretation of $x$ matches the inference rule
$(\mathsf{record})$; that is, it is interpreted as $x = \mu x.(\sim\!\bigwedge\Gamma\wedge\alpha(x))$.
Since $\{ \mu x.(\sim\!\bigwedge\Gamma\wedge\alpha(x))\} \cup \Gamma$ is provable in $\mathsf{Koz}$,
we can obtain the proof diagram we seek.
However, in reality, the branch of refutation is branched,
so inference rules $(\mathsf{Record})$ and $(\mathsf{Forget})$ must be applied very carefully
so that the above argument holds on \textit{all branches}.
The strategies for applying these inference rules when refutation satisfies a special \textit{thin} condition will be
discussed in detail in the proof of Theorem \ref{the: thin refutation}.
\end{Remark}

\section{Automata}\label{sec: automata}
The purpose of this section is to define
the terminology of $\omega$-automata theory and
to prepare the necessary tools to prove the completeness of $\mathsf{Koz}$.
Specifically, we will introduce two important concepts, the \textit{Safra's construction} \cite{safra1988}
and the \textit{index appearance record} defined by K{\v r}et\`ins\`y et al. \cite{Kret2017}.

\subsection{$\omega$-automata}\label{subsec: omega-automata}
$\omega$-automata are finite automata that are interpreted over infinite words and recognise
$\omega$-regular languages $\mathcal{L} \subseteq \Sigma^{\omega}$.
There are several variations of $\omega$-automata, depending on their acceptance conditions.
Among them, we deal with \textit{B\"uchi automata}, \textit{Rabin automata}, and
\textit{parity automata}. Firstly, we define the B\"uchi automata.

\begin{Definition}[\bf B\"uchi automata]\normalfont
A \textit{B\"uchi automaton} is a quintuple $\mathcal{BA} = \langle Q, \Sigma, q_{0}, \Delta, F \rangle$ where:
\begin{itemize}
    \item $Q$ is a finite set of \textit{states} of the automaton, 
    \item $\Sigma$ is an \textit{alphabet}, 
    \item $q_{0} \in Q$ is a state called the \textit{initial state}, 
    \item $\Delta : Q \times \Sigma \rightarrow \mathcal{P}(Q)$ is a \textit{transition function}, and
    \item $F \subseteq Q$ is a set of \textit{final states}. 
\end{itemize}
Using the usual definitions, we say that $\mathcal{BA}$ is \textit{deterministic} if $|\Delta(q, a)| \leq 1$
for every $q \in Q$ and $a \in \Sigma$. Let $\mathcal{BA} = (Q, \Sigma, q_{0}, \Delta, F)$ be a
B\"uchi automaton. A \textit{run} of $\mathcal{BA}$ on an $\omega$-word
$\sigma \in \Sigma^{\omega}$ is an infinite sequence
$\rho \in Q^{\omega}$ of a state where $\rho[1] = q_{0}$ and
$\rho[n+1] \in \Delta(\rho[n], \sigma[n])$ for any $n \geq 1$. An $\omega$-word
$\sigma \in \Sigma^{\omega}$ is \textit{accepted} by $\mathcal{BA}$ if there is a run $\rho$ of
$\mathcal{BA}$ on $\sigma$ satisfying the following condition: 
\[ \mathsf{Infinite}(\rho) \cap F \neq \emptyset. \] 
The $\omega$-language of all $\omega$-words accepted by $\mathcal{BA}$ is denoted by
$\mathcal{L}(\mathcal{BA})$. 
\end{Definition}

\begin{Remark}\normalfont
Consider the problem of converting nondeterministic automaton $\mathcal{BA}$ to
its equivalent deterministic automaton $\mathcal{BA}'$, as in the case of finite automata theory.
Note that the usual powerset construction 
which convert of nondeterministic finite automata to deterministic finite automata
does not work for $\omega$-automata.
\begin{figure}[htbp]
  \centering
  \includegraphics[width=8.5cm]{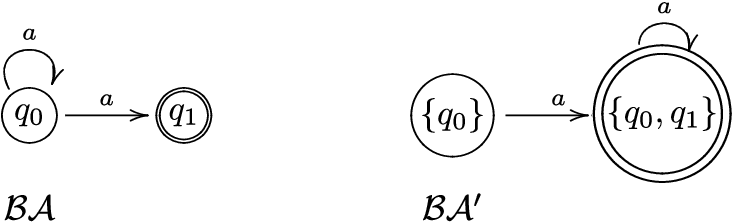}
  \caption{A counterexample of powerset construction.}
  \label{fig: powerset construction}
\end{figure}
For example, consider an
automaton $\mathcal{BA}$ and an automation $\mathcal{BA'}$ constructed by the powerset construction,
shown in Figure \ref{fig: powerset construction}. The
two automata are not equivalent, $\mathcal{L}(\mathcal{BA}) = \emptyset \neq \{ a^{\omega} \} = \mathcal{L}(\mathcal{BA'})$.
The problem is that the fact that the
final state $\{q_{0}, q_{1}\}$ of the powerset automaton $\mathcal{BA'}$ occurs infinitely often on a run
does not guarantee
that the automaton $\mathcal{BA}$ has a run on which its final state $q_{1}$ occurs infinitely often.
\end{Remark}

Secondly, we define the Rabin automata.

\begin{Definition}[\bf Rabin automata]\normalfont
A \textit{Rabin automaton} is a quintuple $\mathcal{RA} = \langle Q, \Sigma, q_{0}, \Delta, \{ (A_{j}, R_{j}) \mid j \in J \} \rangle$
where:
\begin{itemize}
  \item The definition of $Q$, $\Sigma$, $q_{0}$, and $\Delta$
           are exactly the same as for the B\"uchi automaton.
  \item $J$ is a finite set of subscripts (index), and for each $j \in J$, $A_{j}, R_{j} \subseteq Q$.
           $(A_{j}, R_{j})$ is called a Rabin's pair.
           Incidentally, $A_{j}$ is an acronym for ``Accept",
           $R_{j}$ is an acronym for ``Reject", respectively.
\end{itemize}
An $\omega$-word
$\sigma \in \Sigma^{\omega}$ is accepted by $\mathcal{RA}$ if there is a run $\rho$ of
$\mathcal{RA}$ on $\sigma$
and index $j \in J$ such that: 
\[ A_{j} \cap \mathsf{Infinite}(\rho) \neq \emptyset = R_{j} \cap \mathsf{Infinite}(\rho).\]
\end{Definition}

Finally, we define the parity automata.

\begin{Definition}[\bf Parity automata]\normalfont
A \textit{parity automaton} is a quintuple $\mathcal{PA} = \langle Q, \Sigma, q_{0}, \Delta, \mathsf{pri} \rangle$
where:
\begin{itemize}
  \item The definition of $Q$, $\Sigma$, $q_{0}$, and $\Delta$
           are exactly the same as for the B\"uchi automaton.
  \item  $\mathsf{pri}: Q \rightarrow \omega$ is called the \textit{priority function}.
\end{itemize}
An $\omega$-word
$\sigma \in \Sigma^{\omega}$ is accepted by $\mathcal{PA}$ if there is a run $\rho$ of
$\mathcal{PA}$ on $\sigma$ such that: 
\[ \max \{ \mathsf{pri}(q) \mid q \in \mathsf{Infinite}(\rho) \} \equiv 0 \mod 2.\]
\end{Definition}

The class of $\omega$-language characterized by the deterministic B\"uchi automata
is denoted by $\mathsf {DBA}$.
The class of $\omega$-language characterized by nondeterministic B\"uchi automata
is denoted by $\mathsf {NDBA}$.
Similarly, $\mathsf{DRA}$, $\mathsf{NDRA}$, $\mathsf{DPA}$, and $\mathsf{NDPA}$
are classes of  $\omega$-language characterized by
deterministic Rabin automata,
nondeterministic Rabin automata,
deterministic parity automata, and
nondeterministic parity automata, respectively.
Then, it is widely known that the following inclusion holds (see, e.g. the Literature \cite{DBLP:conf/dagstuhl/2001automata}):
\[
  \mathsf{DBA} \subsetneq \mathsf{NDBA} = \mathsf{DRA} = \mathsf{NDRA}
  = \mathsf{DPA} = \mathsf{NDPA}
\]
In the next subsection,
we will prove $\mathsf {NDBA} \subseteq \mathsf{DRA}$ which is the part of the above theorem,
by using a method called \textit{Safra's construction}.

\subsection{Safra's construction}\label{subsec: safra's construction}
In this subsection, the non deterministic B\"uchi automaton $\mathcal{BA} = \langle Q, \Sigma, q_{0}, \Delta, F \rangle$ is
fixed and discussed.
The purpose of this subsection is to construct explicitly a deterministic Rabin automaton equivalent to $\mathcal{BA}$.
For a given $n \in \omega$, we use $\Pi^{n}$ to denote the set of all permutations of $N := \{1, \dots, n\}$,
i.e., the set of all bijective function  $\pi: N \rightarrow N$. We identify $\pi$ with its canonical representation
as a vector $(\pi(1), \dots, \pi(n))$.
In the following, we will often say ``the position of $j \in N$ in $\pi$" or similar to refer to $\pi^{-1}(i)$.

\begin{Definition}[\bf Safra tree]\label{def: Safra tree}\normalfont
The structure $\mathsf{s} = \langle J, C, 1, l \rangle$ is called \textit{Safra tree} for $\mathcal{BA}$
when:
\begin{enumerate}
  \item $J \subseteq \{ 1, 2, \dots, (|Q|+1)^{2} \}$ \footnote{The number of pool of vertex indices is $(|Q|+1)^{2}$,
  where the reason why the upper limit is $(|Q|+1)^{2}$, will be described later in the Remark \ref {rem: node}.}
  is the set of vertices\footnote{The tree structures mentioned in this article are tableaux and Safra trees.
  We use the term vertex for a vertex of Safra tree, and use the term node for a vertex of tableau.
  We use these terms strictly to prevent confusion.},
  where $Q$ is the set of states of $\mathcal{BA}$.
  \item $C$ is a childhood relation over $J$.\footnote{In the definition of the Safra tree, 
           it is common to specify priorities (so-called older-younger relationships) between siblings.
           however, in this article, the older-younger relationship is defined by a \textit{index appearance record},
           so it is not defined here.}
  \item $1$ is a root of Safra tree.
  \item $l: J \rightarrow \mathcal{P}(Q)$ is a labeling function satisfying the following conditions:
  \begin{enumerate}
    \item For any $j \in J$, $l(j) \neq \emptyset$.
    \item For any $j \in J$, $l(j) \supseteq \bigcup_{k \in C(j)}l(k)$. In particular, if $j$ is not the root
    (i.e. $j \neq 1$), then $l(j) \supsetneq \bigcup_{k \in C(j)}l(k)$.\footnote{Again, it is more general to assume that
    $l(1) \supsetneq \bigcup_{k \in C(1)}l(k)$.
    In this article, however, it is intentionally allowed to be $l(1) = \bigcup_{k \in C(1)}l(k)$
    so that it is convenient to prove the completeness of the modal $\mu$-calculation later.}
    \item For any $j_{1}, j_{2} \in J$, if $j_{1}$ and $j_{2}$ are siblings, then $l(j_{1}) \cap l(j_{2}) = \emptyset$.
  \end{enumerate}
\end{enumerate}
\end{Definition}

\begin{Remark}\label{rem: safra 1}\normalfont
Let $\mathsf{s} = \langle J, C, r, l \rangle$ be a Safra tree.
Consider the assignment $Q \rightarrow J\setminus\{1\}$
which assign $j \in J$ for given $q \in Q$, where $q \in l(j)$ and
$q \notin \bigcup_{k \in C(j)}l(k)$.
According to condition (a) and (b) in part $4$ of Definition \ref{def: Safra tree},
it can be said that the assignment is surjective, and thus $|Q| \geq |J|-1$ holds.
In other words, the number of vertices in the Safra tree is at most $|Q|+1$.
\end{Remark}

\begin{Definition}[\bf Index appearance record \cite{Kret2017}]\normalfont
A duplex $\langle \pi, \mathsf {col} \rangle$
is an \textit{index appearance record} for $\mathcal{BA}$ if:
\begin{enumerate}
  \item $\pi \in \Pi^{(|Q|+1)^{2}}$ is a permutation; that is, $\pi = (\pi(1), \dots, \pi((|Q|+1)^{2}))$.
  \item $\mathsf{col}: \{ 1, \dots, (|Q|+1)^{2} \} \rightarrow \{ \textrm{green, red, white, black} \}$
  is a colouring function. For a vertex
  $j \in \{ 1, \dots, (|Q|+1)^{2} \}$, we say that ``$j$ is colored red" or similar if $\mathsf{col}(j) = \textrm{red}$;
  and the same applies to other colors.
\end{enumerate}
Take an index appearance record $\langle \pi, \mathsf{col} \rangle$.
For any $n, m \in \{ 1, \dots, (|Q|+1)^{2} \}$, we say that $n$ is older than $m$
if $n$ is to the right of $m$
(i.e., $\pi^{-1}(m) < \pi^{-1}(n)$).

Let $\mathsf{s} = \langle J, C, r, l \rangle$ be a Safra tree for $\mathcal{BA}$,
and $\langle \pi, \mathsf{col} \rangle$ be an index appearance record for $\mathcal{BA}$.
We say that $j_{1}$ is $j_{2}$'s older brother if
$j_{1}, j_{2} \in J$ are siblings, and $j_{1}$ is older than $j_{2}$ in $\pi$.
\end{Definition}

From now on, we will construct a deterministic Rabin automaton
$\mathcal{RA} = \langle Q', \Sigma, q'_{0}, \Delta', \{ (A_{j}, R_{j}) \mid j \in J \}\rangle$
equivalent to the nondeterministic B\"uchi automaton $\mathcal{BA}$.

\begin{enumerate}
  \item Let $\mathsf{S}$ be the set of all Safra trees for $\mathcal{BA}$, and let $\mathsf{IAR}$
           be the set of index appearance record for $\mathcal{BA}$.
           Then, set $Q' := \mathsf{S} \times \mathsf{IAR}$.
  \item Set $q'_{0} := \langle \{ 1 \}, \emptyset, l_{0}, \pi_{0}, \mathsf{col}_{0} \rangle$; where
           $l_{0}(1) := \{ q_{0} \}$,
           $\pi_{0} := ((|Q|+1)^{2}, \dots, 3, 2, 1)$,
           moreover, we set $\mathsf{col}_{0}(1) = \textrm{white}$ and $\mathsf{col}_{0}(k) = \textrm{black}$ $(k \geq 2)$.
  \item The transition function $\Delta'$ will be described later.
  \item $J := \{ 1, \dots, (|Q|+1)^{2} \}$; that is, the set of indices is the pool of vertex indices of the Safra tree.
           For any $j \in J$, $A_{j}$ is the set of states in which $j$ is shining green, and
           $R_{j}$ is the set of states in which $j$ is shining red.
\end{enumerate}

\begin{Remark}[\bf Intuitive meaning of the index appearance record]\normalfont
Before defining the transition function $\Delta'$, let us explain the intuitive meaning of the index appearance record.
Suppose $\mathcal{RA}$ is in state $\langle \mathsf{s}, \pi, \mathsf{col} \rangle$,
reads an alphabet, and transitions to state $\langle \mathsf{s'}, \pi', \mathsf{col'} \rangle$.
In this situation, $\mathsf{s'}$ is generated by adding new vertices to $\mathsf{s}$ and
removing unnecessary vertices.
The coloring function $\mathsf{col'}$ records the usage of each vertex $j$ in the transition from
$\mathsf{s}$ to $\mathsf{s'}$, and intuitively has the meanings shown in the table \ref{table: meaning of colors}.
\begin{table}[htb]
\caption{Intuitive meaning of colors.}
\label{table: meaning of colors}
  \begin{tabular}{|c|l|} \hline
    Coloring & Intuitive meaning  \\ \hline \hline
    $\mathsf{col'}(j) = \textrm{white}$ & $j$ is used for the vertex of $\mathsf{s'}$.  \\ \hline
    $\mathsf{col'}(j) = \textrm{green}$ & $j$ is used for the vertex of $\mathsf{s'}$; moreover,
      it is related to an acceptance condition.\\ \hline
    $\mathsf{col'}(j) = \textrm{black}$ & $j$ is not used for the vertex of $\mathsf{s'}$,
      and is waiting for reuse in the pool.\\ \hline
    $\mathsf{col'}(j) = \textrm{red}$ & $j$ is not used for the vertex of $\mathsf{s'}$,
      and was deleted in the latest transition. \\ \hline
  \end{tabular}
\end{table}
$\pi'$ represent not only the most recent transition, but also
the seniority-based relationships; that is, the farther to the right, the longer it has been used
as the vertex of the Safra tree. In particular, since the root $1$ is always the oldest,
position of $1$ is always on the far right side of $\pi$.
\end{Remark}
Now let's define the transition $\Delta'$ of $\mathcal{RA}$.
Suppose $\mathcal {RA}$ is in the state $\langle J, C, 1, l, \pi, \mathsf{col} \rangle$ and the alphabet
$a$ is read. The next state is generated in the following 7 steps:
\begin{steps}
  \item {\bf Initialize the index appearance record.}
  Let $\pi^{(1)}$ be a permutation obtained from $\pi$ by moving all indices shining in red to the front.
  For instance, suppose that $\mathsf{col}(\pi[3]) = \mathsf{col}(\pi[5]) = \textrm{red}$, then $\pi^{(1)}$
  is a permutation shown in Figure \ref{fig: initialize the index appearance record};

  \begin{figure}[htbp]
    \centering
    \[
      \xymatrix@R=50pt{
        \pi: & \pi[1] \ar[drr] & \pi[2] \ar[drr] & \text{\mask{$\pi[3]$}{A}}_{\textrm{red}} \ar[dll] 
        & \pi[4] \ar[dr] & \text{\mask{$\pi[5]$}{A}}_{\textrm{red}} \ar[dlll] & \pi[6] \ar[d] & \pi[7] \ar[d]\\
        \pi^{(1)}: & \pi[3]             & \pi[5]            & \pi[1]            & \pi[2]            
        & \pi[4]            & \pi[6]           & \pi[7]
      }
    \]
    \caption{An example of initialize the index appearance record.}
    \label{fig: initialize the index appearance record}
  \end{figure}
  Also, the coloring function $\mathsf{col^{(1)}}$ is defined as follows:
  \[
    \mathsf{col^{(1)}}(j) := \begin{cases}
      \textrm{white} & (j \in J) \\
      \textrm{black} & (j \notin J)
    \end{cases}
  \]
  Then, update the state of $\mathcal{RA}$ to $\langle J, C, 1, l, \pi^{(1)}, \mathsf{col^{(1)}}\rangle$.

  \item {\bf Update vertex labels.}
  Each vertex $j \in J$ is labeled with $l(j) \subseteq Q$.
  The new labeling $l^{(1)}: J \rightarrow \mathcal{P}(Q)$ is defined below:
  \[
    l^{(1)}(j) := \Delta(l(j), a)
  \]
  That is, the label of each vertex is updated according to the transition function
  $\Delta$ of the original B\"{u}chi automaton $\mathcal{BA}$.
  Then, update the state of $\mathcal{RA}$ to $\langle J, C, 1, l^{(1)}, \pi^{(1)}, \mathsf{col^{(1)}} \rangle$.

  \item {\bf Add new children.}
  For each $j \in J$ and $q \in F$, add a new child $k$ of $j$ if $q \in l^{(1)} (j)$ where
  $k$ is the rightmost black-colored vertex in $\langle \pi^{(1)}, \mathsf{col^{(1)}}\rangle$.
  In this way, we extend $J$ to $J^{(1)}$ and $C$ to $C^{(1)}$.
  A label of new child $k$ is $\{q \}$, which extends $l^{(1)}$ to $l^{(2)}$.
  Also, let $\mathsf{col}^{(2)}$ be the coloring function that changed the color of the newly added $k$ from black to white.
  Then, update the state of $\mathcal{RA}$ to $\langle J^{(1)}, C^{(1)}, 1, l^{(2)}, \pi^{(1)}, \mathsf{col^{(2)}} \rangle$.

  \item {\bf Horizontal pruning.}
  We obtain a labeling $l^{(3)}$ from $l^{(2)}$ by removing, for every vertex $j \in J^{(1)}$
  with label $l^{(2)}(j)$ and all states $q \in l^{(2)}(j)$, $q$ from the labels of all younger siblings of $j$
  and all of their descendants.
  Then, update the state of $\mathcal{RA}$ to $\langle J^{(1)}, C^{(1)}, 1, l^{(3)}, \pi^{(1)}, \mathsf{col^{(2)}} \rangle$.
  Figure \ref{fig: an example of the horizontal pruning} is an example of the horizontal pruning.
  \begin{figure}[htbp]
    \centering
    \includegraphics[width=11.5cm]{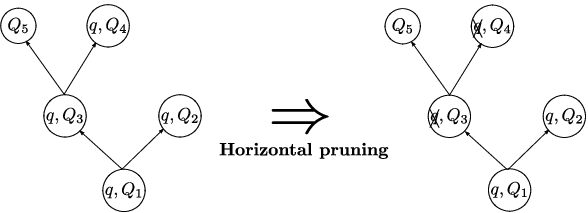}
    \caption{An example of horizontal pruning.}
    \label{fig: an example of the horizontal pruning}
  \end{figure}

  \item {\bf Vertical pruning.}
  For every $j \in J^{(1)} \setminus\{1 \}$, if $l^{(3)}(j) = \bigcup_ {k \in C(j)}l^{(3)}(k)$, then remove all descendants of $j$
  from the Safra tree.
  In this way, $J^{(1)}$ is reduced to $J^{(2)}$. Similarly, $C^{(1)}$ is reduced to $C^{(2)}$, $l^{(3)}$ is reduced to $l^{(4)}$.
  In addition, we changes the color of $j$ from white to green and the color of the deleted descendant $k$ from
  white to red. This coloring function is denoted by $\mathsf{col}^{(3)}$.
  Then, update the state of $\mathcal{RA}$ to $\langle J^{(2)}, C^{(2)}, 1, l^{(4)}, \pi^{(1)}, \mathsf{col^{(3)}} \rangle$.

  \item {\bf Removing vertices with empty label.}
  For each $j \in J^{(2)}$, if $l^{(3)}(j) = \emptyset$, then remove $j$ from the Safra tree.
  Thus, $J^ {(2)}$ is reduced to $J^{(3)}$.
  Similarly, $C^{(2)}$ is reduced to $C^{(3)}$, $l^{(4)}$ is reduced to $l^{(5)}$.
  In addition, we change the color of the deleted vertex $j$ from white to red.
  Let this coloring function be $\mathsf{col}^{(4)}$.
  Then, update the state of $\mathcal{RA}$ to $\langle J^{(3)}, C^{(3)}, 1, l^{(5)}, \pi^{(1)}, \mathsf{col}^{(4)} \rangle$.

  \item {\bf We have a new state.}
  The state obtained through the above operation is the next states of $\mathcal{RA}$.
  In other words, $\Delta'$ is defined as,
  \[  
    \Delta'(\langle J, C, 1, l, \pi, \mathsf{col} \rangle, a) := 
    \langle J^{(3)}, C^{(3)}, 1, l^{(5)}, \pi^{(1)}, \mathsf{col}^{(4)} \rangle.
  \]
\end{steps}

\begin{Remark}\label{rem: node}\normalfont
From the definition of $\Delta'$, it becomes clear why the number of elements of the pool for the vertex
of the Safra tree is $(|Q|+1)^{2}$.
As is mentioned in Remark \ref{rem: safra 1}, the number of vertices in the Safra tree is at most $|Q|+1$.
{\bf Add new children} may add a new vertex $k$ for every $j \in J$ and $q \in F$, that is,
up to $(|Q|+1)\cdot|Q| (\geq (|Q|+1)\cdot|F|)$ vertices $k$ may be added.
That is, to implement {\bf add new children}, temporarily at the maximum,
$|Q|+1 + (|Q|+1)\cdot|Q| = (|Q|+1)^{2}$ vertices need to be prepared.
\end{Remark}

We prove that $\mathcal{RA}$ is equivalent to $\mathcal{BA}$ by the following two lemmas.

\begin{Lemma}\label{lem: safra01}
For any $\omega$-word $\sigma \in \Sigma^{\omega}$,
if $\sigma \in \mathcal{L}(\mathcal{BA})$ then $\sigma \in \mathcal{L}(\mathcal{RA})$.
\end{Lemma}

\begin{proof}
Suppose that $\omega$-word $\sigma \in \Sigma^{\omega}$ is accepted by $\mathcal {BA}$; therefore
there exists a run $\rho$ of $\mathcal{BA}$ on $\sigma$, and $\rho$ satisfies B\"{u}chi's acceptance condition.
That is, a certain final state $q \in F$ exists such that $q \in \mathsf{Infinite}(\rho)$.
Let $\xi$ be the run of $\mathcal{RA}$ on $\sigma$.
For each  $n \geq 1$, set $\xi[n] := \langle J_{n}, C_{n}, 1, l_{n}, \pi_{n}, \mathsf{col_{n}} \rangle$.
In this situation, the definition of Safra's construction shows that for any $n \geq 1$, it becomes $\rho[n] \in l_{n}(1)$.
For a $n \geq 2$ and a vertex  $j \in \{ 1, \dots, (|Q|+1)^{2} \}$,
when $\rho[n-1] \in l_{n-1}(j)$ and $\rho[n] \notin l_{n}(j)$,
we say ``$\rho$ disappears from vertex $j$ on $n$-th transition".
Note that from the definition of the transition function $\Delta'$,
we can assume that $\rho$ disappears from the vertex $j$ in the $n$-th transition only if
either of the following two holds;
\begin{description}
  \item[(Case 1)] In the $n$-th run, $\rho[n]$ moved (joined) to $j$'s older brother by {\bf horizontal pruning}.
  \item[(Case 2)] In the $n$-th run, $j$ itself was deleted by {\bf vertical pruning} (an ancestor of $j$ glowed green).
\end{description}
Note that a certain moment $n_{0}$ exists
and that $\rho$ always appears in one of the children of route $1$ in the $n$-th runs with $\forall n \geq n_{0}$;
because $1$ does not glow green, thus (case 1) cannot occur,
and a final states $q$ appears infinitely often in $\rho$.
From that moment, $\rho$ can only finitely many times move to older brother by {\bf horizontal pruning}.
Let $j_{1} \in C_{n_{0}}(1)$ be the destination where $\rho$ finally moved by {\bf horizontal pruning}.
If $j_{1}$ lights green infinitelly often, then we are done.
Otherwise, a moment $n_{1}$ exists
and that $\rho$ always appears in one of the children of $j_{1}$ in the $n$-th runs with $\forall n \geq n_{1}$.
From that moment, $\rho$ can only finitely many times move to older brother by {\bf horizontal pruning}.
Let $j_{2} \in C_{n_{0}}(j_{1})$ be the destination where $\rho$ finally moved by {\bf horizontal pruning}.
If $j_{2}$ lights green infinitelly often, then we are done.
Otherwise, we repeat the reasoning and find a son of $j_{1}$, $j_{2}$ and so on. Observe that we cannot go this
way forever because Safra tree has a bounded size.
Therefore, there exists a vertex $j_{i}$ ($i \geq 1$) such that $j_{i}$ is deleted only a finite times
and shines green infinitely often.
\end{proof}

\begin{Lemma}\label{lem: safra02}
For any $\omega$-word $\sigma \in \Sigma^{\omega}$,
if $\sigma \in \mathcal{L}(\mathcal{RA})$ then $\sigma \in \mathcal{L}(\mathcal{BA})$.
\end{Lemma}

\begin{proof}
Let $\sigma \in \Sigma^{\omega}$.
Let $\xi$ be the run of $\mathcal{RA}$ on $\sigma$, and
set $\xi[n] = \langle J_{n}, C_{n}, 1, l_{n}, \pi_{n}, \mathsf{col}_{n} \rangle$.
First, suppose that $1 < N < M$ and the vertex $j \in \{ 1, \dots, (|Q|+1)^{2} \}$ satisfy the following two conditions.
\begin{enumerate}
  \item $j$ is always used as a vertex between $\xi[N]$ and $\xi[M]$.
   In other words, $N \leq \forall k \leq M$, $\mathsf{col}_{k}(j)\neq \mathrm{red}$.
  \item $\mathsf{col}_{N}(j) = \mathsf{col}_{M}(j) = \mathrm{green}$, and $N < \forall k < M$, $\mathsf{col}_{k}(j) \neq \mathrm{green}$.
\end{enumerate}
From the definition of $\mathcal{RA}$, for any $q_{M} \in l_{M}(j)$, there exist $q_{N} \in l_{N}(j)$ and a sequence
$q_{N}q_{N+1}\dots q_{M}$ such that:
\[
  N \leq \forall k \leq M-1, \; q_{k+1} \in \Delta(q_{k}, \sigma[k]).
\]
Such a sequence $q_{N}q_{N+1}\dots q_{M}$ is called
``$\Delta$-path from $q_ {N}$ to $q_ {M}$". 
Then, the following claim holds.
\begin{description}
  \item[\bf$(\dag)$]
  For any $q_{N}\in l_{N}(j)$ and $q_{M}\in l_{M}(j)$, if there is a $\Delta$-path from $q_{N}$ to $q_{M}$,
  then at least one of them (let's denote this $q_{N}q_{N+1}\dots q_{M}$) intersects $F$.
  That is, there exists a $L$ with $N \leq L \leq M$ such that $q_{L} \in F$.
\end{description}
The above claim is proved by induction, but we leave it as an exercise for the reader.

Now suppose $\sigma$ is accepted by $\mathcal{RA}$.
Thus, there is a run $\xi$ of $\mathcal{RA}$ on $\sigma$, and $\xi$ satisfies Rabin's acceptance condition;
that is, there is a $j \in \{1, \dots, (|Q|+1)^{2}\}$ which meets the following conditions:
\begin{itemize}
  \item We can take $N_ {0}>1$, and $j$ is always used in the transitions after the $N_ {0}$-th transition.
  \item There exist $N_{0} < N_{1} < \dots < N_{k} < \cdots$ with $(k \in \omega)$ such that
  $\mathsf{col}_{N_{k}}(j) = \mathsf{col}_{N_{k+1}}(j) = \mathrm{green}$, and
  $N_{i} < \forall k < N_{i+1}$, $\mathsf{col}_{k}(j) \neq \mathrm{green}$.
\end{itemize}
Let $\rho$ be a run of $\mathcal{BA}$ on $\sigma$
such that $\forall n > N_{0}$, $\rho[n] \in l_{n}(j)$.
In general, $\rho$ does not always satisfy the B\"{u}chi's acceptance condition,
but from Claim $(\dag)$, the following holds:
\begin{description}
  \item[\bf ($\ddag$)]
  There is a $\Delta$-path from $\rho[N_{i}]$ to $\rho[N_{i+1}]$ which intersects $F$
  (let's denote this $q_{N_{i}}q_{N_{i}+1}\dots q_{N_{i+1}}$).
\end{description}
Let $\rho'$ be the $\Delta$-sequence in which a part  $\rho[N_{i}]$ to $\rho[N_{i+1}]$ of sequence $\rho$
is replaced with another sequence $q_{N_{i}}q_{N_{i}+1}\dots q_{N_{i+1}}$ for every $i \geq 1$.
Then, from Claim $(\ddag)$, $\rho'$ satisfy the B\"{u}chi's acceptance condition, therefore
$\sigma \in \mathcal{L}(\mathcal{BA})$.
\end{proof} 

From Lemma \ref{lem: safra01} and Lemma \ref{lem: safra02}, we obtain the following theorem.

\begin{Theorem}[\bf Safra's construction \cite{safra1988}]\label{the: safra}
For any nondeterministic B\"{u}chi automaton $\mathcal{BA}$,
the deterministic Rabin automaton $\mathcal{RA}$ generated by Safra's construction
satisfies $\mathcal{L}(\mathcal{BA}) = \mathcal{L}(\mathcal{RA})$.
\end{Theorem}

\subsection{Conversion to parity automaton}\label{subsec: conversion to parity automaton}
Suppose that an arbitrarily nondeterministic B\"uchi automaton $\mathcal{BA}$ is given.
In this subsection,
we will construct an equivalent deterministic parity automaton $\mathcal{PA}$.
In fact, most of the content to be
discussed is completed in the previous subsection \ref{subsec: safra's construction}.

Let $\langle \pi, \mathsf{col} \rangle$ be an index appearance record for $\mathcal{BA}$.
Set
\[
  \mathsf{maxind}(\pi, \mathsf{col}) :=
  \max \{ \pi^{-1}(j) \mid \mathsf{col}(j) \in \{ \mathrm{green}, \mathrm{red} \} \} \cup \{ 0 \}.
\]
In other words, among the vertices colored in either green or red in $\pi$,
the position of the vertex on the far right of these is $\mathsf{maxind}(\pi, \mathsf{col})$.
From this, we will concretely build the deterministic Parity automaton
$\mathcal{PA} = \langle Q', \Sigma, q'_{0}, \Delta', \mathsf{pri} \rangle$ as follows:
\begin{enumerate}
  \item
  $Q'$, $q'_ {0}$, and $\Delta'$ are exactly the same as
  the Rabin automaton defined in Safra's construction.
  \item The priority function is defined as follows:
  \[
    \mathsf{pri}\big(\langle \mathsf{s}, \pi, \mathsf{col} \rangle\big) := \begin{cases}
      1 & \big(\mathsf{maxind}(\pi, \mathsf{col}) = 0\big) \\
      2\cdot\mathsf{maxind}(\pi, \mathsf{col}) &
        \big(\mathsf{col}(\pi[\mathsf{maxind}(\pi, \mathsf{col})]) = \mathrm{green}\big)\\
      2\cdot\mathsf{maxind}(\pi, \mathsf{col}) + 1&
        \big(\mathsf{col}(\pi[\mathsf{maxind}(\pi, \mathsf{col})]) = \mathrm{red}\big)
    \end{cases}
  \]
\end{enumerate}

\begin{Theorem}[\bf Legitimacy of $\mathcal{PA}$ \cite{Kret2017}]
For any nondeterministic B\"{u}chi automaton $\mathcal{BA}$,
the deterministic parity automaton $\mathcal{PA}$ generated by construction shown above
satisfies $\mathcal{L}(\mathcal{BA}) = \mathcal{L}(\mathcal{PA})$.
\end{Theorem}

\begin{proof}
Let $\mathcal{RA}$ be the Rabin automaton constructed by Safra's construction.
From Theorem \ref{the: safra}, it is enough to show that $\mathcal{L}(\mathcal{RA}) = \mathcal{L}(\mathcal{PA})$.
Take an $\omega$-word $\sigma \in \Sigma^{\omega}$ arbitrarily.
Let $\xi$ be the run of $\mathcal{PA}$ on $\sigma$.
Note that $\xi$ is also a run of $\mathcal{RA}$.

First, note that the position of any vertex $j \in \{ 1, \dots, (|Q|+1)^{2} \}$ only changes in two
different ways:
\begin{itemize}
  \item $j$ itself is removed from the Safra tree and driven to the far left in {\bf Initialize the index appearance record}.
  In this case, we say that $j$ was \textit{demoted} in the transition.
  \item Some $k \in \{ 1, \dots, (|Q|+1)^{2} \}$ with a position older than $j$ has
  been removed (demoted), increasing the position of $j$.
  In this case, we say that $j$ was \textit{promoted} in the transition.
\end{itemize}
Set $\xi[n] = \langle J_{n}, C_{n}, 1, l_{n}, \pi_{n}, \mathsf{col_{n}} \rangle$ for $n \geq 1$.
Suppose that a vertex $j \in \{ 1, \dots, (|Q|+1)^{2} \}$ and a natural number $N \geq 1$
satisfy $\forall n \geq N$, $j \in J_{n}$.
In this situation, $j$ will not be demoted in the $N$-th and subsequent transitions,
and promotion can be done only finitely many times, so if a sufficiently large $M> N$ is taken,
then $j$ will not be demoted nor promoted in the $M$-th and subsequent transitions.
The position of $j$ when it is no longer demoted nor promoted $\pi^{-1}_{M}(j)$ is called the
\textit{stable position} of $ j $ in the run $\xi$ (notated as $\mathsf{stable}_{\xi}(j)$).

It is obvious from the definition of the priority function $\mathsf{pri}$ that
if $\xi$ satisfies the parity condition, then $\xi$ also satisfies the Rabin's acceptance condition.
On the contrary, if $\xi$ satisfies Rabin's acceptance condition,
then there exists some $j \in \{1, \dots, (|Q|+1)^{2}\}$ such that
\[ A_{j} \cap \mathsf{Infinite}(\xi) \neq \emptyset = R_{j} \cap \mathsf{Infinite}(\xi).\]
Let $k$ be the vertex with the largest stable position among such $j$s.
In the transition sufficiently far along the run, $k$ is in a stable position, colored green infinitely often, and
the elders of $k$ are not colored red ($\because$ if the elders of $k$ are removed, $k$ is promoted).
Therefore, we have 
\[
  \max \{ \mathsf{pri}(q) \mid q \in \mathsf{Infinite}(\xi) \} = 2\cdot\mathsf{stable}_{\xi}(k),
\]
that is, $\xi$ satisfy the parity condition.
Hence, $\mathcal{L}(\mathcal{RA}) = \mathcal{L}(\mathcal{PA})$ holds.
\end{proof}

\section{Application of automata to the modal $\mu$-calculous}\label{sec: application}
In this section, we apply the results of Section \ref{sec: automata} to the modal $\mu$-calculous to prove two important
results. First, in Subsection \ref{subsec: automata that determines}, we give an automaton
that determines the parity of the tableau branch.
In the following Subsection \ref{subsec: completeness for thin}, this automaton is used to prove
completeness of $\mathsf{Koz}$ for the thin refutation; which
is Claim (f) mentioned in Section \ref{sec: introduction}.
In the last Subsection \ref{subsec: automaton normal form}, the proof of the existence of the automaton normal form
(Claim (b)) is proved along with the concrete construction method.

\subsection{Automata that determines the parity of tableau branches}\label{subsec: automata that determines}

\begin{Definition}[\bf Activeness]\label{def: activeness}\normalfont
Let $\varphi$ be a well-named formula, and $\preceq_{\varphi}$ be its dependency order
(recall Definition \ref{def: alternation depth}). Then,
For any $\psi \in \mathsf{Sub}(\varphi)$ and $x \in \mathsf{Bound}(\varphi)$, we say $x$ is
\textit{active} in $\psi$ if there exists $y \in \mathsf{Sub}(\psi)\cap\mathsf{Bound}(\varphi)$
such that $x \preceq_{\varphi} y$.
\end{Definition}
Suppose that a well-named formula $\varphi$ is arbitrarily given.
From now on, we will construct a nondeterministic B\"uchi automaton
$\mathcal{BA}_{\varphi} = \langle Q, \Sigma, q_{0}, \Delta, F\rangle$
which determine the parity of the tableau branch for $\varphi$.
The letter handled by the automaton $\mathcal{BA}_{\varphi}$ is
subset $\Gamma \subseteq \mathsf{Sub}(\varphi)$, therefore $\Sigma := \mathcal{P}(\mathsf{Sub}(\varphi))$.
The state $q \in Q$ is of the form $q = (\Gamma, \gamma) \in \mathcal{P}(\mathsf{Sub}(\varphi))\times\mathsf{Sub}(\varphi)$
or $q = (\Gamma, \gamma, x) \in \mathcal{P}(\mathsf{Sub}(\varphi))\times\mathsf{Sub}(\varphi)\times\mathsf{Bound}(\varphi)$
which satisfies the following three conditions:
\begin{enumerate}
  \item $\gamma \in \Gamma$
  \item $x$ is active in $\gamma$.
  \item $x$ is $\mu$-variable in $\varphi$. That is, $\mu x.\varphi_{x}(x) \in \mathsf{Sub}(\varphi)$.
\end{enumerate}
The initial state is $(\{\varphi\}, \varphi)$. The transition function $\Delta$ is defined as follows:
\begin{align*}
 \Delta\big((\Gamma, \gamma), \Gamma' \big)&:= \{ (\Gamma', \gamma'), (\Gamma', \gamma', x) \mid
  \gamma' \in \mathsf{TR}_{\Gamma, \Gamma'}(\gamma), \; x \in \mathsf{Bound}(\varphi)\},\\
 \Delta\big((\Gamma, \gamma, x), \Gamma' \big)&:= \{ (\Gamma', \gamma', x) \mid
  \gamma' \in \mathsf{TR}_{\Gamma, \Gamma'}(\gamma)\}.
\end{align*}
Finally, the final state is defined as $F := \{ (\Gamma, x, x)  \in Q \mid x \in \mathsf{Bound}(\varphi) \}$.
$\mathcal{BA}_{\varphi}$ embodies a naive way to seek $\mu$-trace non-deterministically.
Indeed, let
\[
  \xi = (\Gamma_{1}, \gamma_{1})\dots(\Gamma_{k}, \gamma_{k})(\Gamma_{k+1}, \gamma_{k+1}, x)(\Gamma_{k+2}, \gamma_{k+2}, x)\dots
\]
be a run of $\mathcal{BA}_{\varphi}$ on $\vec{L}(\rho)$ where $\rho$ is an infinite branch of a tableau $\mathcal{T}_{\varphi}$. 
Then, from the definition of $\mathcal{BA}_{\varphi}$, it can be seen that 
$\vec{L}(\rho) = \Gamma_{1}\Gamma_{2}\dots$ and that $\gamma_{1}\gamma_{2}\dots$ is a trace on $\vec{L}(\rho)$.
In short, a run picked a specific trace $\gamma_{1}\gamma_{2}\dots$ from multiple traces on $\rho$.
The intuitive meaning of transitioning from $(\Gamma_{k}, \gamma_{k})$ to $(\Gamma_{k+1}, \gamma_{k+1}, x)$ is,
\begin{itemize}
  \item[$(\star)$]
  $x$ is a variable such that the value of $\Omega_{\varphi}(x)$ is maximized in the $(k+1)$-th and subsequent transitions;
  and that regenerated infinitely often.
\end{itemize}
Indeed, for any $y \in \mathsf{Bound}(\varphi)$, if $\Omega_{\varphi}(y) > \Omega_{\varphi}(x)$, then
$(\Gamma, y, x)$ cannot be a states of automaton because $y$ is not active in $x$.
Therefore, $y$, which has a higher priority than $x$, does not appear in the traces $\gamma_{k+1}\gamma_{k+2}\dots$.
In addition, if $\rho$ is accepted, a states in the form of $(\Gamma, x, x)$ must appear infinitely often in $\rho$.
This means that $x$ will be regenerated infinitely often in the trace $\gamma_{1}\gamma_{2}\dots$.
Therefore, Claim $(\star)$ agrees that the automaton accepts $\xi$.
From the above, $\mathcal{BA}_{\varphi}$ certainly determines the parity of tableau branches.

Let $\mathcal{PA}_{\varphi}$ be the parity automaton which is converted from
$\mathcal{BA}_{\varphi}$ by the method introduced in Subsection \ref{subsec: conversion to parity automaton}.
Then $\mathcal{PA}_{\varphi}$ becomes a deterministic parity automaton that determines the parity of tableau branches.
Hereinafter, we denote $N_{\varphi} := (|Q|+1)^{2}$; where $Q$ is the state set of $\mathcal{BA}_{\varphi}$.

\begin{Remark}\label{rem: duplication}\normalfont
Let $\mathcal{PA}_{\varphi} = \langle Q, \Sigma, q_{0}, \Delta, \mathsf{pri} \rangle$ be the parity automaton
given above.
Let $\mathcal{T}_{\varphi} = (T, C, r, L)$ be a tableau for $\varphi$.
For any node $t \in T$, set
\[
\Delta(\vec{L}(\vec{rt}), q_{0}) = \langle J_{t}, C_{t}, 1, l_{t}, \pi_{t}, \mathsf{col}_{t} \rangle.
\]
Then, if $(\Gamma, \gamma) \in l_{t}(1)$ or $(\Gamma, \gamma, x) \in l_{t}(1)$, then $\Gamma = L(t)$ holds.
Moreover, $\{ \gamma \mid (\Gamma, \gamma) \in l_{t}(1) \} = L(t)$ holds.
What this means is that there is duplication of information in the first and second quadrants
of label elements of the Safra tree .
With this in mind, we can omit the first quadrant, hence each vertex of the Safra tree is labeld by
$\mathsf{Sub}(\varphi)\cup \big(\mathsf{Sub}(\varphi)\times\mathsf{Bound}(\varphi)\big)$.
In this article, we will think so in the following.
In other words, the label of the Safra tree is considered to be in the shape of
\[ \{ \gamma_{1}, \dots, \gamma_{j} \} \cup \{ (\gamma'_{1}, x_{1}), \dots, (\gamma'_{k}, x_{k})\} \quad (j, k \geq 0). \]
\end{Remark}

\subsection{Completeness for thin refutation}\label{subsec: completeness for thin}
In this subsection, the completeness of $\mathsf{Koz}$ will be proven when $\varphi$ has thin refutation.
The parity automaton $\mathcal{PA}_{\varphi}$ created in Subsection \ref{subsec: automata that determines}
is used for the proof.
First, we will define the concept of the thin refutation.

\begin{Definition}[\bf Thin refutation]\normalfont\label{def: thin refutation}
Let $\mathcal{R}_{\varphi}$ be a refutation for some well-named formula $\varphi$. We say that
$\mathcal{R}_{\varphi}$ is \textit{thin} if, whenever a formula of the form $\alpha\wedge\beta$ is
reduced, some node of the refutation and some variable is active in $\alpha$ as well as $\beta$, then
at least one of $\alpha$ and $\beta$ is immediately discarded by using the $(\mathsf{Weak})$-rule.
\end{Definition}

\begin{Remark}\normalfont
Let $\mathcal{R}_{\varphi} = (T, C, r, L)$ be a thin refutation.
Let $\mathcal{PA}_{\varphi}$ be a parity automaton that determines the parity of tableau branches.
For any node $t \in T$, set $\Delta(\vec{L}(\vec{rt}), q_{0}) = \langle J_{t}, C_{t}, 1, l_{t}, \pi_{t}, \mathsf{col}_{t}\rangle$;
then it has the following distinctive characteristics:
\begin{itemize}
  \item
  For non-root vertices $k \in \{2, 3, \dots, N_{\varphi}\}$, $k$ is labeled with elements in the form of $(\gamma, x)$.
  \item
  For any $k \in \{2, 3, \dots, N_{\varphi}\}$, $l_{t}(k)$ consists of at most two elements.
  \item
  For any $k \in \{2, 3, \dots, N_{\varphi}\}$, if $l_{t}(k)$ consists of two elements, then $t$ is a
  $(\mathsf{Weak})$-node, and one of them is discarded by $(\mathsf{Weak})$-rule
  in the transition between $t$ and its child.
\end{itemize}
In short, if $\mathcal{PA}_{\varphi}$ load $\vec{L}(\xi)$ where $\xi$ is a branch of thin refutation,
each vertex of the Safra tree (except the root) will be labeled with a single element in almost all cases.
\end{Remark}

\begin{Definition}[\bf Definition list]\normalfont
Let $\varphi$ be a well-named formula.
The sequence $(x_{1}, x_{2}, \dots x_{N})$ is a linear ordering of all bound variables of $\varphi$
which is compatible with dependency order, i.e., if $x_{i} \preceq_{\varphi} x_{j}$ then $i \leq j$.
A definition list $\mathcal{D}_{\varphi}$ is a following format:
\[
  \mathcal{D}_{\varphi} := \big( x_{1} = \eta_{1}x_{1}.\alpha_{1}(x_{1}),\dots, x_{N} = \eta_{N}x_{N}.\alpha_{N}(x_{N}) \big);
\]
where for any $k \leq N$, 
$\eta_{k}x_{k}.\alpha_{k}(x_{k}) = \eta_{x_{k}}x_{k}.\varphi_{x_{k}}(x_{k})$, or
$\eta_{k}x_{k}.\alpha_{k}(x_{k}) = \eta_{x_{k}}x_{k}.(\beta\wedge\varphi_{x_{k}}(x_{k}))$ holds with
$\beta$ is arbitrary formula such that $x_{k} \notin \mathsf{Free}(\beta)$.
Especially, we call
\[
  \big( x_{1} = \eta_{1}x_{1}.\varphi_{x_{1}}(x_{1}),\dots, x_{N} = \eta_{N}x_{N}.\varphi_{x_{N}}(x_{N}) \big)
\]
a plain definition list.
For any $\beta \in \mathsf{Sub}(\varphi)$, we define a expansion
$\expa{\beta}_{\mathcal{D}_{\varphi}}$ of 
$\beta$ by $\mathcal{D}_{\varphi}$ as follows:
\[
  \langle\![ \beta ]\!\rangle_{\mathcal{D}_{\varphi}} :=
  \beta[x_{N}/\eta_{N}x_{N}.\alpha_{N}(x_{N})]\dots[x_{1}/\eta_{1}x_{1}.\alpha_{1}(x_{1})].
\]
In addition, for $\Gamma \subseteq \mathsf{Sub}(\varphi)$, we set $\expa{\Gamma}_{\mathcal{D}_{\varphi}} :=
\{ \expa{\gamma}_{\mathcal{D}_{\varphi}} \mid \gamma \in \Gamma \}$.
In the following, if the $\varphi$ being discussed is clear from the context, it may be represented by
$\mathcal{D}$, omitting the subscript of $\mathcal{D}_{\varphi}$.
\end{Definition}

\begin{Theorem}[\bf Completeness for formulas in which the thin refutation exists]\label{the: thin refutation}
Let $\varphi$ be a well-named formula. If there exists a thin refutation for $\varphi$, then
$\sim\!\varphi$ is provable in $\mathsf{Koz}$.
\end{Theorem}

\begin{proof}
Let $\mathcal{R}_{\varphi} = (T, C, r, L)$ be a thin refutation for $\varphi$.
Our goal is to show that there exists label $f: T \rightarrow \mathcal{P}(\mathsf{Form})$ satisfies
the following conditions:
\begin{enumerate}
  \item[(C1)]
  $f(r) = \{ \varphi \}$.

  \item[(C2)]
  For any $t \in T$ and its children $u_{1}, \dots, u_{k} \in C(t)$,
  $$
    \infer
    {f(t)}
    {f(u_{1}) \;\mid \dots \mid \; f(u_{k})}
  $$
  can be simulated with axiomatic system $\mathsf{Koz}$.
  That is, if $f(u_{i})\vdash_{\mathsf{Koz}}$ for every $i \;(1 \leq i \leq k)$, then $f(t)\vdash_{\mathsf{Koz}}$.

  \item[(C3)]
  For any leaf $t$ of $\mathcal{R}_{\varphi}$, $f(t)$ is inconsistent; thus $f(t) \vdash_{\mathsf{Koz}}$ holds.

  \item[(C4)]
  For any infinite branch $\xi$ of $\mathcal{R}_{\varphi}$,
  there exists a node $t$ on $\xi$ such that $f(t) = \{ \Gamma, \; \mu x.\big(\bigwedge\!\sim\!\Gamma\wedge\alpha(x)\big) \}$
  for some $\alpha(x), \beta \in \mathsf{Form}$.
  Thus $f(t) \vdash_{\mathsf{Koz}}$ holds.
\end{enumerate}
It is clear that the theorem holds if the above $f$ can be defined.
As is mentioned in Remark \ref{rem: thin}, in constructing $f$,
we must carefully apply inference rules $(\mathsf{Record})$ and $(\mathsf{Forget})$.
Figure \ref{fig: the idea of construction of function f} illustrates the idea of how to construct the function $f$.
\begin{figure}[htbp]
  \centering
  \includegraphics[width=15.5cm]{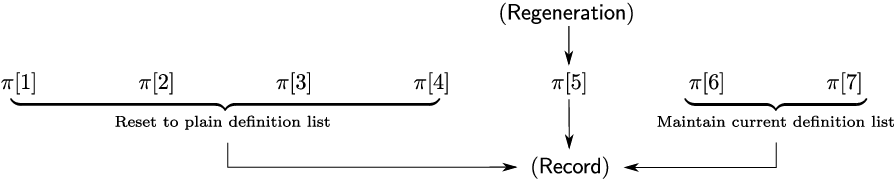}
  \caption{The idea of construction of function $f$.}
  \label{fig: the idea of construction of function f}
\end{figure}
Let $j$ be a vertex where a formula belonging to $j$ is reduced by $(\mathsf{Regeneration})$-rule.
The basic strategy is to apply $(\mathsf{Forget})$ at the vertex on the left side of  $j$,
and to apply $(\mathsf{Record})$ at $j$.

For any node $t \in T$, set
\[
  \Delta_{\varphi}(\vec{L}(\vec{rt}), q_{0}) = q_{t} = \langle J_{t}, C_{t}, 1, l_{t}, \pi_{t}, \mathsf{col}_{t} \rangle.
\]
From now, for each node $t \in T$ and vertex $j \in \{ 1, \dots, N_{\varphi} \}$,
the definition list $\mathcal{D}_{t, j}$ is inductively defined from the root of the tree toward the leaves.
First, set $\mathcal{D}_{r, j}$ $(j \geq 1)$ be a plain definition list.
Next, assuming that $\mathcal{D}_{t, j}$ $(t \in T, j \leq N_{\varphi})$ is already defined,
for each $t' \in C(t)$ and $j' \in \{ 1, \dots, N_{\varphi} \}$, $\mathcal{D}_{t', j'}$ is defined by case as follows:
\begin{description}
  \item[(Case 1)] $t$ is not a $(\mathsf{Regeneration})$-node:
  For the vertex $j'$ that is deleted in the transition from $q_{t}$ to $q_{t'}$, $\mathcal{D}_{t', j'}$ is a plain definition list.
  For other $j'$, set $\mathcal{D}_{t', j'} := \mathcal{D}_{t, j'}$;
  that is, it inherits the same definition list from the same vertex of the parent node.

  \item[(Case 2)] $t$ is a $(\mathsf{Regeneration})$-node: 
  Suppose the following inferences are made between $t$ and $t'$:
  $$
    \infer[(\mathsf{Regeneration})]{x, \Gamma}
    {\varphi_{x}(x), \Gamma}
  $$
  If $x$ is a $\nu$-variable, then $\mathcal{D}_{t', k'}$ is similar to (Case 1).
  If $x$ is a $\mu$-variable, and there does \textit{not} exist the vertex $k \in \{1, \dots, N_{\varphi}\}$
  such that $l_{t}(k) = \{ (x, x) \}$, then $\mathcal{D}_{t', k'}$ is also similar to (Case 1).
  If $x$ is a $\mu$-variable, and there exists the vertex $k \in \{1, \dots, N_{\varphi}\}$
  such that $l_{t}(k) = \{ (x, x) \}$, then
  \begin{itemize}
    \item
    Let $\mathcal{D}_{t', k'} := \mathcal{D}_{t, k'}$ for $k'$ on the right side of $k$ in $\pi$.
    That is, for $k'$ older than $k$, the same definition list is inherited from the same vertex of the parent node.

    \item
    For $k'$ on the left side of $k$ in $\pi$, $\mathcal{D}_{t', k'}$ is a plain definition list.

    \item
    The definition list $\mathcal{D}_{t', k}$ is obtained by replacing the definition of $x$ in definition list
    $\mathcal{D}_{t, k}$ with
    $x = \mu x. \big(\!\sim\!\expa{\Gamma} \wedge \varphi_{x}(x)\big)$; where
    $\expa{\Gamma} :=
    \{ \expa{\gamma}_{\mathcal{D}_{t', k'}} \mid \gamma \in l_{t'}(k'), \; k' \in J_{t'}\setminus\{k\}.\}$.
  \end{itemize}
\end{description}
Set $f(t) := \bigcup_{j \in J_{t}} \expa{l_{t}(j)}_{\mathcal{D}_{t, j}}$, then $f$ is the function we seek.
Indeed, it is obvious that $f$ satisfies (C1), (C2), and (C3).
For (C4), take an infinite branch $\xi$ arbitrarily and consider the run
$\vec{q} := q_{\xi[1]}q_{\xi[2]}\dots$.
Since $\vec{L}(\xi)$ contains $\mu$-trace, $\vec{q}$ satisfies the parity condition;
therefore we can find a vertex $k \in \{ 1, \dots, N_{\varphi}\}$ and
$M > N > 1$ such that
\begin{itemize}
  \item $\max \{ \mathsf{pri}(q) \mid q \in \mathsf{Infinite}(\vec{q}) \} = 2\cdot \mathsf{stable}_{\vec{q}}(k)\equiv 0 \mod 2$.
  \item $\pi^{-1}_{\xi[N]}(k) = \pi^{-1}_{\xi[M]}(k) = \mathsf{stable}_{\vec{q}}(k)$.
  \item $\xi[N]$ and $\xi[M]$ are both $(\mathsf{Regeneration})$-nodes and regenerated $x$ in them where
  vertex $k$ is labeled by $\{ (x, x) \}$ in $\xi[N]$ and $\xi[M]$.
  \item For any $L$ $(N < L < M)$ and $k$'s older brother $j$ in $\xi[L]$, the label for $j$ is not regenerated.
  In other words, the definition list of $k$'s older brother does not change between $\xi[N]$ and $\xi[M]$.
\end{itemize}
From the above conditions, $\mathcal{D}_{\xi[N], j} = \mathcal{D}_{\xi[M], j}$ for any $j \leq N_{\varphi} \setminus \{ k \}$ and thus
$f(\xi[M]) = \{ \expa{\Gamma}, \mu x.\big(\!\sim\!\bigwedge\expa{\Gamma}\wedge\varphi_{x}(x)\big) \}$
with $\expa{\Gamma} =
    \{ \expa{\gamma}_{\mathcal{D}_{\xi[M], k'}} \mid \gamma \in l_{\xi[M]}(k'), \; k' \in J_{\xi[M]}\setminus\{k\}.\}$.
Therefore, it certainly satisfies (C4).
\end{proof}

\subsection{Automaton normal form}\label{subsec: automaton normal form}

\begin{Definition}[\bf Automaton normal form]\label{def: automaton normal form}\normalfont
The set of an \textit{automaton normal form} $\mathsf{ANF}$ is the smallest set of formulas defined by the
following clauses:

\begin{enumerate}
\item
  If $l_{1}, \dots, l_{i} \in \mathsf{Lit}$, then $\bigwedge_{1 \leq j \leq i} l_{j} \in \mathsf{ANF}$.
\item
  If $\alpha \vee \beta \in \mathsf{ANF}$,
  $\mathsf{Bound}(\alpha)\cap\mathsf{Free}(\beta) = \emptyset$
  and $\mathsf{Free}(\alpha)\cap\mathsf{Bound}(\beta) = \emptyset$,
  then $\alpha\vee\beta \in \mathsf{ANF}$.
\item
  If $\alpha(x) \in \mathsf{ANF}$ where $x$ occurs only positively in the scope of some modal
  operator (cover modality), occurs at once, and $\mathsf{Sub}(\alpha(x))$ does not contain a formula of the form
  $x\wedge\beta$ where $\beta \neq \top$. Then, $\eta x.\alpha(x) \in \mathsf{ANF}$.
\item
  If $\Phi \subseteq \mathsf{ANF}$ is a finite set such that for any $\varphi_{1}, \varphi_{2} \in \Phi$,
  we have
  $\mathsf{Bound}(\varphi_{1})\cap\mathsf{Free}(\varphi_{2}) = \emptyset$,
  then
  $(\triangledown \Phi) \wedge (\bigwedge_{1 \leq i \leq j}l_{i}) \in \mathsf{ANF}$
  where $l_{1}, \dots, l_{j} \in \mathsf{Lit}\setminus
  \bigcup_{\varphi \in \Phi}\mathsf{Bound}(\varphi)$ with
  $0 \leq j$.
\item
  If $\alpha \in \mathsf{ANF}$ then $\alpha \wedge \top \in \mathsf{ANF}$.
\end{enumerate}
Note that the above clauses imply $\mathsf{ANF} \subseteq \mathsf{WNF}$.
\end{Definition}

\begin{Remark}\label{rem: shape of anf}\normalfont
For any automaton normal form $\widehat{\varphi}$, a tableau
$\mathcal{T}_{\widehat{\varphi}} = (T, C, r, L)$
for $\widehat{\varphi}$ forms
very simple shapes. Indeed, for any node $t \in T$,
there exists at most one formula
$\widehat{\alpha} \in L(t)$ which includes some bound variables.
Note that for any infinite trace $\mathsf{tr}$, $\mathsf{tr}[n]$ must include some bound variables.
Consequently, for any infinite branch of the tableau for an automaton
normal form, there exists a unique trace on it.
\end{Remark}

\begin{Definition}[\bf Tableau bisimulation]\label{def: tableau bisimulation}\normalfont
Let $\mathcal{T}_{\alpha} = (T, C, r, L)$ and $\mathcal{T}_{\beta} = (T', C', r', L')$ be two tableaux for
some well-named formulas $\alpha$ and $\beta$. Let $T_{m}$ and $T'_{m}$
be sets of modal nodes of $\mathcal{T}_{\alpha}$ and $\mathcal{T}_{\beta}$, respectively, and let
$T_{c}$ and $T'_{c}$ be a set of choice nodes of $\mathcal{T}_{\alpha}$ and $\mathcal{T}_{\beta}$, 
respectively. Then $\mathcal{T}_{\alpha}$ and $\mathcal{T}_{\beta}$ are said to be
\textit{tableau bisimilar}
(notation: $\mathcal{T}_{\alpha}\rightleftharpoons\mathcal{T}_{\beta}$)
if there exists a binary relation $Z\subseteq (T_{m}\times T'_{m})\cup(T_{c}\times T'_{c})$
satisfying the following seven conditions:
\begin{description}
\item[Root condition:]
$(r, r') \in Z$.

\item[Prop condition:]
For any $t \in T_{m}$ and $t' \in T'_{m}$, if $(t, t') \in Z$, then
\[
  L(t)\cap\mathsf{Lit}(\alpha) = L'(t')\cap\mathsf{Lit}(\beta).
\]
Consequently $L(t)$ is consistent if and only if $L'(t')$ is consistent.

\item[Forth condition on modal nodes:]
  Take $t \in T_{m}$, $u \in T_{c}$ and $t' \in T'_{m}$ arbitrarily.
  If $(t, t') \in Z$ and $u \in C(t)$, then there exists $u' \in C'(t')$ such that $(u, u') \in Z$
  (see Figure \ref{fig: forth condition}).
  \begin{figure}[htbp]
    \centering
    \includegraphics[width=12cm]{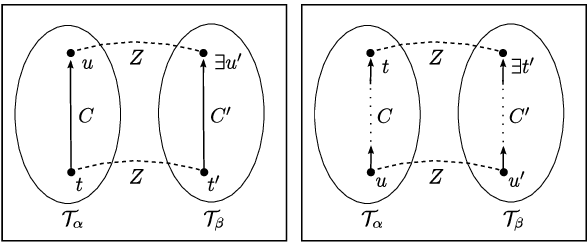}
    \caption{The forth conditions.}
    \label{fig: forth condition}
  \end{figure}

\item[Back condition on modal nodes:]
  The converse of the forth condition on modal nodes:
  Take $t \in T_{m}$, $t' \in T'_{m}$ and $u' \in T'_{c}$ arbitrarily.
  If $(t, t') \in Z$ and $u' \in C'(t')$, then there exists $u \in C(t)$ such that $(u, u') \in Z$.

\item[Forth condition on choice nodes:]
  Take $u \in T_{c}$, $t \in T_{m}$ and $u' \in T'_{c}$ arbitrarily.
  If $(u, u') \in Z$ and $t$ is near $u$, then there exists $t' \in T'_{m}$ such that $(t, t') \in Z$
  and $t'$ is near $u'$ (see Figure \ref{fig: forth condition}).

  \item[Back condition on choice nodes:]
  The converse of the forth condition on choice nodes:
  Take $u \in T_{c}$, $u' \in T'_{c}$ and $t' \in T'_{m}$ arbitrarily.
  If $(u, u') \in Z$ and $t'$ is near $u'$, then there exists $t \in T_{m}$ such that $(t, t') \in Z$
  and $t$ is near $u$.

  \item[Parity condition:]
  Let $\xi$ and $\xi'$ be infinite branches of $\mathcal{T}_{\alpha}$ and $\mathcal{T}_{\beta}$,
  respectively. We say that $\xi$ and $\xi'$ are \textit{associated} with each other if
  the $k$-th modal nodes $\xi[i_{k}]$ and $\xi'[i'_{k}]$ satisfy
  $(\xi[j_{k}], \xi'[j'_{k}]) \in Z$ for any $k \in \omega\setminus\{0\}$.
  For any $\xi$ and $\xi'$ which are associated with each other, we have $\xi$ is even if and only if
  $\xi'$ is even.
\end{description}
If $\mathcal{T}_{\alpha}$ and $\mathcal{T}_{\beta}$ are tableau bisimilar with $Z$, then $Z$ is
called a \textit{tableau bisimulation} from $\mathcal{T}_{\alpha}$ to $\mathcal{T}_{\beta}$.
\end{Definition}

We accept the following Lemma \ref{lem: basic properties of tableau bisimulation} without
proof (see, e.g., the literature \cite{DBLP:conf/dagstuhl/2001automata}).
\begin{Lemma}\label{lem: basic properties of tableau bisimulation}
Let $\alpha$, $\beta$ be well-named formulas.
If $\mathcal{T}_{\alpha}\rightleftharpoons\mathcal{T}_{\beta}$,
then
$\models \alpha \leftrightarrow \beta$.
\end{Lemma}

\begin{Theorem}[\bf Janin and Walukiewicz \cite{conf/mfcs/JaninW95}]\label{the: automaton normal form}
For any well-named formula $\alpha$, we can construct an automaton normal
form $\mathsf{anf}(\alpha)$ such that
$\mathcal{T}_{\alpha}\rightleftharpoons\mathcal{T}_{\mathsf{anf}(\alpha)}$
for some tableau $\mathcal{T}_{\alpha}$ for $\alpha$.\footnote{Note that the tableau of  $\mathsf{anf}(\alpha)$ is uniquely determined.}
\end{Theorem}

\begin{proof}
Let $\mathcal{T}'_{\alpha} = (T, C, r, L)$ be a tableau for a given formula $\alpha$, let
$\mathcal{PA}_{\alpha} = (Q, \Sigma, \Delta, q_{0}, \mathsf{pri})$ be a parity
automaton that is given in Subsection \ref{subsec: conversion to parity automaton}.
For any node $t \in T$, set
\[
  \Delta(\vec{L}(\vec{rt}), q_{0}) = q_{t} = \langle J_{t}, C_{t}, 1, l_{t}, \pi_{t}, \mathsf{col}_{t} \rangle.
\]
First, we construct a tableau-like structure
$\mathcal{TB}_{\alpha} =
  (T_{b}, C_{b}, r_{b}, L_{b}, B_{b})$
called a \textit{tableau with back edge} from $\mathcal{T}'_{\alpha}$ as follows:
\begin{itemize}
  \item The node $t \in T$ is called a \textit{loop node} if;
  \begin{description}
    \item[$(\spadesuit)$] There is a proper ancestor $t'$ such that
    $q_{t} = q_{t'}$, and
    \item[$(\heartsuit)$] for any $u \in T$ such that $u \in C^{\ast}(t')$ and $t \in C^{\ast}(u)$,
    we have $\mathsf{pri}(q_{u}) \leq \mathsf{pri}(q_{t}) (= \mathsf{pri}(q_{t'}))$. 
  \end{description}
  In this situation, the node $t'$ is called a \textit{return node} of $t$.
  Note that for any infinite branch $\xi$ of $\mathcal{T}_{\alpha}$, there exists
  a loop node on $\xi$ since $Q$ is finite. We define the set $T_{b}$ of nodes as follows:
  \[
    T_{b} := \{ t \in T \mid
    \text{for any proper ancestor $t'$ of $t$, $t'$ is \textit{not} a loop node}\}
  \]
  Intuitively speaking, we trace the nodes on each branch from the root
  and as soon as we arrive at a return node, we cut off the former branch from the tableau.

  \item Set $C_{b} := C|_{T_{b}\times T_{b}}$, $r_{b} := r$ and $L_{b} := L|_{T_{b}}$.

  \item $B_{b} := \{(t, t') \in T_{b}\times T_{b} \mid
  \text{$t$ is a loop node and $t'$ is a return node of $t$}\}$.
  An element of $B_{b}$ is called \textit{back edge}.
\end{itemize}
By K\"onig's lemma, we can assume that $\mathcal{TB}_{\alpha}$ is a finite structure because it has
no infinite branches.
The tableau with back edge is very
similar to the basic tableau. In fact, the unwinding $\mathsf{UNW}_{r_{b}}(\mathcal{TB}_{\alpha})$ is a
tableau for $\alpha$.
Therefore, we use the terminology and concepts of the tableau, such as the
concept of the parity of the sequence of nodes. From the definition of loop and return nodes
(particularly Condition $(\heartsuit)$), we can assume that
\begin{quote}
  $(\dag)$: Let $\xi$ be an infinite $(C_{b}\cup B_{b})$-sequence and let $t \in T_{b}$
  be the return node which appears infinitely often in $\xi$ and is nearest to the root of all such
  return nodes. Then, $\xi$ is even if and only if $\mathsf{pri}(q_{t})$ is even.
\end{quote}

Next, we assign an automaton normal form $\mathsf{anf}(t)$ to each node $t \in T_{b}$ by using
top-down fashion:
\begin{description}
  \item[Base step:]
  Let $t \in T_{b}$ be a leaf. If $t$ is not a loop node, then $t$ must be a modal node with
  an inconsistent label or contain no formula of the form $\triangledown \Phi$. In both cases, we
  assign $\mathsf{anf}(t) := \bigwedge_{1 \leq k \leq i} l_{k}$
  where $\{ l_{1}, \dots, l_{i}\} = L_{b}(t)\cap\mathsf{Lit}(\alpha)$. If $t$ is a loop node,
  we take $x_{t} \in \mathsf{Prop}\setminus\mathsf{Sub}(\varphi)$ uniquely for each such leaf and we set
  $\mathsf{anf}(t) := x_{t}$.
 
  \item[Inductive step I:]
  Suppose $t \in T_{b}$ is a $(\triangledown)$-node where
  $t$ is labeled by $\{ \triangledown \Psi_{1}, \dots, \triangledown \Psi_{i}, l_{1}, \dots, l_{j} \}$
  with $l_{1}, \dots, l_{j} \in \mathsf{Lit}(\alpha)$,
  and we have already assigned the automaton normal form
  $\mathsf{anf}(u)$
  for each child
  $u \in C_{b}(t)$.
  In this situation, we first assign $\mathsf{anf}^{-}(t)$ to $t$ as follows:
  \begin{align}
    \mathsf{anf}^{-}(t) &:= 
      \triangledown \{\mathsf{anf}(u)\mid u \in C_{b}(t)\}\wedge
      \left(\bigwedge_{1\leq k\leq j}l_{k}\right) \notag \\
    &=\left( \bigwedge_{u \in C_{b}(t)} \Diamond \mathsf{anf}(u) \right) \wedge \square
      \left( \bigvee_{1 \leq k \leq i}
        \left(
          \bigvee_{u \in C^{(k)}_{b}(t)}\mathsf{anf}(u)
        \right)
      \right) \wedge
      \left(\bigwedge_{1\leq k\leq j}l_{k}\right) \label{eq: automaton normal form 1}
  \end{align}
  where $C^{(k)}_{b}(t)$ denotes the set of all children $u \in C_{b}(t)$ such that
  $\triangledown \Psi_{k}$ is reduced to some $\psi_{k} \in \Psi_{k}$ between $t$ and $u$. That is,
  we designate the order of disjunction in $\mathsf{anf}^{-}(t)$ for technical reasons
  (see Remark \ref{rem: automaton normal form}).
  If $t$ is not a return node, then we set $\mathsf{anf}(t) := \mathsf{anf}^{-}(t)$.
  Alternatively, if $t$ is a return node, then
  let $t_{1}, \dots, t_{n}$ be all the loop nodes such that
  $(t_{k}, t) \in B_{b}$ $(1 \leq k \leq n)$. We set
  \begin{eqnarray}\label{eq: automaton normal form 2}
    \eta_{t} := \left\{\begin{array}{ll}
    \mu & \text{If $\mathsf{pri}(t) (= \mathsf{pri}(t_{1}) = \dots
      = \mathsf{pri}(t_{n})) \equiv 1 \pmod 2$}\\
    \nu & \text{If $\mathsf{pri}(t) (= \mathsf{pri}(t_{1}) = \dots
      = \mathsf{pri}(t_{n})) \equiv 0 \pmod 2$}
    \end{array}\right.
  \end{eqnarray}
  In this case we define $\mathsf{anf}(t)$
  as $\mathsf{anf}(t) := \eta_{t}x_{t_{1}}.\dots\eta_{t}x_{t_{n}}.\mathsf{anf}^{-}(t)$.

  \item[Inductive step II:]
  Suppose $t \in T_{b}$ is a $(\vee)$-node where, for both children $u_{1}, u_{2} \in C_{b}(t)$,
  we have already assigned the automaton normal forms $\mathsf{anf}(u_{1})$ and $\mathsf{anf}(u_{2})$,
  respectively.
  If $t$ is not a return node, then we set $\mathsf{anf}(t) := \mathsf{anf}(u_{1}) \vee \mathsf{anf}(u_{2})$.
  Suppose $t$ is a return node. Let $t_{1}, \dots, t_{n}$ be all the loop nodes such that
  $(t_{k}, t) \in B_{b}$ $(1 \leq k \leq n)$. In this case, $\eta_{t}$ is defined in the same way as
  $(\ref{eq: automaton normal form 2})$ and
  we define $\mathsf{anf}(t)$ as
  $\mathsf{anf}(t) :=
    \eta_{t}x_{t_{1}}.\dots\eta_{t}x_{t_{n}}.\big(\mathsf{anf}(u_{1}) \vee \mathsf{anf}(u_{2})\big)$.

  \item[Inductive step III:]
  Suppose $t \in T_{b}$ is a $(\wedge)$-, $(\eta)$- or $(\mathsf{Regeneration})$-node
  where
  we have already assigned the automaton normal form $\mathsf{anf}(u)$
  for the child $u \in C_{b}(t)$.
  If $t$ is not a return node, then we assign
  $\mathsf{anf}(t) := \mathsf{anf}(u) \wedge \top$.
  If $t$ is a return node and
  $t_{1}, \dots, t_{n}$ are all the loop nodes such that
  $(t_{k}, t) \in B_{b}$ $(1 \leq k \leq n)$, then, $\eta_{t}$ is defined in the same way as
  $(\ref{eq: automaton normal form 2})$,
  and we define $\mathsf{anf}(t)$ as
  $\mathsf{anf}(t) := \eta_{t}x_{t_{1}}.\dots\eta_{t}x_{t_{n}}.\;\mathsf{anf}(u)$.
\end{description}
We take $\mathsf{anf}(\alpha) := \mathsf{anf}(r_{b})$.

Consider the structure $(T_{b}, C_{b}, r_{b}, \mathsf{anf}, B_{b})$. We intuit that
this structure is almost a tableau with back edge for $\mathsf{anf}(\alpha)$. To clarify
this intuition, we give a
structure $\mathcal{TB}_{\mathsf{anf}(\alpha)} = (\widehat{T}, \widehat{C}, \widehat{r}, \widehat{L},
\widehat{B})$ by applying the following four steps of procedure
re-formatting $(T_{b}, C_{b}, r_{b}, \mathsf{anf}, B_{b})$ so that $\mathcal{TB}_{\mathsf{anf}(\alpha)}$
can be seen as a proper tableau with back edge.
At the same time, we define the relation $Z^{+} \subseteq T_{b} \times \widehat{T}$.
\begin{description}
\item[Step I (insert $(\eta)$-nodes)]
  Initially, we set $(\widehat{T}, \widehat{C}, \widehat{r}, \widehat{L}, \widehat{B}) :=
    (T_{b}, C_{b}, r_{b}, \widehat{L}, B_{b})$ where $\widehat{L}(t) := \{ \mathsf{anf}(t) \}$, and
  set $Z^{+} := \{ (t, t) \mid t \in T_{b} \}$.
  Let $t \in \widehat{T}$ be a return node where $t_{1}, \dots, t_{n}$ are all the
  loop nodes such that
  $(t_{k}, t) \in \widehat{B}$ $(1 \leq k \leq n)$.
  Then, we insert the $(\eta)$-nodes $u_{1}, \dots, u_{n}$ between $t$ and its children
  in such a way that
  \[
    \mathsf{anf}(t) = \eta_{t}x_{t_{1}}.\eta_{t}x_{t_{2}}.\dots\eta_{t}x_{t_{n}}.
    \beta(x_{t_{1}}, \dots, x_{t_{n}})
  \]
  is reduced to $\beta(x_{t_{1}}, \dots, x_{t_{n}})$ from $u_{1}$ to $u_{n}$.\footnote{
  In other words, we add  $u_{1}, \dots, u_{n}$ into $\widehat{T}$,
  add $(t, u_{1}), (u_{1}, u_{2}), \dots, (u_{n-1}, u_{n})$ and
  $\{ (u_{n}, u) \mid u \in \widehat{C}(t) \}$  
  into $\widehat{C}$, discard $\{ (t, u) \mid u \in \widehat{C}(t) \}$ from $\widehat{C}$,
  and expand $\widehat{L}$ to $u_{1}, \dots, u_{n}$ appropriately.
  }
  Moreover, we expand the relation $Z^{+}$ by adding $\{ (t, u_{k}) \mid 1 \leq k \leq n \}$.
  For example, if $t$ is a $(\vee)$-node in $\mathcal{TB}_{\alpha}$ such that
  $\{ v_{1}, v_{2} \} = C_{b}(t)$,
  then our procedure would be as follows:
$$
  \infer[\quad \Rightarrow]{\eta_{t}x_{t_{1}}.\eta_{t}x_{t_{2}}.\dots\eta_{t}x_{t_{n}}.\left(
    \mathsf{anf}(v_{1}) \vee \mathsf{anf}(v_{2})
    \right)}
  {\mathsf{anf}(v_{1}) \;\mid\; \mathsf{anf}(v_{2})} \quad
\infer[(\eta)]
  {
    \eta_{t}x_{t_{1}}.\eta_{t}x_{t_{2}}.\dots\eta_{t}x_{t_{n}}.\left(
    \mathsf{anf}(v_{1}) \vee \mathsf{anf}(v_{2})
    \right)
  }
  {\infer*[(\eta)]
    {
      \eta_{t}x_{t_{2}}.\dots\eta_{t}x_{t_{n}}.\left(
      \mathsf{anf}(v_{1}) \vee \mathsf{anf}(v_{2})
      \right)
    }
    {
      \infer[(\vee)]
        {\mathsf{anf}(v_{1}) \vee \mathsf{anf}(v_{2})}
        {\mathsf{anf}(v_{1}) \:\mid\: \mathsf{anf}(v_{2})}
    }
  }
$$
  \item[Step II (insert $(\wedge)$-nodes)]
    Let $t \in \widehat{T}$ be a node which is labeled by; 
    \[\triangledown \{\mathsf{anf}(u)\mid u \in \widehat{C}(t)\}\wedge
    \left(\bigwedge_{1\leq k\leq j}l_{k}\right).\]
Then, we insert the $(\wedge)$-nodes $u_{0}, \dots, u_{i}$ between $t'$ and its children (i.e., the
nodes of $\widehat{C}(t)$) and label such $u_{1}, \dots, u_{j}$ as below:
$$
\infer[\quad \Rightarrow]
{\triangledown \{\mathsf{anf}(u)\mid u \in \widehat{C}(t)\}\wedge\left(\bigwedge_{1\leq k\leq j}l_{k}\right)}
{\mathsf{anf}(u) \;\mid\; u \in \widehat{C}(t)} \quad
\infer[(\wedge)]
{
  \triangledown \{\mathsf{anf}(u)\mid u \in \widehat{C}(t)\}\wedge\left(\bigwedge_{1\leq k\leq j}l_{k}\right)
}
{\infer*[(\wedge)]
  {
    \triangledown \{\mathsf{anf}(u)\mid u \in \widehat{C}(t)\}, \left(\bigwedge_{1\leq k\leq j}l_{k}\right)
  }
  {
    \infer[(\triangledown)]
      {\triangledown \{\mathsf{anf}(u)\mid u \in \widehat{C}(t)\}, l_{1}, \dots, l_{j}}
      {\mathsf{anf}(u) \;\mid\; u \in \widehat{C}(t)}
  }
}
$$
Further, we expand the relation $Z^{+}$ by adding $\{ (t, u_{k}) \mid 1 \leq k \leq j \}$.

\item[Step III (revise the back edges)]
Let $t_{k}$ with $1 \leq k \leq n$ be the loop node, and $t$ be the return node of $t_{k}$ such that
\begin{align*}
  \mathsf{anf}(t_{k})& = x_{t_{k}}\\
  \mathsf{anf}(t)& = \eta_{t}x_{t_{1}}.\eta_{t}x_{t_{2}}.\dots\eta_{t}x_{t_{n}}.
    \beta(x_{t_{1}}, \dots, x_{t_{n}}).
\end{align*}
If $2 \leq k$, then we delete $(t_{k}, t)$ from $\widehat{B}$ and add $(t_{k}, u_{k})$ into $\widehat{B}$
where $u_{k}$ is the unique nodes satisfying;
\[
\widehat{L}(u_{k}) = \{ \eta_{t}x_{t_{k}}.\dots\eta_{t}x_{t_{n}}.
    \beta(x_{t_{1}}, \dots, x_{t_{n}})\}.
\]
By this revising procedure, for any loop node $t$ and its return node $u$,
$\widehat{L}(t)$ and $\widehat{L}(u)$ form the $(\mathsf{Regeneration})$-rule of $\mathsf{anf}(\alpha)$.

\item[Step IV (add top to label)]
Suppose $t \in \widehat{T}$ and its child $u$ are labeled as follows;
$$
  \infer
  {\mathsf{anf}(u)\wedge \top}
  {\mathsf{anf}(u)}
$$
Then, we add $\top$ to $\widehat{L}(v)$ where $v \in (\widehat{C}\cup\widehat{B})^{+}(t)$
such that, between
the $(\widehat{C}\cup\widehat{B})$-path from
$t$ to $v$, there does not exist a $(\triangledown)$-node.
By this adding procedure, such a $t$ becomes a proper $(\wedge)$-node.
\end{description}
The structure 
$\mathcal{TB}_{\mathsf{anf}(\alpha)} = (\widehat{T}, \widehat{C}, \widehat{r}, \widehat{L},
\widehat{B})$ repaired by the above four procedures
can be seen as a tableau with back edge for $\mathsf{anf}(\alpha)$ in the sense that the following two
assertions hold:
\begin{description}
\item[$(\clubsuit)$]
The unwinding $\mathsf{UNW}_{\widehat{r}}(\mathcal{TB}_{\mathsf{anf}(\alpha)})$ is a tableau of
$\mathsf{anf}(\alpha)$.

\item[$(\diamondsuit)$]
Let $\widehat{\xi}$ be an infinite $(\widehat{C}\cup\widehat{B})$-sequence and let
  $\widehat{t} \in \widehat{T}$
  be the return node which appears infinitely often in $\widehat{\xi}$ and is nearest to the root of all such
  return nodes. Then $\widehat{\xi}$ is even if and only if $\widehat{L}(\widehat{t})$ includes
  a $\mu$-formula.
\end{description}

Set
$Z := Z^{+}|_{((T_{b})_{m} \times \widehat{T}_{m}) \cup ((T_{b})_{c} \times \widehat{T}_{c})}$.
If we extend the relation $Z$ to the pair of nodes of
$\mathsf{UNW}_{r}(\mathcal{TB}_{\alpha})$ and
$\mathsf{UNW}_{\widehat{r}}(\mathcal{TB}_{\mathsf{anf}(\alpha)})$, then $Z$ clearly satisfies the
root condition, prop condition, back conditions and forth conditions.
Moreover, from $(\dag)$ and $(\diamondsuit)$,
we can assume that $Z$ satisfies the Parity condition. 
Therefore, we have
$\mathsf{UNW}_{r}(\mathcal{TB}_{\alpha})\rightleftharpoons\mathsf{UNW}_{\widehat{r}}(\mathcal{TB}_{\mathsf{anf}(\alpha)})$,
and so $\mathcal{T}_{\alpha} := \mathsf{UNW}_{r}(\mathcal{TB}_{\alpha})$ and
$\mathsf{anf}(\alpha)$ satisfy the required condition.
\end{proof}

\begin{Remark}\normalfont\label{rem: automaton normal form}
Let $\mathsf{Sub}'(\mathsf{anf}(\alpha))$ be the set of subformulas of $\mathsf{anf}(\alpha)$ which
contains some bound variables.
From the relation $Z^{+}$ constructed in the proof of Theorem \ref{the: automaton normal form},
we can construct 
a function $f$ from $\mathsf{Sub}'(\mathsf{anf}(\alpha))$
to $\mathcal{P}(\mathsf{Sub}(\alpha))$ naturally because of the following:
\begin{itemize}
  \item for any
    $\widehat{\beta} \in \mathsf{Sub}'(\mathsf{anf}(\alpha))$,
    there exists a unique $\widehat{t} \in \widehat{T}$ such that
    $\widehat{\beta} \in \widehat{L}(\widehat{t})$; and
  \item for any $\widehat{t} \in \widehat{T}$ there exists a unique $t \in T_{b}$ such that
    $(t, \widehat{t}) \in Z^{+}$.
\end{itemize}
Therefore, if we define
$f(\widehat{\beta}) := L(t)$ where $\widehat{\beta} \in \widehat{L}(\widehat{t})$ and
$(t, \widehat{t}) \in Z^{+}$, then the function $f$ is well-defined.
Moreover, let $t \in T_{b}$ be a $(\triangledown)$-node
such that $L_{b}(t) = \{ \triangledown \Psi_{1}, \dots, \triangledown \Psi_{i}, l_{1}, \dots, l_{j} \}$.
Then,
we expand $f$ to the formula $\chi_{1}$ and $\chi_{2}$ such that
\[
\mathsf{anf}(u) \leq \chi_{1} \leq
\bigvee_{u \in C^{(k)}_{b}(t)}\mathsf{anf}(u) \leq
\chi_{2} \leq
\left(
  \bigvee_{1 \leq k \leq i}\left(\bigvee_{u \in C^{(k)}_{b}(t)}\mathsf{anf}(u)\right)
\right),
\]
for every $k$ where $1 \leq k \leq i$ and for every $u \in C^{(k)}_{b}(t)$.
Now, we define $f(\chi_{2})$ as
\[
  f(\chi_{2}) := \left\{ \bigvee \Psi_{n} \mid 1 \leq n \leq i \right\}.
\]
Next, we note that for any $u \in C^{(k)}_{b}(t)$ there is a unique $\psi_{k} \in \Psi_{k}$ such that
$\triangledown \Psi_{k}$ is reduced to $\psi_{k}$. We denote such a $\psi_{k}$ by $\mathsf{cor}(u)$.
Suppose $\chi_{1} = \bigvee_{u \in X^{(k)}}\mathsf{anf}(u)$ where $X^{(k)} \subseteq C^{(k)}_{b}(t)$.
Then we define $f(\chi_{1})$ as;
\[
  f(\chi_{1}) := \left\{ \bigvee \Psi_{n} \mid 1 \leq n \leq i, \; n \neq k \right\} \cup
  \left\{ \bigvee_{u \in X^{(k)}}\mathsf{cor}(u) \right\}.
\]
Recalling Equation $(\ref{eq: automaton normal form 1})$,
the reason we designated the order of disjunction in $\mathsf{anf}(t)$ is that,
in conjunction with above definition of $f$, we obtain the following useful property:
\begin{description}
\item[(Corresponding Property):]
Consider the section of the tableau which has the root labeled by 
\[
\left(
  \bigvee_{1 \leq k \leq i}\left(\bigvee_{u \in C^{(k)}_{b}(t)}\mathsf{anf}(u)\right)
  \right),
\]
and every leaf labeled by some $\mathsf{anf}(u)$. Then, for any node $u$ and its children
$v_{1}$ and $v_{2}$ we have (i) $f(L(u)) = f(L(v_{1})) =  f(L(v_{2}))$ or, (ii)
$f(L(u))$, $f(L(v_{1}))$ and $f(L(v_{2}))$ forming a $(\vee)$-rule.
\end{description}

Let us confirm the above property by observing a concrete example as depicted in Figure
\ref{fig: an example of the corresponding property}.
\begin{figure}[htbp]
  \centering
  \includegraphics[width=16cm]{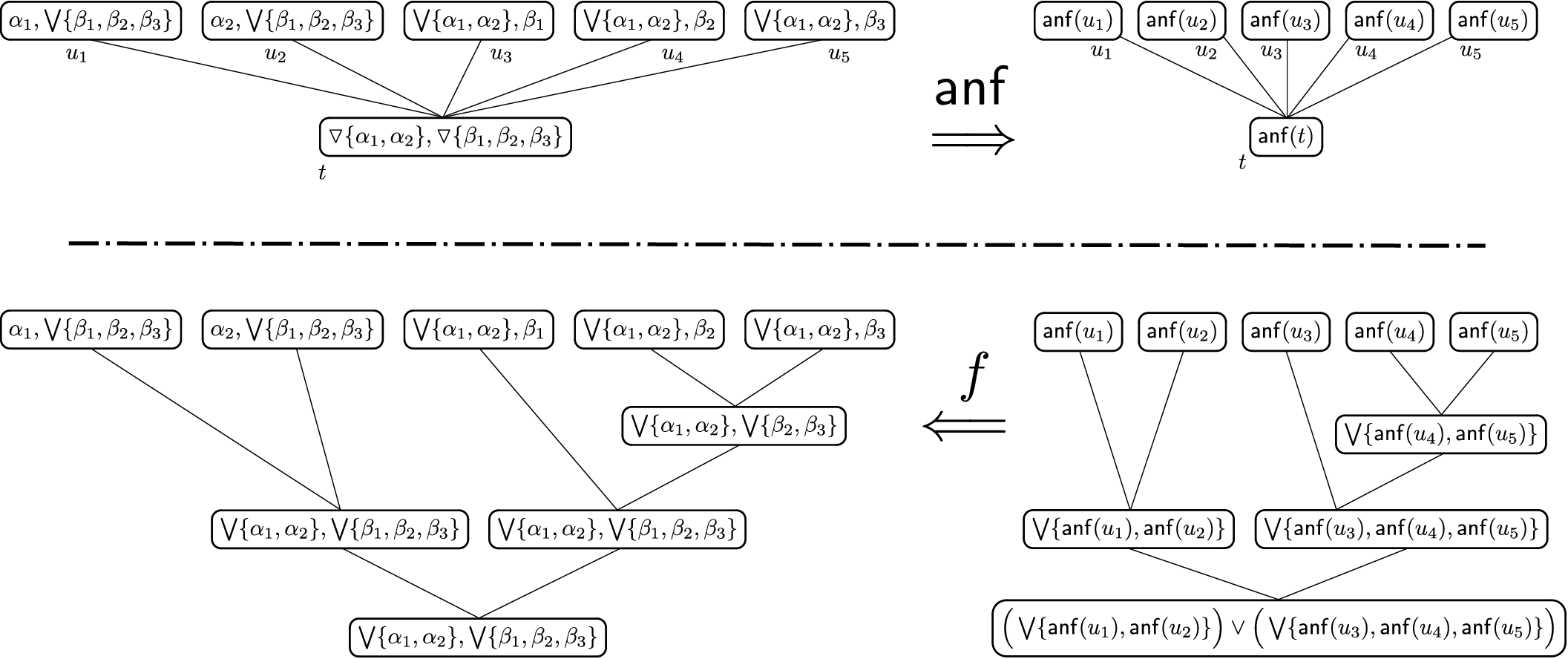}
  \caption{An example of the corresponding property.}
  \label{fig: an example of the corresponding property}
\end{figure}
In this example, the root and its children satisfy (i), and the child of the root and its children
form a $(\vee)$-rule. Thus, (ii) is satisfied.

The function $f$ will be used in the proof of Part $4$
of Lemma \ref{lem: basic properties of tableau consequence}.
\end{Remark}

\begin{Corollary}
For any well-named formula $\alpha$, we can construct an automaton normal form
$\mathsf{anf}(\alpha)$ which is semantically equivalent to $\alpha$.
Moreover, for any $x \in \mathsf{Free}(\alpha)$
which occurs only positively in $\alpha$, it holds that
$x \in \mathsf{Free}(\mathsf{anf}(\alpha))$ and $x$ occurs
only positively in $\mathsf{anf}(\alpha)$.
\end{Corollary}

\begin{proof}
This is an immediate consequence of
Lemma \ref{lem: basic properties of tableau bisimulation}
and Theorem \ref{the: automaton normal form}.
\end{proof}

\section{Completeness}\label{sec: completeness}
This section is the final section of this article and includes the main part.
In Subsection \ref{subsec: tableau consequence}, 
we give the concept of \textit{tableau consequence} and show Claim (g); that may be the most difficult to understand
in Walukiewicz \cite{Walukiewicz2000142}.
In Subsection \ref{subsec: proof of completeness}, we prove the completeness of $\mathsf{Koz}$
by proving Claim (h) and (d), in that order.

\subsection{Tableau consequence}\label{subsec: tableau consequence}
First, we extend the definition of tableau for technical reasons.

\begin{Definition}[\bf An extension of tableau]\normalfont
Let $\varphi$ be a well-named formula. The 
rule of a \textit{extended tableau} for $\varphi$ is obtained by adding the following three rules
to the rule of tableau:
$$
  \infer[(\epsilon_{1})]
  {\Gamma}
  {\Gamma} \qquad
  \infer[(\epsilon_{2})]
  {\Gamma}
  {\Gamma \;\mid\; \Gamma}
$$
$$
  \infer[(\triangledown_{e})]
  {\triangledown\Psi_{1}, \dots, \triangledown\Psi_{i}, l_{1}, \dots, l_{j}}
  {\bigvee\Psi_{1}, \dots, \bigvee\Psi_{i} \mid \cdots \mid \bigvee\Psi_{1}, \dots, \bigvee\Psi_{i}}
$$
where in the $(\triangledown_{e})$-rule, $l_{1}, \dots, l_{j} \in \mathsf{Lit}(\varphi)$ and,
the label of premises are all the same (i.e., $\{ \bigvee\Psi_{1}, \dots, \bigvee\Psi_{i} \}$),
and the number of premises is an arbitrary finite number.

An \textit{extended tableau} for $\varphi$ is the structure
defined as a tableau for $\varphi$, but satisfying the following additional clause:
\begin{enumerate}
\setcounter{enumi}{3}
\item
  For any infinite branch $\xi$ of an extended tableau $\mathcal{T}_{\varphi}$,
  $\{ n \in \omega \mid \text{$\xi[n]$ is $(\triangledown)$-node or $(\triangledown_{e})$-node}\}$
  is an infinite set.
\end{enumerate}
Clause $4$ restrains a branch that does not reach any modal node eternally by
infinitely applying $(\epsilon_{1})$ and $(\epsilon_{2})$.
\end{Definition}

\begin{Remark}\normalfont
A tableau can be considered a special case of an extended tableau, in which the extended rules
are not used. Various concepts for tableau, such as trace, parity and tableau bisimulation, can be introduced into
this extended tableau as well. Thus, we apply these concepts and results
freely to this new structure.
\end{Remark}

\begin{Definition}[\bf Tableau consequence]\label{def: tableau consequence}\normalfont
Let $\mathcal{T}_{\alpha} = (T, C, r, L)$ and $\mathcal{T}_{\beta} = (T', C', r', L')$ be two
extended tableaux for some well-named formula $\alpha$ and $\beta$. Let $T_{m}$
and $T'_{m}$ be the set of modal nodes of $\mathcal{T}_{\alpha}$ and $\mathcal{T}_{\beta}$,
and let $T_{c}$ and $T'_{c}$ be the set of choice nodes of
$\mathcal{T}_{\alpha}$ and $\mathcal{T}_{\beta}$, respectively. Then $\mathcal{T}_{\beta}$ is called a
\textit{tableau consequence} of $\mathcal{T}_{\alpha}$
(notation: $\mathcal{T}_{\alpha}\rightharpoonup\mathcal{T}_{\beta}$) if there exists a binary
relation $Z\subseteq (T_{m}\times T'_{m})\cup(T_{c}\times T'_{c})$ satisfying the following six conditions
(here, the condition of the tableau consequence is similar to the condition of tableau
bisimulation so we have illustrated the differences between these two conditions using underlines):
\begin{description}
\item[Root condition:]
$(r, r') \in Z$.

\item[Prop condition:]
For any $t \in T_{m}$ and $t' \in T'_{m}$, if $(t, t') \in Z$, then
\[
  L(t)\cap\mathsf{Lit}(\alpha) \uwave{\;\supseteq\;}
    L'(t')\cap\mathsf{Lit}(\beta).
\]
Consequently, $L(t)$ is consistent \uwave{only if} $L'(t')$ is consistent.

  \item[Forth condition on modal nodes:]
  \uwave{Take $t, u \in T_{m}$ and $t' \in T'_{m}$ arbitrarily.
  If $(t, t') \in Z$ and $u$ is a next modal node of $t$, then $C'(t') = \emptyset$ or
  there exists $u' \in T'_{m}$
  which is a next modal node of $t'$ such that $(u, u') \in Z$.}

  \item[Back condition on modal nodes:]
  Take $t \in T_{m}$, $t' \in T'_{m}$ and $u' \in T'_{c}$ arbitrarily.
  If $(t, t') \in Z$ and $u' \in C'(t')$, then \uwave{$C(t) = \emptyset$} or there exists $u \in C(t)$
  such that $(u, u') \in Z$.

  \item[Forth condition on choice nodes:]
  Take $u \in T_{c}$, $t \in T_{m}$ and $u' \in T'_{c}$ arbitrarily.
  If $(u, u') \in Z$ and $t$ is near $u$, then there exists $t' \in T'_{m}$ such that $(t, t') \in Z$
  and $t'$ is near $u'$.

  \item[Back condition on choice nodes:]
  \uwave{No condition.}

  \item[Parity condition:]
  Let $\xi$ and $\xi'$ be infinite branches of $\mathcal{T}_{\alpha}$ and $\mathcal{T}_{\beta}$
  respectively.
  If $\xi$ and $\xi'$ are associated with each other, then $\xi$ is even \uwave{if}
  $\xi'$ is even.
\end{description}
A relation $Z$ which satisfies the above six conditions called \textit{tableau consequence relation}
from $\mathcal{T}_{\alpha}$ to $\mathcal{T}_{\beta}$.
\end{Definition}

\begin{Remark}\label{rem: tableau bisimulation}\normalfont
As will be shown in Lemma \ref{lem: basic properties of tableau bisimulation},
if $\mathcal{T}_{\alpha}$ and $\mathcal{T}_{\beta}$ are tableau bisimilar, then,
$\alpha$ and $\beta$ are semantically equivalent. However, the reverse is not applied.
For example, consider the following two tableaux, say $\mathcal{T}_{1}$ and $\mathcal{T}_{2}$:
$$
\infer[(\wedge)]
{
  (p \wedge (q \vee r)) \wedge (q \vee r)
}
{
  \infer[(\wedge)]
  {
    p \wedge (q \vee r), q \vee r
  }
  {
    \infer[(\vee)]
    {
      p, q \vee r
    }
    {
      p, q \quad \mid \quad p, r
    }
  }
}
\qquad
\infer[(\wedge)]
{
  (p \wedge (q \vee r)) \wedge (q \vee r)
}
{
  \infer[(\vee)]
  {
    p \wedge (q \vee r), q \vee r
  }
  {
    \infer[(\wedge)]
    {
      p \wedge (q \vee r), q
    }
    {
      \infer[(\vee)]
      {
        p, q \vee r, q
      }
      {
        p, q
        \quad \mid \quad
        p, q, r
      }
    }
    \quad \mid \quad
    \infer[(\wedge)]
    {
      p \wedge (q \vee r), r
    }
    {
      \infer[(\vee)]
      {
        p, q \vee r, r
      }
      {
        p, q, r
        \quad \mid \quad
        p, r
      }
    }
  }
}
$$
In this example, even $\mathcal{T}_{1}$ and $\mathcal{T}_{2}$ are tableaux for the same formula 
$(p \wedge (q \vee r)) \wedge (q \vee r)$,
there does not exist a tableau bisimulation between them. Because,
$\mathcal{T}_{2}$ has leaves labeled by $\{ p, q, r\}$ but $\mathcal{T}_{1}$ does not.

On the other hand, we can assume that
$\mathcal{T}_{2}\rightharpoonup\mathcal{T}_{1}$.
Suppose $t$ is a node of some tableau labeled by $\{ \gamma \} \cup \Gamma$ and,
$u$ is its child labeled by $\{ \gamma' \} \cup \Gamma$. Then, there are two possibilities;
$\gamma' \in \Gamma$ or $\gamma' \notin \Gamma$.
We say a collision occurred between $t$ and $u$ if $\gamma' \in \Gamma$.
In the above example, we can find collisions in
$\mathcal{T}_{1}$ but cannot in $\mathcal{T}_{2}$. In general, if we construct a tableau
$\mathcal{ST}_{\varphi}$ for a given
formula $\varphi$ so that collisions occur as many as possible, then,
we have $\mathcal{T}_{\varphi}\rightharpoonup\mathcal{ST}_{\varphi}$
for any tableau $\mathcal{T}_{\varphi}$ for $\varphi$. To denote this fact correctly, we introduce the
following definition and lemma.
\end{Remark}

\begin{Definition}[\bf Small tableau]\normalfont
A well-named formula $\varphi$ and a set $\Gamma \subseteq \mathsf{Sub}(\varphi)$ are given.
For a formula $\gamma \in \Gamma$, a closure of $\gamma$ (denotation: $\mathsf{cl}(\gamma)$)
is defined as follows:
\begin{itemize}
\item $\gamma \in \mathsf{cl}(\gamma)$.
\item If $\alpha \circ \beta \in \mathsf{cl}(\gamma)$, then $\alpha, \beta \in \mathsf{cl}(\gamma)$
      where $\circ \in \{ \vee, \wedge \}$.
\item If $\eta_{x}x.\varphi_{x}(x) \in \mathsf{cl}(\gamma)$, then $\varphi_{x}(x) \in \mathsf{cl}(\gamma)$.
\item If $x \in \mathsf{cl}(\gamma) \cap \mathsf{Bound}(\varphi)$, then $\varphi_{x}(x) \in \mathsf{cl}(\gamma)$.
\end{itemize}
In other words, $\mathsf{cl}(\gamma)$ is a set of all formulas $\delta$ such that
for any tableau $\mathcal{T}_{\varphi} = (T, C, r, L)$ and its node $t \in T$,
if $\gamma \in L(t)$, then,
there is a descendant $u \in C^{\ast}(t)$ near $t$
and a trace $\mathsf{tr}$ on the $C$-sequence from $t$ to $u$ where $\mathsf{tr}[1] = \gamma$
and $\mathsf{tr}[|\mathsf{tr}|] = \delta$.
We say $\gamma$ is \textit{reducible} in $\Gamma$ if, for any $\gamma' \in \Gamma \setminus \{ \gamma \}$,
we have $\gamma \notin \mathsf{cl}(\gamma')$.
A tableau $\mathcal{ST}_{\varphi} = (T, C, r, L)$ is said \textit{small} if for any node $t \in T$
which is not modal, the reduced formula $\gamma \in L(t)$
between $t$ and its children is reducible in $L(t)$.
\end{Definition}

\begin{Lemma}\label{lem: small tableau}
For any well-named formula $\varphi$, we can construct a small tableau $\mathcal{ST}_{\varphi}$ for $\varphi$.
Moreover, for any extended tableau $\mathcal{T}_{\varphi}$ for $\varphi$,
we have $\mathcal{T}_{\varphi}\rightharpoonup\mathcal{ST}_{\varphi}$.
\end{Lemma}
\begin{proof}
Let $\varphi$ be a well-named formula. Then, it is enough to show that
for any $\Gamma \subseteq \mathsf{Sub}(\varphi)$ which is not modal, there exists
a reducible formula $\gamma \in \Gamma$.
Suppose, moving toward a contradiction, that there exists $\Gamma \subseteq \mathsf{Sub}(\varphi)$
which is not modal and contains no reducible formula.
Take a formula $\gamma_{1} \in \Gamma$ such that $\mathsf{cl}(\gamma_{1}) \supsetneq \{ \gamma_{1} \}$.
Since $\gamma_{1}$ is not reducible in $\Gamma$, there exists
$\gamma_{2} \in \Gamma \setminus \{ \gamma_{1} \}$ such that $\gamma_{1} \in \mathsf{cl}(\gamma_{2})$.
Since $\gamma_{2}$ is not reducible in $\Gamma$, there exists
$\gamma_{3} \in \Gamma \setminus \{ \gamma_{2} \}$ such that $\gamma_{2} \in \mathsf{cl}(\gamma_{3})$.
And so forth, we obtain the sequence $\langle \gamma_{n} \mid n \in \omega \setminus \{ 0 \} \rangle$
such that $\gamma_{n+1} \in \Gamma \setminus \{ \gamma_{n} \}$
and $\gamma_{n} \in \mathsf{cl}(\gamma_{n+1})$ for any $n \in \omega \setminus \{ 0 \}$.
Since $|\Gamma|$ is finite, there exists $i, j \in \omega$
such that $1 \leq i < j$ and $\gamma_{i} = \gamma_{j}$. Consider the tableau
$\mathcal{T}_{\varphi} = (T, C, r, L)$ and its node $t \in T$ such that $\gamma_{j} \in L(t)$.
Then, from the definition of the closure $\mathsf{cl}$, there exists a
trace $\mathsf{tr}$ on $\pi$ such that:
\begin{description}
\item[$(\heartsuit)$] $\pi$ is a finite $C$-sequence starting at $t$ where
  $(\triangledown)$-rule is not applied between $\pi$.
\item[$(\clubsuit)$] $\mathsf{tr}[1] = \mathsf{tr}[|\mathsf{tr}|] =  \gamma_{j}$.
\end{description}
On the other hand, since $\varphi$ is well-named, for any bound variable $x \in \mathsf{Bound}(\varphi)$,
$x$ is in the scope of some modal operator (cover modality) in $\varphi_{x}(x)$. Thus we have:
\begin{description}
\item[$(\spadesuit)$] For any trace $\mathsf{tr}$ on $\pi$, if $(\clubsuit)$ is satisfied,
  then $\pi$ includes a $(\triangledown)$-node or $(\triangledown_{w})$-node.
\end{description}
$(\heartsuit)$ and $(\spadesuit)$ contradict each other.
The proof of the second half of the lemma is left as a reader's exercise.
\end{proof}

The next lemma states some important properties of the tableau consequence;
where the proof of the lemma is easier to understand than Walukiewicz's proof, and is the main contribution of this article.

\begin{Lemma}\label{lem: basic properties of tableau consequence}
Let $\alpha$, $\beta$, $\gamma$ and $\varphi(x)$ be well-named formulas
where $x$ appears only positively and in the scope of some modality in $\varphi(x)$.
Then, we have:
\begin{enumerate}
\item
  If $\mathcal{T}_{\alpha}\rightleftharpoons\mathcal{T}_{\beta}$, then 
  $\mathcal{T}_{\alpha}\rightharpoonup\mathcal{T}_{\beta}$,
  for any extended tableaux $\mathcal{T}_{\alpha}$ and $\mathcal{T}_{\beta}$.

\item
  If $\mathcal{T}_{\alpha}\rightharpoonup\mathcal{T}_{\beta}$ and
  $\mathcal{T}_{\beta}\rightharpoonup\mathcal{T}_{\gamma}$, then
  $\mathcal{T}_{\alpha}\rightharpoonup\mathcal{T}_{\gamma}$,
  for any extended tableaux $\mathcal{T}_{\alpha}$, $\mathcal{T}_{\beta}$ and $\mathcal{T}_{\gamma}$.

\item
  For any extended tableau $\mathcal{T}_{\varphi(\mu \vec{x}.\varphi(\vec{x}))}$,
  there exists an extended tableau $\mathcal{T}_{\mu \vec{x}.\varphi(\vec{x})}$ such that
  $\mathcal{T}_{\varphi(\mu \vec{x}.\varphi(\vec{x}))}
    \rightharpoonup\mathcal{T}_{\mu \vec{x}.\varphi(\vec{x})}$.

\item
  For any extended tableau $\mathcal{T}_{\varphi(\mathsf{anf}(\alpha))}$,
  there exists an extended tableau
  $\mathcal{T}_{\varphi(\alpha)}$ such that
  $\mathcal{T}_{\varphi(\mathsf{anf}(\alpha))}\rightharpoonup\mathcal{T}_{\varphi(\alpha)}$.
\end{enumerate}
\end{Lemma}

\begin{proof}
Part $1$ and Part $2$ are obvious from the definition.

\fbox{Part$3$}
First, we divide $\mathsf{Sub}(\varphi(\mu x.\varphi(x)))$ into two disjoint sets:
\begin{align*}
    \mathsf{Sub}_{1} &:=
    \big\{ \alpha(\mu x.\varphi(x)) \mid \alpha(x) \in \mathsf{Sub}(\varphi(x)) \big\} \setminus \{ \mu x.\varphi(x) \}\\
    \mathsf{Sub}_{2} &:= \mathsf{Sub}(\varphi(\mu x.\varphi(x))) \setminus \mathsf{Sub}_{1}=\mathsf{Sub}(\mu x.\varphi(x))
\end{align*}
A function $f: \mathsf{Sub}(\varphi(\mu x.\varphi(x))) \rightarrow \mathsf{Sub}(\mu x.\varphi(x))$ is defined as follows:
\begin{eqnarray*}
f(\psi) := \left\{\begin{array}{ll}
    \alpha(x) & \text{$\psi = \alpha(\mu x.\varphi(x)) \in \mathsf{Sub}_{1}$,}\\
    \psi & \text{otherwise.}
  \end{array}\right.
\end{eqnarray*}
Take an extended tableau $\mathcal{T}_{\varphi(\mu x.\varphi(x))} = (T, C, r, L)$ arbitrarily.
Set $Z := \{ (t, t) \mid t \in T \}$. If
\begin{equation}
  \mathcal{T}_{\varphi(\mu x.\varphi(x))}\rightharpoonup(T, C, r, f\circ L) \label{eq: tableau consequence 1}
\end{equation}
with $Z$,
then we are done. However unfortunately $(\ref{eq: tableau consequence 1})$ is generally incorrect.
$Z$ generally does not satisfy the parity condition among the requests for tableu consequences.
Let's explain that with a concrete example.
Suppose that $\xi$ is an infinite branch of $\mathcal{T}_{\varphi(\mu x.\varphi(x))}$ where
there are only two traces, $\mathsf{tr}_{1}$ and $\mathsf{tr}_{2}$ on it.
Moreover, suppose that $\mathsf{tr}_{1}$ is a trace on $\mathsf{Sub}_{1}$
(i.e., $\mathsf{Infinite}(\mathsf{tr}_{1}) \subseteq \mathsf{Sub}_{1}$), and
 $\mathsf{tr}_{2}$ is a trace on $\mathsf{Sub}_{2}$
(i.e., $\mathsf{Infinite}(\mathsf{tr}_{2}) \subseteq \mathsf{Sub}_{2}$).
Note that
\begin{itemize}
  \item[$(\spadesuit)$] $\mathsf{tr}_{i}$ is even (i.e., $\mathsf{tr}_{i}$ is a $\mu$-trace) $\Leftrightarrow$
  $\vec{f}(\mathsf{tr}_{i})$ is even (i.e., $\vec{f}(\mathsf{tr}_{i})$ is a $\mu$-trace) $(i = 1, 2)$
\end{itemize}
holds from the definition of $f$.
Suppose that $\vec{f}(\mathsf{tr}_{1})$ and $\vec{f}(\mathsf{tr}_{2})$ repeat merging and branching as shown in
Figure \ref{fig: an example that does not satisfy the parity condition}.
\begin{figure}[htbp]
  \centering
\[
  \xymatrix@R=20pt{
    \cdots &
    \mathsf{tr}_{1}(k) \ar[dd]^{f} &
    \mathsf{tr}_{1}(k+1) \ar[d]^{f} &
    \mathsf{tr}_{1}(k+2) \ar[dd]^{f} &
    \mathsf{tr}_{1}(k+3) \ar[d]^{f}  &
    \mathsf{tr}_{1}(k+4) \ar[dd]^{f} &
    \cdots
    \\
    \ar@{.>}[dr] &
    &
    \alpha_{1} \ar@{=>}[dr] &
    &
    \alpha_{3} \ar@{.>}[dr] &
    &
    \\
    &
    \alpha_{2} \ar@{=>}[ur] \ar@{.>}[dr] &
    &
    \alpha_{2} \ar@{.>}[ur] \ar@{=>}[dr] &
    &
    \alpha_{2} \ar@{=>}[ur] \ar@{.>}[dr] &
    \\
    \ar@{=>}[ur] &
    &
    \alpha_{3} \ar@{.>}[ur] &
    &
    \alpha_{1} \ar@{=>}[ur] &
    &
    \\
    \cdots &
    \mathsf{tr}_{2}(k) \ar[uu]^{f} &
    \mathsf{tr}_{2}(k+1) \ar[u]^{f} &
    \mathsf{tr}_{2}(k+2) \ar[uu]^{f} &
    \mathsf{tr}_{2}(k+3) \ar[u]^{f}  &
    \mathsf{tr}_{2}(k+4) \ar[uu]^{f} &
    \cdots
    }
\]
  \caption{An example that does not satisfy the parity condition.}
  \label{fig: an example that does not satisfy the parity condition}
\end{figure}
Suppose $\Omega_{\mu x.\varphi(x)}(\alpha_{i}) = i$ ($i = 1, 2, 3$).
Then, we have
\[
  \max \Omega_{\mu x.\varphi(x)}\big(\mathsf{Infinite}(\vec{f}(\mathsf{tr}_{1}))\big) = 
  \max \Omega_{\mu x.\varphi(x)}\big(\mathsf{Infinite}(\vec{f}(\mathsf{tr}_{2}))\big) = 3.
\]
Therefore, from $(\spadesuit)$, it turns out that $\mathsf{tr}_{1}$ and $\mathsf{tr}_{2}$ are odd.
Thus, $\vec{L}(\xi)$ is also odd.
On the other hand, set $\mathsf{tr}_{3} = \vec{f}(\mathsf{tr}_{1}[1, k-1])(\alpha_{2}\alpha_{1})^{\omega}$
(i.e., $\mathsf{tr}_{3}$ is the trace represented by $\Rightarrow$ in
Figure \ref{fig: an example that does not satisfy the parity condition}).
Since $\max \Omega_{\mu x.\varphi(x)}\big(\mathsf{Infinite}(\vec{f}(\mathsf{tr}_{3}))\big) = 2$,
$\overrightarrow{f\circ L}(\xi)$ is even.
This means that $Z$ does not satisfy the parity condition.

It turns out that simply compositing $f$ and the label of $\mathcal{T}_{\varphi(\mu x.\varphi(x))}$ didn't work.
The problem is that $\mathsf{tr}_{1}$ and $\mathsf{tr}_{2}$ may exist such that
$\vec{f}(\mathsf{tr}_{1})$ and $\vec{f}(\mathsf{tr}_{2})$ repeat branching and merging infinitely often,
and these may break the parity condition guaranteed by $(\spadesuit)$.
Therefore, we overcome this obstacle by using the {\bf horizontal prunning} technique
shown in Safra's construction.
We will construct an extended tableau $\mathcal{T}_{\mu x.\varphi(x)}$ where
$\mathcal{T}_{\varphi(\mu x.\varphi(x))}
    \rightharpoonup\mathcal{T}_{\mu x.\varphi(x)}$ holds by $Z$ in the following 5 steps:
\begin{steps}
  \item We define the B\"{u}chi automaton $\mathcal{BA}_{\varphi(\mu x.\varphi(x))} := \langle Q, \Sigma, q_{0}, \Delta, F \rangle$
  as follows:
  \begin{itemize}
    \item $Q := \{ (\Gamma, \gamma) \mid \Gamma \subseteq \mathsf{Sub}(\varphi(\mu x.\varphi(x))), \; \gamma \in \Gamma \}$.
    \item $\Sigma := \mathcal{P}(\mathsf{Sub}(\varphi(\mu x.\varphi(x))))$.
    \item $q_{0} := \big(\{ \varphi(\mu x.\varphi(x)) \}, \varphi(\mu x.\varphi(x))\big)$.
    \item $\Delta(\Gamma', (\Gamma, \gamma)) := \big\{ (\Gamma', \gamma') \mid
    \gamma' \in \mathsf{TR}_{\Gamma, \Gamma'}(\gamma) \big\}$. 
    \item $F := \big\{ q_{0} \big\} \bigcup \big\{ (\Gamma, \mu x.\varphi(x))
    \mid (\Gamma, \mu x.\varphi(x)) \in Q \big\}$.
  \end{itemize}
  Note that we are not interested in $\mathcal{L}(\mathcal{BA}_{\varphi(\mu x.\varphi (x))})$.
  $\mathcal{BA}_{\varphi(\mu x.\varphi (x))}$ is constructed only for the use of {\bf horizontal prunning}
  in the Rabin automaton that will be constructed later.

  \item
  Convert nondeterministic B\"uchi automaton $\mathcal{BA}_{\varphi(\mu x.\varphi(x))}$
  to deterministic Rabin automaton $\mathcal{RA}_{\varphi(\mu x.\varphi(x))}$
  using Safra's construction.
  However, the following two points are changed from the construction described in
  Subsection \ref{subsec: safra's construction}:
  \begin{itemize}
    \item
    The automaton $\mathcal{BA}_{\varphi(\mu x.\varphi(x))}$ reads the alphabet
    $\{ \varphi(\mu x.\varphi(x)) \}$ in the initial state
    \[ q_{0} = \big(\{ \varphi(\mu x.\varphi(x)) \}, \varphi(\mu x.\varphi(x))\big) \]
    and transitions to the next state $q_{0}$ (as a result, the state does not change).
    Since $q_{0} \in F$ and $\pi_{1} = \big(N_{\varphi(\mu x.\varphi(x))}, \dots, 3, 2, 1\big)$,
    normally, by {\bf add new children}, we add $2$ as a new child.
    Now change the child to be added from $2$ to $N_{\varphi(\mu x.\varphi(x))}$.

    \item
    In the {\bf initialize the index appearance record}, abolish driving $j$ painted in red to the left end.
    Instead, change it so that it is driven to the left end excluding $N_{\varphi(\mu x.\varphi(x))}$
    (see Figure \ref{fig: a change of initialize the index appearance record}):
    \begin{figure}[htbp]
    \centering
      \[
        \xymatrix@R=50pt{
          N_{\varphi(\mu x.\varphi(x))} \ar[d] & \pi[2] \ar[drr] & \text{\mask{$\pi[3]$}{A}}_{\mathrm{red}} \ar[dl] 
          & \pi[4] \ar[rd] & \text{\mask{$\pi[5]$}{A}}_{\mathrm{red}} \ar[dll] & \pi[6] \ar[d] & \pi[7] \ar[d]\\
          N_{\varphi(\mu x.\varphi(x))}             & \pi[3]            & \pi[5]            & \pi[2]            
          & \pi[4]            & \pi[6]           & \pi[7]
        }
      \]
    \caption{A change of initialize the index appearance record.}
    \label{fig: a change of initialize the index appearance record}
    \end{figure}
  \end{itemize}

  \item
  Let the automaton defined above be
  $\mathcal{RA}_{\varphi(\mu x.\varphi(x))} = \langle Q', \Sigma, q_{0}, \Delta', \{ (A_{j}, R_{j}) \mid j \in J \}\rangle$.
  For a tableau node $t$, set $\Delta'(\vec{L}(\vec{rt}), q_{0}) := \langle S_{t}, C_{t}, 1, l_{t}, \pi_{t}, \mathsf{col}_{t}\rangle$.
  In this situation, Safra tree $\langle S_{t}, C_{t}, 1, l_{t}, \pi_{t}, \mathsf{col}_{t}\rangle$
  looks like Figure \ref{fig: a state of automaton RA}.\footnote{Here, in the same way as Remark \ref{rem: duplication},
  instead of thinking that each vertex $j$ is labeled
  with a set of elements in the shape of $(\Gamma, \gamma)$, it is simply labeled with a set of formulas.}
    \begin{figure}[htbp]
    \centering
    \includegraphics[width=12cm]{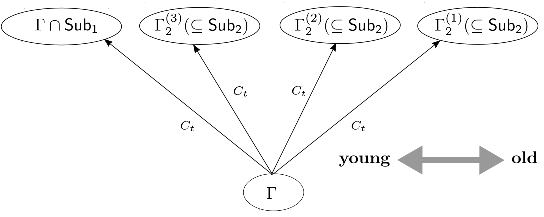}
    \caption{A state of automaton $\mathcal{RA}_{\varphi(\mu x.\varphi(x))}$.}
    \label{fig: a state of automaton RA}
    \end{figure}
  That is, the youngest child of the root $1$ is $N_{\varphi(\mu x.\varphi(x))}$, labeled with a subset of $\mathsf{Sub}_{1}$.
  The other children of the root $1$ are labeled with a subset of $\mathsf{Sub}_{2}$.

  \item
  For each node $t$, we will define a labeled tree $\langle J_{t}, C_{t}, 1, l'_{t}, \pi_{t}, \mathsf{col}_{t}\rangle$ inductively
  from the root to the leaf; where
  $l'_{t}: J_{t} \rightarrow \mathcal{P}(\mathsf{Sub}(\mu x.\varphi(x)))$.
  \begin{description}
    \item[The basis of induction:] $l'_{r}(r) := \{ \mu x.\varphi(x) \} = f(\{ \varphi(\mu x.\varphi(x)) \})$.
    \item[The step of induction:] Suppose $t, u \in T$ fills $u \in C(t)$ and $l'_{t}$ is already determined.
    Then, for each $j \in S_{u}$, set
  \begin{eqnarray*}
  l^{(1)}_{u}(j) := \left\{\begin{array}{ll}
    \mathsf{TR}_{f(l(t)), f(l(u))}(l'_{t}(j)) & \text{if $j \in J_{t}$,}\\
    l_{u}(j) (= \{ \mu x.\varphi(x) \}) & \text{otherwise.}
  \end{array}\right.
  \end{eqnarray*}
  Next, suppose $j_{1}$ and $j_{2}$ are siblings and $j_{1}$ is older.
  Then for every $\beta \in l^{(1)}_{u}(j_{1}) \cap l^{(1)}_{u}(j_{2})$,
  remove $\beta$ from the labels of $j_{2}$ and its descendants.
  That is, execute {\bf horizontal pruning}.
  In this way, the reduced label of $l^{(1)}_ {u}$ is $l'_{u}$.
  \end{description}

  \item
  The new label $L': T \rightarrow \mathcal{P}(\mathsf{Sub}(\mu x.\varphi(x)))$ is defined as follows:
  \[
    L'(t) := \bigcup_{j \in C_{t}(1)}l'_{t}(j)
  \]
  From the above, we have completed the definition of $\mathcal{T}_{\mu x.\varphi(x)} := (T, C, r, L')$.
\end{steps}
The extended tableau $\mathcal{T}_{\mu x.\varphi(x)}$ is what we want; that is,
$\mathcal{T}_{\varphi(\mu x.\varphi(x))}
    \rightharpoonup\mathcal{T}_{\mu x.\varphi(x)}$
holds by $Z$. To show this,
let's make sure that these satisfy the parity condition.
Take any even infinite branch $\xi$ of $\mathcal{T}_{\mu x.\varphi(x)}$.
Then, there exists an even trace $\mathsf{tr'}$ of $\mathcal{T}_{\mu x.\varphi(x)}$.
If $\mathsf{tr'}$ stays at vertex $N_{\varphi(\mu x.\varphi(x))}$ consecutively
(i.e., $\mathsf{tr'}[n] \in l'_{n}(N_{\varphi(\mu x.\varphi(x))})$ for every $n \geq 1$), then,
from $(\spadesuit)$, we can find even trace $\mathsf{tr}$ of $\mathcal{T}_{\varphi(\mu x.\varphi(x))}$
which stays at vertex $N_{\varphi(\mu x.\varphi(x))}$ consecutively.
Similarly, If $\mathsf{tr'}$ is a trace that stays at vertex $j \neq N_{\varphi(\mu x.\varphi(x))}$ consecutively
(i.e., $\{ n \in \omega \mid \mathsf{tr'}[n] \in l'_{n}(j)\}$ is an infinite set),
then, again from $(\spadesuit)$, we can find even trace $\mathsf{tr}$ of $\mathcal{T}_{\varphi(\mu x.\varphi(x))}$
which stays at vertex $j$ consecutively. Therefore, $Z$ certainly satisfies the parity condition.

\fbox{Part $4$}
First, we divide $\mathsf{Sub}(\varphi(\mathsf{anf}(\alpha)))$ into two disjoint sets:
\begin{align*}
    \mathsf{Sub}_{1} &:=
    \big\{ \beta(\mathsf{anf}(\alpha)) \mid \beta(x) \in \mathsf{Sub}(\varphi(x)) \big\} \setminus \{ \mathsf{anf}(\alpha) \}\\
    \mathsf{Sub}_{2} &:= \mathsf{Sub}(\varphi(\mathsf{anf}(\alpha))) \setminus \mathsf{Sub}_{1}
    (=\mathsf{Sub}(\mathsf{anf}(\alpha)))
\end{align*}
A function $g: \mathsf{Sub}(\varphi(\mathsf{anf}(\alpha))) \rightarrow
\mathcal{P}\big(\mathsf{Sub}(\varphi(\alpha))\big)$ is defined as follows:
\begin{eqnarray*}
g(\psi) := \left\{\begin{array}{ll}
    \{ \beta(\alpha) \} & \text{if $\psi = \beta(\mathsf{anf}(\alpha)) \in \mathsf{Sub}_{1}$,}\\
   f(\psi) & \text{otherwise.}
  \end{array}\right.
\end{eqnarray*}
Here, $f$ is the function mentioned in Remark \ref{rem: automaton normal form}.
Note that
\begin{itemize}
  \item[$(\clubsuit)$] $\mathsf{tr}$ is even ($\mu$-trace)
  $\Leftrightarrow$ $\vec{g}(\mathsf{tr})$ is even (i.e., include $\mu$-trace)
\end{itemize}
holds from the definition of $g$.
Take an extended tableau $\mathcal{T}_{\varphi(\mathsf{anf}(\alpha))} = (T, C, r, L)$ arbitrarily.
Set $Z := \{ (t, t) \mid t \in T \}$. If 
\begin{equation}
  \mathcal{T}_{\varphi(\mathsf{anf}(\alpha))}\rightharpoonup(T, C, r, g\circ L) \label{eq: tableau consequence 2}
\end{equation}
with $Z$, then we are done. However unfortunately $(\ref{eq: tableau consequence 2})$ is generally incorrect
for the same reasons as mentioned in Part $3$.
We will construct an extended tableau $\mathcal{T}_{\varphi(\alpha)}$ where
$\mathcal{T}_{\varphi(\mathsf{anf}(\alpha))} \rightharpoonup\mathcal{T}_{\varphi(\alpha)}$ holds by $Z$
in the following 5 steps:

\begin{steps}
  \item
 We define the B\"{u}chi automaton $\mathcal{BA}_{\varphi(\mathsf{anf}(\alpha))} := \langle Q, \Sigma, q_{0}, \Delta, F \rangle$
  as follows:
  \begin{itemize}
    \item $Q := \{ (\Gamma, \gamma) \mid \Gamma \subseteq \mathsf{Sub}(\varphi(\mathsf{anf}(\alpha))),
    \; \gamma \in \Gamma \}$.
    \item $\Sigma := \mathcal{P}(\mathsf{Sub}(\varphi(\mathsf{anf}(\alpha)))$.
    \item $q_{0} := \big(\{ \varphi(\mathsf{anf}(\alpha)) \}, \varphi(\mathsf{anf}(\alpha))\big)$.
    \item $\Delta(\Gamma', (\Gamma, \gamma)) := \big\{ (\Gamma', \gamma') \mid
    \gamma' \in \mathsf{TR}_{\Gamma, \Gamma'}(\gamma) \big\}$. 
    \item $F := \big\{ q_{0} \big\} \bigcup \big\{ (\Gamma, \mathsf{anf}(\alpha))
    \mid (\Gamma, \mathsf{anf}(\alpha)) \in Q \big\}$.
  \end{itemize}

  \item
  Convert nondeterministic B\"uchi automaton $\mathcal{BA}_{\varphi(\mathsf{anf}(\alpha))}$
  to deterministic Rabin automaton $\mathcal{RA}_{\varphi(\mathsf{anf}(\alpha))}$
  using Safra's construction.
  However, the following two points are changed from the conversion described in
  Subsection \ref{subsec: safra's construction}:
  \begin{itemize}
    \item
    In the {\bf add new children}, 
    change the child added in the first transition from $2$ to $N_{\varphi(\mathsf{anf}(\alpha))}$, similar to the
    method described in Part $3$.

    \item
    In the {\bf initialize the index appearance record}, abolish driving $j$ painted in red to the left end.
    Instead, change it so that it is driven to the left end excluding $N_{\varphi(\mathsf{anf}(\alpha))}$,
    similar to the method described in Part 3.
  \end{itemize}

  \item
  Let the automaton defined above be
  $\mathcal{RA}_{\varphi(\mathsf{anf}(\alpha))} = \langle Q', \Sigma, q_{0}, \Delta', \{ (A_{j}, R_{j}) \mid j \in J \}\rangle$.
  Then, note that the youngest child of the root $1$ is $N_{\varphi(\mathsf{anf}(\alpha))}$,
  labeled with a subset of $\mathsf{Sub}_{1}$.
  The other children of the root $1$ are labeled with a subset of $\mathsf{Sub}_{2}$.

  \item
  For each node $t$, we will define a labeled tree $\langle J_{t}, C_{t}, 1, l'_{t}, \pi_{t}, \mathsf{col}_{t}\rangle$ inductively
  from the root to the leaf; where
  $l'_{t}: S_{t} \rightarrow \mathcal{P}\big(\mathsf{Sub}(\varphi(\alpha))\big)$.
  \begin{description}
    \item[The basis of induction:] $l'_{r}(r) := \{ \varphi(\alpha) \} = g(\{ \varphi(\mathsf{anf}(\alpha)) \})$.
    \item[The step of induction:] Suppose $t, u \in T$ fills $u \in C(t)$ and $l'_{t}$ is already determined.
    Then, for each $j \in S_{u}$, set
  \begin{eqnarray*}
  l^{(1)}_{u}(j) := \left\{\begin{array}{ll}
    \mathsf{TR}_{g(l(t)), g(l(u))}(l'_{t}(j)) & \text{If $j \in J_{t}$,}\\
    l_{u}(j) (= \{ \mathsf{anf}(\alpha) \}) & \text{Otherwise.}
  \end{array}\right.
  \end{eqnarray*}
  Next, suppose $j_{1}$ and $j_{2}$ are siblings and $j_{1}$ is older.
  Then for every $\beta \in l^{(1)}_{u}(j_{1}) \cap l^{(1)}_{u}(j_{2})$,
  remove $\beta$ from the labels of $j_{2}$ and its descendants.
  That is, execute {\bf horizontal pruning}.
  In this way, the reduced label of $l^{(1)}_ {u}$ is $l'_{u}$.
  \end{description}

  \item
  The new label $L': T \rightarrow \mathcal{P}(\mathsf{Sub}(\varphi(\alpha))$ is defined as follows:
  \[
    L'(t) := \bigcup_{j \in C_{t}(1)}l'_{t}(j)
  \]
  From the above, we have completed the definition of $\mathcal{T}_{\varphi(\alpha)} := (T, C, r, L')$.
\end{steps}

The extended tableau $\mathcal{T}_{\varphi(\alpha)}$ is what we want; that is,
$\mathcal{T}_{\varphi(\mathsf{anf}(\alpha))}
    \rightharpoonup\mathcal{T}_{\varphi(\alpha)}$
holds by $Z$. Indeed, from $(\clubsuit)$, we can show that $Z$ satisfies the parity condition, just as we did in Part $3$.
\end{proof}

\begin{Corollary}\label{cor: tableau}
Let $\widehat{\alpha}(x)$ be an automaton normal form in which
$x \in \mathsf{Free}(\widehat{\alpha}(x))$ occurs at once,
positively, moreover, $x$ is in the scope of some modal operators.
Set $\widehat{\varphi} := \mathsf{anf}(\mu x.\widehat{\alpha}(x))$.
Then there exist tableaux $\mathcal{T}_{\widehat{\alpha}(\widehat{\varphi})}$ and $\mathcal{T}_{\widehat{\varphi}}$
such that
$\mathcal{T}_{\widehat{\alpha}(\widehat{\varphi})}\rightharpoonup\mathcal{T}_{\widehat{\varphi}}$.
\end{Corollary}

\begin{proof}
This corollary is proved using three tableaux; Figure \ref{fig: the plan for the proof of the corollary}
depicts the plan of the proof.
\begin{figure}[htbp]
  \centering
  \includegraphics[width=15cm]{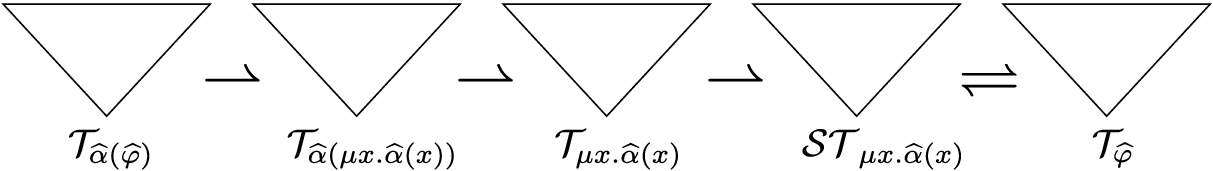}
  \caption{The plan for the proof of the corollary.}
  \label{fig: the plan for the proof of the corollary}
\end{figure}

Let $\mathcal{ST'}_{\mu x.\widehat{\alpha}(x)}$ be a small tableau
for $\varphi$ whose existence is guaranteed by Lemma \ref{lem: small tableau}.
Let $\widehat{\varphi} = \mathsf{anf}(\mu x.\widehat{\alpha}(x))$ be an automaton normal form
generated from $\mathcal{ST'}_{\mu x.\widehat{\alpha}(x)}$.
Let $\mathcal{TB}_{\mu x.\widehat{\alpha}(x)}$ be a tableau with back edge
generated from $\mathcal{ST'}_{\mu x.\widehat{\alpha}(x)}$ in the process of 
creating $\mathsf{anf}(\mu x.\widehat{\alpha}(x))$.
Set $\mathcal{ST}_{\mu x.\widehat{\alpha}(x)} = \mathsf{UNW}_{r}(\mathcal{TB}_{\mu x.\widehat{\alpha}(x)})$.
Note that $\mathcal{ST}_{\mu x.\widehat{\alpha}(x)}$ is also a small tableau.

First, we have
$\mathcal{T}_{\widehat{\alpha}(\widehat{\varphi})}\rightharpoonup
  \mathcal{T}_{\widehat{\alpha}(\mu x.\widehat{\alpha}(x))}$
for some extended tableau $\mathcal{T}_{\widehat{\alpha}(\mu x.\widehat{\alpha}(x))}$;
from Part $4$ of Lemma \ref{lem: basic properties of tableau consequence}.
Second, we have
$\mathcal{T}_{\widehat{\alpha}(\mu x.\widehat{\alpha}(x))}\rightharpoonup
  \mathcal{T}_{\mu x.\widehat{\alpha}(x)}$
for some extended tableau $\mathcal{T}_{\mu x.\widehat{\alpha}(x)}$;
from Part $3$ of Lemma \ref{lem: basic properties of tableau consequence}.
Third,
$\mathcal{T}_{\mu x.\widehat{\alpha}(x)}\rightharpoonup
  \mathcal{ST}_{\mu x.\widehat{\alpha}(x)}$
from Lemma \ref{lem: small tableau}.
Fourth,
since $\mathsf{anf}(\mu x.\widehat{\alpha}(x))$ is generated from
$\mathcal{ST}_{\mu x.\widehat{\alpha}(x)}$, 
$\mathcal{ST}_{\mu x.\widehat{\alpha}(x)}\rightleftharpoons
  \mathcal{T}_{\widehat{\varphi}}$.
Finally, by applying Part $1$ and $2$ of Lemma \ref{lem: basic properties of tableau consequence}
repeatedly,
we obtain
$\mathcal{T}_{\widehat{\alpha}(\widehat{\varphi})}\rightharpoonup\mathcal{T}_{\widehat{\varphi}}$.
\end{proof}

\subsection{Proof of completeness}\label{subsec: proof of completeness}

\begin{Definition}[\bf Aconjunctive formula]\normalfont
Let $\varphi$ be a well-named formula, and $\preceq_{\varphi}$ be its dependency order
(recall Definition \ref{def: alternation depth}). Then,
A variable $x \in \mathsf{Bound}(\varphi)$ is called \textit{aconjunctive} if, for any
$\alpha\wedge\beta \in \mathsf{Sub}(\varphi_{x}(x))$, $x$ is active in at most one of $\alpha$ or
$\beta$.
$\varphi$ is called \textit{aconjunctive} if every $x \in \mathsf{Bound}(\varphi)$ such that
$\eta_{x} = \mu$ is aconjunctive.
\end{Definition}

\begin{Corollary}\label{cor: completeness for anf}
Let $\widehat{\varphi}$ be an automaton normal form. Then, we have
\begin{enumerate}
  \item $\widehat{\varphi}$ is aconjunctive.
  \item If $\widehat{\varphi}$ is not satisfiable, then $\widehat{\varphi} \vdash$.
\end{enumerate}
\end{Corollary}
\begin{proof}
The first assertion of the Corollary is obvious from the observation of Remark \ref{rem: shape of anf}.
For the second assertion, suppose that $\widehat{\varphi}$ is not satisfiable.
Note that, from the definition, a refutation for a aconjunctive formula is always thin.
Then, from Lemma \ref{the: refutation}, there exists a thin refutation for $\widehat{\varphi}$.
From Theorem \ref{the: thin refutation}, we obtain $\widehat{\varphi}\vdash$.
\end{proof}

In the next Lemma, we confirm that some compositions preserve aconjunctiveness.

\begin{Lemma}[\bf Composition]\label{lem: composition}
Let $\varphi$, $\psi$ and $\alpha(x)$ be aconjunctive formulas where
$x \in \mathsf{Prop}$ appears only positively in $\alpha(x)$. Then
$\varphi\wedge\psi$, $\alpha(\varphi)$ and
$\nu \vec{x}.\alpha(\vec{x})$ are also aconjunctive.
\end{Lemma}

\begin{proof}
We leave the proofs of these statements as an exercise to the reader.
\end{proof}
Next, in preparation for proving claim (h), we extend the definition of the trace given in Definition \ref{def: trace}.
\begin{Definition}[\bf An extension of trace]\label{def: an extension of trace}\normalfont
Let $\mathcal{T}_{\varphi} = (T, C, r, L)$ be a tableau for some well-named formula $\varphi$.
Let $\xi$ be a finite or infinite branch of $\mathcal{T}_{\varphi}$ and let $\mathsf{tr}$ be a trace on $\xi$.
The set of all traces on $\xi$ is denoted by $\mathsf{TR}(\xi)$. $\mathsf{TR}(\xi[n, m])$ denotes the set
$\{ \mathsf{tr}[n, m] \mid \mathsf{tr} \in \mathsf{TR}(\xi) \}$
and may also be written
$\mathsf{TR}(\xi[n], \xi[m])$.
For any two factors $\mathsf{tr}[n, m]$ and $\mathsf{tr'}[n', m']$, we say
$\mathsf{tr}[n, m]$ and $\mathsf{tr'}[n', m']$ are \textit{equivalent}
(denoted $\mathsf{tr}[n, m] \equiv \mathsf{tr'}[n', m']$) if,
by ignoring invariant portions of the traces, they can be seen as the same sequence. For example, let;
\[
\xymatrix @C=5mm@R=1mm{
  *{\mathsf{tr}[n, n+3] =}&
  *{\langle (\alpha\wedge\beta)\vee\gamma,} &
  *{(\alpha\wedge\beta)\vee\gamma,} &
  *{\alpha\wedge\beta,} &
  *{\beta \rangle} &
  &\\
  *{\mathsf{tr'}[n', n'+4] =}&
  *{\langle (\alpha\wedge\beta)\vee\gamma,} &
  *{\alpha\wedge\beta,} &
  *{\alpha\wedge\beta,} &
  *{\alpha\wedge\beta,} &
  *{\beta \rangle} &
  }
\]
then $\mathsf{tr}[n, n+3]$ and $\mathsf{tr'}[n', n'+4]$ are equivalent to each other.
Let $X$ and $Y$ be the set of some factors of some traces. Then we write $X \Subset Y$
if for any $\mathsf{tr}[n, m] \in X$ there exists $\mathsf{tr'}[n', m'] \in Y$ such that
$\mathsf{tr}[n, m] \equiv \mathsf{tr'}[n', m']$; and write $X \equiv Y$ if
$X \Subset Y$ and $X \Supset Y$.

For technical reasons, we will need an \textit{extended trace} (denotation: $\mathsf{tr^{+}}$) for
each trace $\mathsf{tr}$ which is constructed by the following procedure $(\dag)$
(see also Figure \ref{fig: an extended trace});
\begin{figure}[htbp]
  \centering
  \includegraphics[width=8cm]{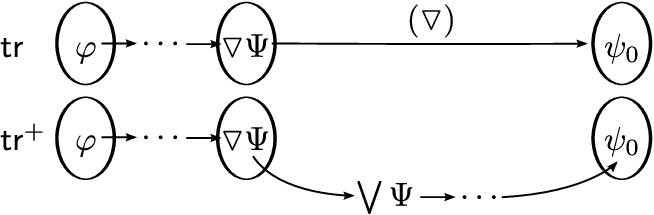}
  \caption{An extended trace.}
  \label{fig: an extended trace}
\end{figure}
\begin{quote}
$(\dag)$: Suppose $\Psi = \{\psi_{0}, \psi_{1}, \dots, \psi_{k}\}$ and that
$\xi[n]$ is a $(\triangledown)$-node in which
$\mathsf{tr}[n] = \triangledown \Psi$ is reduced into $\mathsf{tr}[n+1] = \psi_{0}$. Then, we insert
the sequence
\[
\langle \bigvee \Psi, \bigvee (\Psi\setminus\{\ \psi_{1}\}),
\bigvee (\Psi\setminus\{\ \psi_{1}, \psi_{2}\}),
\dots, \bigvee \{ \psi_{0}, \psi_{k-1}, \psi_{k} \},
\bigvee \{ \psi_{0}, \psi_{k} \} \rangle
\]
between $\mathsf{tr}[n]$ and $\mathsf{tr}[n+1]$.
\end{quote}
Note that $\mathsf{tr}$ is even if and only if $\mathsf{tr}^{+}$ is even because inserted formulas
are all $\vee$-formulas and, thus, the priorities of these formulas are equal to $0$
(recall Equation $(\ref{eq: priority of formulas})$). The set of extended
traces $\mathsf{TR^{+}}(\xi)$
and the set of factors of extended traces $\mathsf{TR^{+}}(\pi[n, m])$
or $\mathsf{TR^{+}}(\pi[n], \pi[m])$ are defined similarly.
\end{Definition}

The next lemma is the claim (h) mentioned in Section \ref{sec: introduction}.
The proof is long, but if you look closely, you can see that it is a natural proof.

\begin{Lemma}\label{lem: completeness for tableau consequence}
Let $\alpha$ be an aconjunctive formula, and $\widehat{\varphi}$ be an automaton normal form. A tableau
$\mathcal{T}_{\alpha} = (T_{\alpha}, C_{\alpha}, r_{\alpha}, L_{\alpha})$ for $\alpha$ and a tableau
$\mathcal{T}_{\widehat{\varphi}} =
  (T_{\widehat{\varphi}}, C_{\widehat{\varphi}}, r_{\widehat{\varphi}}, L_{\widehat{\varphi}})$
for $\widehat{\varphi}$ are given. If $\mathcal{T}_{\widehat{\varphi}}$ is a tableau consequence of
$\mathcal{T}_{\alpha}$, then we can construct a thin refutation $\mathcal{R}$ for
$\alpha\wedge\sim\!\widehat{\varphi}\; (\equiv \; \sim\!(\alpha \rightarrow \widehat{\varphi}))$.
\end{Lemma}

\begin{proof}
Let $\mathcal{T}_{\alpha}$ and $\mathcal{T}_{\widehat{\varphi}}$ be the tableaux satisfying the
condition of the Lemma. Then, there exists a tableau consequence relation $Z$ from $\mathcal{T}_{\alpha}$
to $\mathcal{T}_{\widehat{\varphi}}$. Now, we will construct a thin refutation
$\mathcal{R} = (T, C, r, L)$ for $\alpha \wedge\!\sim\!\widehat{\varphi}$ inductively.
To facilitate the construction, we define two correspondence functions
$\mathsf{Cor}_{\alpha}: T \rightarrow T_{\alpha}$ and
$\mathsf{Cor}_{\widehat{\varphi}}:  T \rightarrow T_{\widehat{\varphi}}$.
These functions are partial and, in every considered node $t$ of $\mathcal{R}$, the following conditions
are satisfied:
\begin{align}
  L(t) = L_{\alpha}(\mathsf{Cor}_{\alpha}(t)) \cup
    \left\{ \sim \bigvee L_{\widehat{\varphi}}(\mathsf{Cor}_{\widehat{\varphi}}(t))\right\}
    \label{eq: thin refutation 01} \\
  (\mathsf{Cor}_{\alpha}(t), \mathsf{Cor}_{\widehat{\varphi}}(t)) \in Z \label{eq: thin refutation 02}
\end{align}
Of course, the root of $\mathcal{R}$ is labeled by $\{ \alpha \wedge \sim\!\widehat{\varphi} \}$ and
its child, say $t_{0}$, is labeled by $\{ \alpha, \sim\!\widehat{\varphi} \}$. For the base step,
set $\mathsf{Cor}_{\alpha}(t_{0}) := r_{\alpha}$ and
$\mathsf{Cor}_{\widehat{\varphi}}(t_{0}) := r_{\widehat{\varphi}}$. Then, the Condition
$(\ref{eq: thin refutation 01})$ and  $(\ref{eq: thin refutation 02})$ are indeed satisfied. The
remaining construction is divided into two cases; the second of which will be further divided into four
cases.
\begin{description}
\item[Inductive step I]
Suppose we have already constructed $\mathcal{R}$ up to a node $t$ where $\mathsf{Cor}_{\alpha}(t)$
and $\mathsf{Cor}_{\widehat{\varphi}}(t)$ are choice nodes of appropriate tableaux and satisfy Conditions
$(\ref{eq: thin refutation 01})$ and  $(\ref{eq: thin refutation 02})$. In this case, we prolong
$\mathcal{R}$ up to $u$ so that:
\begin{enumerate}
\item $\mathsf{Cor}_{\alpha}(u)$ is a modal node of $\mathcal{T}_{\alpha}$ near
  $\mathsf{Cor}_{\alpha}(t)$.
\item $\mathsf{Cor}_{\widehat{\varphi}}(u)$ is a modal node of
  $\mathcal{T}_{\widehat{\varphi}}$ near $\mathsf{Cor}_{\widehat{\varphi}}(t)$.
\item Conditions $(\ref{eq: thin refutation 01})$ and  $(\ref{eq: thin refutation 02})$ are satisfied in
  $u$.
\item $\mathsf{TR}[t, u] \equiv \mathsf{TR}[\mathsf{Cor}_{\alpha}(t), \mathsf{Cor}_{\alpha}(u)] \cup
  \left\{ \langle \sim\!\bigvee L_{\widehat{\varphi}}(t_{1}), \cdots,
  \sim\!\bigvee L_{\widehat{\varphi}}(t_{k}) \rangle \right\}$
  where $t_{1}\cdots t_{k} \in
      T^{+}_{\widehat{\varphi}}$ is the
    $C_{\widehat{\varphi}}$-sequence starting at $\mathsf{Cor}_{\widehat{\varphi}}(t)$ and
    ending at $\mathsf{Cor}_{\widehat{\varphi}}(u)$.
\end{enumerate}
The idea of the prolonging procedure is represented in Figure
\ref{fig: the prolonging procedure for inductive step i}.
\begin{figure}[htbp]
  \centering
  \includegraphics[width=12cm]{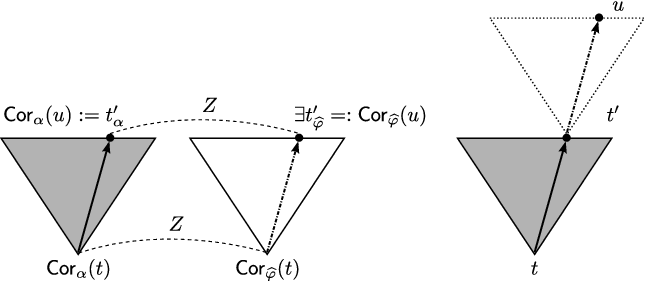}
  \caption{The prolonging procedure for Inductive step I.}
  \label{fig: the prolonging procedure for inductive step i}
\end{figure}
From $t$, we first apply the tableau rules to the formulas of
$\mathsf{Sub}(L_{\alpha}(\mathsf{Cor}_{\alpha}(t)))$ in the same order as they were applied from
$\mathsf{Cor}_{\alpha}(t)$ and its nearest modal nodes. Then, we obtain a finite tree rooted in $t$ which
is isomorphic to the section of $\mathcal{T}_{\alpha}$ between $\mathsf{Cor}_{\alpha}(t)$ and its nearest modal
nodes. Therefore, for each leaf $t'$ of this section of $\mathcal{R}$, we can take unique modal node
$t'_{\alpha}$ of $\mathcal{T}_{\alpha}$ that is isomorphic to $t'$. Note that
$L(t') = L_{\alpha}(t'_{\alpha}) \cup
  \{ \sim\!\bigvee L_{\widehat{\varphi}}(\mathsf{Cor}_{\widehat{\varphi}}(t)) \}$.
Now, the forth condition on the choice node of $Z$ is used. From $(\ref{eq: thin refutation 02})$, we can
find $t'_{\widehat{\varphi}} \in T_{\widehat{\varphi}}$ which is near
$\mathsf{Cor}_{\widehat{\varphi}}(t)$ and satisfies
$(t'_{\alpha}, t'_{\widehat{\varphi}}) \in Z$. Let us look at the path from
$\mathsf{Cor}_{\widehat{\varphi}}(t)$ to $t'_{\widehat{\varphi}}$ in $\mathcal{T}_{\widehat{\varphi}}$.
Since $\widehat{\varphi}$ is an automaton normal form on this path only the $(\vee)$-, $(\eta)$-
and $(\mathsf{Regeneration})$-rules,
and $(\wedge)$-rules reducing $\widehat{\psi}\wedge\top$ to $\{\widehat{\psi}, \top\}$
may be applied first. Then, we have zero or more applications of
the $(\wedge)$-rule.
Let us apply dual rules to $\sim\!\bigvee L_{\widehat{\varphi}}(\mathsf{Cor}_{\widehat{\varphi}}(t))$
(note that $(\mathsf{Regeneration})$ and $(\eta)$ are self-dual).

For an application of the $(\vee)$-rule
in $\mathcal{T}_{\widehat{\varphi}}$, we apply the $(\wedge)$-rule followed by the $(\mathsf{Weak})$-rule
to leave only the conjunct which appears on the path to $t'_{\widehat{\varphi}}$. In this way,
we ensure the resulting path of $\mathcal{R}$ will be thin.

For an application of the $(\wedge)$-rule
reducing $\widehat{\psi}\wedge\top$ to $\{\widehat{\psi}, \top\}$
in $\mathcal{T}_{\widehat{\varphi}}$, we apply
the $(\vee)$-rule in $\mathcal{R}$. Then, we have two children, say $v_{1}$ and $v_{2}$ such that
$L(v_{1})$ includes $\sim\!\widehat{\psi}$ and $L(v_{2})$ includes $\sim\!\top = \bot$.
Since $L(v_{2})$ is inconsistent, if we further prolong $\mathcal{R}$ from $v_{2}$ to its
nearest modal nodes, such modal nodes also labeled inconsistent set. This means that
the modal nodes can be leaves of a refutation. We therefore stop the prolonging procedure
on such modal nodes.

After these reductions, we get a node $u$ which is labeled by
$L_{\alpha}(t'_{\alpha}) \cup \{ \sim \bigvee L_{\widehat{\varphi}}(t'_{\widehat{\varphi}})\}$.
Setting $\mathsf{Cor}_{\alpha}(u) := t'_{\alpha}$ and
$\mathsf{Cor}_{\widehat{\varphi}}(u) := t'_{\widehat{\varphi}}$
establishes Conditions
$(\ref{eq: thin refutation 01})$ and  $(\ref{eq: thin refutation 02})$. Conditions $1$ through $4$
follow directly from the construction.

\item[Inductive step II]
Suppose we have already constructed $\mathcal{R}$ up to a node $t$ where
$\mathsf{Cor}_{\alpha}(t)$ and $\mathsf{Cor}_{\widehat{\varphi}}(t)$ are modal nodes of appropriate
tableaux and satisfy Conditions $(\ref{eq: thin refutation 01})$ and $(\ref{eq: thin refutation 02})$.
Note that, since $\widehat{\varphi}$ is an automaton normal form, we can put
$L_{\widehat{\varphi}}(\mathsf{Cor}_{\widehat{\varphi}}(t))
  = \{ \triangledown \Psi, l_{1}, \dots, l_{i} \}$
or
$L_{\widehat{\varphi}}(\mathsf{Cor}_{\widehat{\varphi}}(t))
  = \{ l_{1}, \dots, l_{i} \}$
where $l_{1}, \dots, l_{i} \in \mathsf{Lit}(\widehat{\varphi})$. Moreover, observe that
\begin{align*}
  \sim\!\left(\triangledown \Psi \wedge \bigwedge_{1 \leq k \leq i} l_{k}\right)&
    \equiv\: \sim\!\triangledown \Psi \vee \left(\bigvee_{1 \leq k \leq i}\sim\!l_{k}\right)\\
  &\equiv\: \sim\!\left( \left(\bigwedge \Diamond \Psi\right) \wedge \square
    \left(\bigvee \Psi\right) \right)
    \vee \left(\bigvee_{1 \leq k \leq i}\sim\!l_{k}\right)\\
  &\equiv\: \left(\bigvee_{\psi \in \Psi} \square \sim\!\psi\right) \vee
    \Diamond \left(\bigwedge \sim\!\Psi\right)
    \vee \left(\bigvee_{1 \leq k \leq i}\sim\!l_{k}\right)\\
  &\equiv\: \left(\bigvee_{\psi \in \Psi} (\triangledown \{\sim\!\psi\} \vee
    \triangledown \emptyset)\right) \vee
    \triangledown \left\{\left(\bigwedge \sim\!\Psi\right), \top\right\}
    \vee \left(\bigvee_{1 \leq k \leq i}\sim\!l_{k}\right).
\end{align*}
Therefore, if we prolong $\mathcal{R}$ from $t$ up to its nearest modal nodes $u$
by applying the $(\vee)$-rule repeatedly,
the label of $u$ can be categorized as one of following four cases:
\begin{description}
  \item[(Case 1):] $L(u) = L_{\alpha}(\mathsf{Cor}_{\alpha}(t))\cup
    \{ \sim\!l_{k} \}$ for some $k$ such that $1 \leq k \leq i$.
  \item[(Case 2):] $L(u) = L_{\alpha}(\mathsf{Cor}_{\alpha}(t))\cup\{ \triangledown \emptyset \}$.
  \item[(Case 3):] $L(u) = L_{\alpha}(\mathsf{Cor}_{\alpha}(t))\cup\{ \triangledown \{\sim\!\psi\}\}$
    for some $\psi \in \Psi$.
  \item[(Case 4):] $L(u) = L_{\alpha}(\mathsf{Cor}_{\alpha}(t))\cup
    \left\{ \triangledown \left\{\left(\bigwedge\!\sim\!\Psi\right), \top\right\} \right\}$.
\end{description}
In every cases, it is possible that $L_{\alpha}(\mathsf{Cor}_{\alpha}(t))$ is inconsistent and, thus,
$L(u)$ is also inconsistent. If this is so, all $u$ can be a leaf of a refutation.
Therefore, we stop the prolonging procedure on $u$ in this case. Now, we consider the case where
$L_{\alpha}(\mathsf{Cor}_{\alpha}(t))$ is consistent.

In Case $1$, the prop condition is used; by Condition $(\ref{eq: thin refutation 02})$, we have
$l_{k} \in L_{\alpha}(\mathsf{Cor}_{\alpha}(t))$. Thus, $L(u)$ includes $l_{k}$ and $\sim\!l_{k}$. This
means that $L(u)$ is inconsistent and so $u$ can be a leaf of a refutation. We therefore
stop the prolonging procedure on $u$ in this case.

In Case $2$, the back condition on modal nodes is used. Since
$C_{\widehat{\varphi}}(\mathsf{Cor}_{\widehat{\varphi}}(t)) \neq \emptyset$,
it must hold that $C_{\alpha}(\mathsf{Cor}_{\alpha}(t)) \neq \emptyset$.
Take $v_{\alpha} \in C_{\alpha}(\mathsf{Cor}_{\alpha}(t))$ arbitrarily. We prolong $\mathcal{R}$ from
$u$ to $v \in C(u)$ in such a way that
$L(v) = L_{\alpha}(v_{\alpha}) \cup \{ \bigvee \emptyset (\equiv \bot) \}$. Since $L(v)$ is
inconsistent, if we further prolong $\mathcal{R}$ from $v$ to its nearest modal nodes, such modal
nodes are also inconsistent. This means that the modal nodes can be leaves of a
refutation. We therefore stop the prolonging procedure on such modal nodes in this case.

In Case $3$, the back condition on modal nodes is used. Let $v_{\widehat{\varphi}}$ be a child of
$\mathsf{Cor}_{\widehat{\varphi}}(t)$ such that
$L_{\widehat{\varphi}}(v_{\widehat{\varphi}}) = \{ \psi \}$. Then, by Condition
$(\ref{eq: thin refutation 02})$,
we can find $v_{\alpha} \in C_{\alpha}(\mathsf{Cor}_{\alpha}(t))$ such that
$(v_{\alpha}, v_{\widehat{\varphi}}) \in Z$. We create a new child $v$ of $u$ which is labeled
by $L_{\alpha}(\mathsf{Cor}_{\alpha}(v_{\alpha})) \cup \{ \sim\!\psi\}$. Moreover,
we set
$\mathsf{Cor}_{\alpha}(v) := v_{\alpha}$ and
$\mathsf{Cor}_{\widehat{\varphi}}(v) := v_{\widehat{\varphi}}$. This prolonging procedure
preserves Conditions $(\ref{eq: thin refutation 01})$ and $(\ref{eq: thin refutation 02})$.
Note that, in this case, $\mathsf{Cor}_{\alpha}(v)$ and $\mathsf{Cor}_{\widehat{\varphi}}(v)$ are choice
nodes of appropriate tableaux.

In Case $4$, the forth condition on modal nodes is used. The idea of the prolonging procedure is represented
in Figure \ref{fig: the prolonging procedure for case 4}.
\begin{figure}[htbp]
  \centering
  \includegraphics[width=14cm]{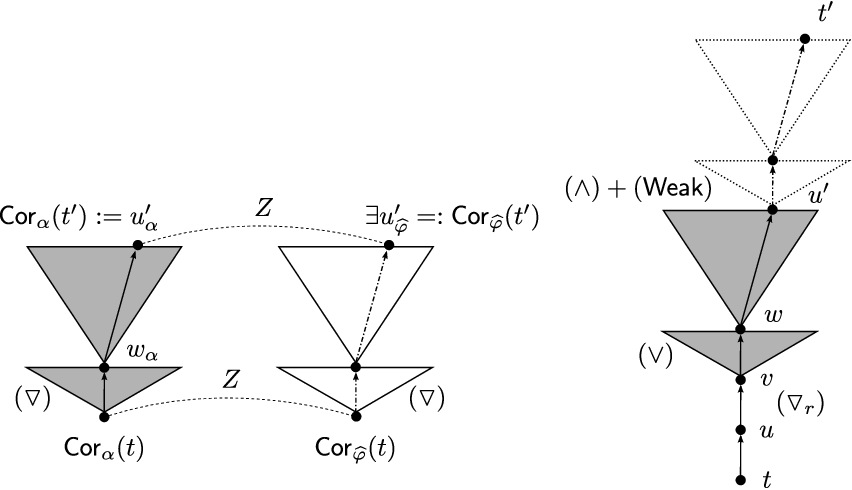}
  \caption{The prolonging procedure for Case $4$.}
  \label{fig: the prolonging procedure for case 4}
\end{figure}
Let
$L_{\alpha}(\mathsf{Cor}_{\alpha}(t)) =
  \{ \triangledown \Delta_{1}, \dots, \triangledown \Delta_{i}, l_{1}, \dots, l_{j}\}$.
In this case, we first create a new child $v$ of $u$ such that
\[
  L(v) = \left\{
    \bigvee \Delta_{1}, \dots, \bigvee \Delta_{i}
  \right\} \cup
  \left\{
    \bigwedge\!\sim\!\Psi
  \right\}.
\]
From the choice node $v$, we further prolong $\mathcal{R}$ up to its nearest modal nodes $t'$
so that
\begin{enumerate}
\setcounter{enumi}{4}
  \item $\mathsf{Cor}_{\alpha}(t')$ is a next modal node of $\mathsf{Cor}_{\alpha}(t)$.
  \item $\mathsf{Cor}_{\widehat{\varphi}}(t')$ is a next modal node of $\mathsf{Cor}_{\widehat{\varphi}}(t)$.
  \item Condition $(\ref{eq: thin refutation 01})$ and 
    $(\ref{eq: thin refutation 02})$ are satisfied in $t'$.
  \item $\mathsf{TR}[u, t'] \equiv \mathsf{TR}^{+}[\mathsf{Cor}_{\alpha}(t), \mathsf{Cor}_{\alpha}(t')] \cup
    \{ \langle
    \triangledown \left\{\left(\bigwedge\!\sim\!\Psi\right), \top\right\},
    \bigwedge\!\sim\!\Psi, \dots,
    \sim\!\psi =
    \sim\!\bigvee L_{\widehat{\varphi}}(t_{1}), \cdots, 
    \sim\!\bigvee L_{\widehat{\varphi}}(t_{k})
    \rangle \}$
    where $t_{1}\cdots t_{k} \in T_{\widehat{\varphi}}^{+}$ is the
    $C_{\widehat{\varphi}}$-sequence starting at the child of $\mathsf{Cor}_{\widehat{\varphi}}(t)$
    labeled by $\{ \psi \}$ and ending at $\mathsf{Cor}_{\widehat{\varphi}}(t')$.
\end{enumerate}
Next, we apply $(\vee)$-rules to $\bigvee \Delta_{1}$ repeatedly until we arrive at the node
$w$ such that
\begin{equation*}
  L(w) =
  \{ \delta_{1} \} \cup
  \left\{ \bigvee \Delta_{2}, \dots, \bigvee \Delta_{i} \right\} \cup
  \left\{ \bigwedge\!\sim\!\Psi \right\}
\end{equation*}
where $\delta_{1} \in \Delta_{1}$. Note that there exists
$w_{\alpha} \in C_{\alpha}(\mathsf{Cor}_{\alpha}(t))$ such that
\begin{equation*}
  L_{\alpha}(w_{\alpha}) =
  \{ \delta_{1} \} \cup
  \left\{ \bigvee \Delta_{2}, \dots, \bigvee \Delta_{i} \right\}
\end{equation*}
From $w$, we apply the tableau rules to formulas of
$\mathsf{Sub}(L_{\alpha}(w_{\alpha}))$ in the same order as they were applied from
$w_{\alpha}$ and its nearest modal nodes. Then, we obtain a finite tree rooted in $w$ which
is isomorphic to the section of $\mathcal{T}_{\alpha}$ between $w_{\alpha}$ and nearest modal
nodes. Therefore, for each leaf $u'$ of this section of $\mathcal{R}$, we can take a unique modal node
$u'_{\alpha}$ of $\mathcal{T}_{\alpha}$ which is isomorphic to $u'$. Note that
$L(u') = L_{\alpha}(u'_{\alpha}) \cup
  \left\{
    \bigwedge\!\sim\!\Psi
  \right\}$.
Since $u'_{\alpha}$ is a next modal node of $\mathsf{Cor}_{\alpha}(t)$,
from Condition $(\ref{eq: thin refutation 02})$ and the forth condition on modal nodes,
we can assume that there exists
$u'_{\widehat{\varphi}}$ which is a next modal node of $\mathsf{Cor}_{\widehat{\varphi}}(t)$
and satisfies $(u'_{\alpha}, u'_{\widehat{\varphi}}) \in Z$.
We will now look at the path from
$\mathsf{Cor}_{\widehat{\varphi}}(t)$ to $t'_{\widehat{\varphi}}$ in $\mathcal{T}_{\widehat{\varphi}}$
and exploit $(\wedge)$-rules and
$(\mathsf{Weak})$-rules so that
the trace $\mathsf{tr}$ on this path satisfies Condition $8$. Finally, we get a node $t'$
which is labeled by
$L_{\alpha}(u'_{\alpha})\cup \{ \sim \! \bigvee
  L_{\widehat{\varphi}}(u'_{\widehat{\varphi}}) \}$.
Setting $\mathsf{Cor}_{\alpha}(t') := u'_{\alpha}$ and
$\mathsf{Cor}_{\widehat{\varphi}}(t') := u'_{\widehat{\varphi}}$
establishes Conditions
$(\ref{eq: thin refutation 01})$ and $(\ref{eq: thin refutation 02})$. Then, Conditions $5$ through $8$
follow directly from the construction.
\end{description}

The above two procedures completely describe $\mathcal{R}$. All the leaves are labeled by
an inconsistent set.
Moreover, take an infinite branch $\xi$ of $\mathcal{R}$ arbitrarily.
Let $\xi_{\alpha}$ be the branch of $\mathcal{T}_{\alpha}$ such that
$\{ n \in \omega \mid \mathsf{Cor}_{\alpha}(\xi) = \xi_{\alpha}[n] \}$ is an infinite set.
Let $\xi_{\widehat{\varphi}}$ be the branch of $\mathcal{T}_{\widehat{\varphi}}$ such that
$\{ n \in \omega \mid \mathsf{Cor}_{\widehat{\varphi}}(\xi) = \xi_{\widehat{\varphi}}[n] \}$ is an infinite set.
For any trace $\mathsf{tr} \in \mathsf{TR}(\xi)$, we have
$\mathsf{tr}[1] = \alpha \wedge\!\sim\!\widehat{\varphi}$
and, $\mathsf{tr}[2] = \alpha$ or $\mathsf{tr}[2] = \sim\!\widehat{\varphi}$.
$\mathsf{TR}_{1}(\xi)$ denotes the set of all the trace $\mathsf{tr} \in \mathsf{TR}(\xi)$ such that
$\mathsf{tr}[2] = \alpha$. $\mathsf{tr}_{2} \in \mathsf{TR}(\xi)$ denotes the trace such that
$\mathsf{tr}_{2}[2] = \sim\!\widehat{\varphi}$. Then, from the construction of $\mathcal{R}$, we have;
\begin{description}
\item[(T1)]
  $\mathsf{TR}(\xi) = \mathsf{TR}_{1}(\xi) \cup \{ \mathsf{tr}_{2} \}$.
\item[(T2)]
  $\mathsf{TR}^{+}_{1}(\xi) \equiv \mathsf{TR}^{+}(\xi_{\alpha})$.
\item[(T3)]
  $\mathsf{tr}_{2}$ is even if and only if $\xi_{\widehat{\varphi}}$ is odd.
\item[(T4)]
  $\xi_{\alpha}$ and $\xi_{\widehat{\varphi}}$ are associated with each other.
\end{description}
Above conditions imply that $\xi$ is odd. Indeed, if $\xi_{\alpha}$ is odd, then,
from {\bf (T2)}, $\xi$ is also odd.
If $\xi_{\alpha}$ is even, then, from {\bf (T4)}, $\xi_{\widehat{\varphi}}$ is also even.
Therefore, from {\bf (T3)}, $\mathsf{tr}_{2}$ is odd. From {\bf (T1)}, we can assume that $\xi$ is odd.
$\mathcal{R}$ is also thin because $\alpha$
is aconjunctive and whenever we reduce a $\wedge$-formula originated from $\sim\!\widehat{\varphi}$,
we leave only one conjunction and discard the other by applying $(\mathsf{Weak})$-rule.
Therefore, $\mathcal{R}$ is a thin refutation as required.
\end{proof}

\begin{Lemma}[\bf Main lemma]\label{lem: completeness}
For any well-named formula $\varphi$, there exists a semantically equivalent automaton normal form $\widehat{\varphi}$ such that $\varphi \rightarrow \widehat{\varphi}$ is provable in $\mathsf{Koz}$.
Moreover, for any $x \in \mathsf{Free}(\varphi)$ which occurs only positively in $\varphi$,
it hold that $x \in \mathsf{Free}(\widehat{\varphi})$ and $x$ occurs only positively in
$\widehat{\varphi}$.
\end{Lemma}

\begin{proof}
We prove the lemma by the induction on the structure of $\varphi$.

\begin{description}

\item[Case: $\varphi \in \mathsf{Lit}$.]
In this case, $\widehat{\varphi}$ is just $\varphi$.

\item[Case: $\varphi = \alpha \vee \beta$.]
By the induction assumption, there exist automaton normal forms $\widehat{\alpha}$ and
$\widehat{\beta}$ which are equivalent to $\alpha$ and $\beta$, respectively, such that
$\vdash \alpha \rightarrow \widehat{\alpha}$ and $\vdash \beta \rightarrow \widehat{\beta}$.
Set $\widehat{\varphi} := \widehat{\alpha} \vee \widehat{\beta}$. Then, we have
$\vdash \alpha \vee \beta \rightarrow \widehat{\varphi}$.

\item[Case: $\varphi = \triangledown \Psi$.]
This case is very similar to the previous one.

\item[Case: $\varphi = \alpha \wedge \beta$.]
By the induction assumption, there exist automaton normal forms $\widehat{\alpha}$ and $\widehat{\beta}$
which are equivalent to $\alpha$ and $\beta$ respectively, such that
$\vdash \alpha \rightarrow \widehat{\alpha}$ and $\vdash \beta \rightarrow \widehat{\beta}$;
thus, we have $\vdash \alpha \wedge \beta \rightarrow \widehat{\alpha} \wedge \widehat{\beta}$.
Set $\widehat{\varphi} := \mathsf{anf}(\widehat{\alpha} \wedge \widehat{\beta})$. Then,
from Theorem \ref{the: automaton normal form}, we have
$\mathcal{T}_{\widehat{\alpha} \wedge \widehat{\beta}} \rightleftharpoons \mathcal{T}_{\widehat{\varphi}}$
for some $\mathcal{T}_{\widehat{\alpha} \wedge \widehat{\beta}}$
and, thus,
$\mathcal{T}_{\widehat{\alpha} \wedge \widehat{\beta}} \rightharpoonup \mathcal{T}_{\widehat{\varphi}}$. 
On the other hand, by Lemma \ref{lem: composition}, we can assume that
$\widehat{\alpha} \wedge \widehat{\beta}$ is
aconjunctive. From Lemma \ref{lem: completeness for tableau consequence}
and Theorem \ref{the: thin refutation}, we have
$\vdash \widehat{\alpha} \wedge \widehat{\beta} \rightarrow \widehat{\varphi}$. Therefore, we have
$\vdash \alpha \wedge \beta \rightarrow \widehat{\varphi}$.

\item[Case: $\varphi = \nu x.\alpha(x)$.]
By the induction assumption, we have an equivalent automaton normal form $\widehat{\alpha}(x)$
of $\alpha(x)$ such
that $\vdash \alpha(x) \rightarrow \widehat{\alpha}(x)$. Therefore,
$\vdash \nu x.\alpha(x) \rightarrow \nu x.\widehat{\alpha}(x)$.
Set
$\widehat{\varphi} := \mathsf{anf}(\nu x.\widehat{\alpha}(x))$.
Then, from Theorem
\ref{the: automaton normal form}, we have
$\mathcal{T}_{\nu x.\widehat{\alpha}(x)} \rightleftharpoons \mathcal{T}_{\widehat{\varphi}}$
for some $\mathcal{T}_{\nu x.\widehat{\alpha}(x)}$
and, thus,
$\mathcal{T}_{\nu x.\widehat{\alpha}(x)} \rightharpoonup \mathcal{T}_{\widehat{\varphi}}$.
On the other hand, by Lemma \ref{lem: composition}, we can assume that
$\nu x.\widehat{\alpha}(x)$ is
aconjunctive. From Lemma \ref{lem: completeness for tableau consequence}
and Theorem \ref{the: thin refutation}, we have
$\vdash \nu x.\widehat{\alpha}(x) \rightarrow \widehat{\varphi}$. Therefore,
$\vdash \nu x.\alpha(x) \rightarrow \widehat{\varphi}$.

\item[Case: $\varphi = \mu x.\alpha(x)$.]
By the induction assumption, we have an equivalent automaton normal form $\widehat{\alpha}(x)$
of $\alpha(x)$ such
that $\vdash \alpha(x) \rightarrow \widehat{\alpha}(x)$. Therefore,
$\vdash \mu x.\alpha(x) \rightarrow \mu x.\widehat{\alpha}(x)$.
Set
$\widehat{\varphi} := \mathsf{anf}(\mu x.\widehat{\alpha}(x))$.
Then, from Corollary
\ref{cor: tableau}, we have
$\mathcal{T}_{\widehat{\alpha}(\widehat{\varphi})} \rightharpoonup \mathcal{T}_{\widehat{\varphi}}$
for some $\mathcal{T}_{\widehat{\alpha}(\widehat{\varphi})}$.
On the other hand, by Lemma \ref{lem: composition}, we can assume that
$\widehat{\alpha}(\widehat{\varphi})$ is
aconjunctive. From Lemma \ref{lem: completeness for tableau consequence}
and Theorem \ref{the: thin refutation},
$\vdash \widehat{\alpha}(\widehat{\varphi}) \rightarrow \widehat{\varphi}$.
By applying the $(\mathsf{Ind})$-rule, we obtain 
$\vdash \mu x.\widehat{\alpha}(x) \rightarrow \widehat{\varphi}$.
Thus,
$\vdash \mu x.\alpha(x) \rightarrow \widehat{\varphi}$.

\end{description}
Hence, we have proved the Lemma for all cases.
\end{proof}

\begin{Theorem}[\bf Completeness]
For any formula $\varphi$, if $\varphi$ is not satisfiable, then $\sim\!\varphi$ is provable in
$\mathsf{Koz}$.
\end{Theorem}

\begin{proof}
Let $\varphi$ be an unsatisfiable formula. By Part $5$ of Lemma \ref{lem: basic properties of KOZ 01}, we
can construct a well-named formula $\mathsf{wnf}(\varphi)$ such that
\begin{equation}\label{eq: completeness 01}
  \vdash \varphi \leftrightarrow \mathsf{wnf}(\varphi)
\end{equation}
On the other hand, from Lemma \ref{lem: completeness}, there
exists an automaton normal form $(\mathsf{wnf}(\varphi))\verb|^|$ which is semantically equivalent to
$\mathsf{wnf}(\varphi)$ and thus to $\varphi$ such that
\begin{equation}\label{eq: completeness 02}
  \vdash \mathsf{wnf}(\varphi) \rightarrow (\mathsf{wnf}(\varphi))\verb|^|
\end{equation}
Since $(\mathsf{wnf}(\varphi))\verb|^|$ is not satisfiable, by Corollary \ref{cor: completeness for anf}
we have
\begin{equation}\label{eq: completeness 03}
  \vdash (\mathsf{wnf}(\varphi))\verb|^| \rightarrow \bot
\end{equation}
Finally by combining Equations $(\ref{eq: completeness 01})$ through $(\ref{eq: completeness 03})$ we
obtain $\vdash \varphi \rightarrow \bot$
as required.
\end{proof}

\subsection*{Acknowledgements}
The author gives great thanks to the anonymous referees for their valuable comments; in particular, they pointed out
some mathematical errors (especially in Lemma \ref{lem: basic properties of tableau consequence})
in a preliminary version of this paper.

\bibliographystyle{amsplain}
\bibliography{ref.bib}

\begin{thebibliography}{10}
\bibitem{tarski1955}
Alfred Tarski.
\newblock A lattice-theoretical fixpoint theorem and its applications.
\newblock {\em Pacific Journal of Mathematics}, 5(2):285--309, 1955.

\bibitem{DBLP:conf/fossacs/CateF10}
Balder ten Cate and Ga{\"e}lle Fontaine.
\newblock An easy completeness proof for the modal mu-calculus on finite trees.
\newblock In {\em FOSSACS}, pages 161--175, 2010.

\bibitem{Niwinski199699}
Damian Niwinski and Igor Walukiewicz.
\newblock Games for the $\mu$-calculus.
\newblock {\em Theor. Comput. Sci.}, 163(1–2):99-116, 1996.

\bibitem{conf/mfcs/JaninW95}
David Janin and Igor Walukiewicz.
\newblock Automata for the modal mu-calculus and related results.
\newblock In Jirí Wiedermann and Petr Hájek, editors, {\em MFCS}, volume 969
  of {\em Lecture Notes in Computer Science}, pages 552--562. Springer, 1995.

\bibitem{DBLP:journals/tcs/Kozen83}
Dexter Kozen.
\newblock Results on the propositional mu-calculus.
\newblock {\em Theor. Comput. Sci.}, 27:333--354, 1983.

\bibitem{DBLP:conf/dagstuhl/2001automata}
Erich Gr{\"a}del, Wolfgang Thomas, and Thomas Wilke, editors.
\newblock Automata, Logics, and Infinite Games: A Guide to Current
  Research [outcome of a Dagstuhl seminar, February 2001], volume 2500 of {\em
  Lecture Notes in Computer Science}. Springer, 2002.

\bibitem{Lenzi05}
G.~Lenzi.
\newblock The modal $\mu$-calculus: a survey.
\newblock {\em TASK Quarterly}, 9(3):293--316, 2005.

\bibitem{Walukiewicz2000142}
Igor Walukiewicz.
\newblock Completeness of kozen's axiomatisation of the propositional
  $\mu$-calculus.
\newblock {\em Information and Computation}, 157(1–2):142 -- 182, 2000.

\bibitem{Bradfield07modalmu-calculi}
Julian Bradfield and Colin Stirling.
\newblock Modal mu-calculi.
\newblock In {\em HANDBOOK OF MODAL LOGIC}, pages 721--756. Elsevier, 2007.

\bibitem{Kret2017}
J. K{\v r}et\`ins\`y, T. Meggendorfer, C. Waldmann, and M. Weininger.
\newblock Index appearance record for transforming rabin automata into parity automata.
\newblock In {\em TACAS 2017}, pages 443--460, 2017.

\bibitem{bak-sco69}
J.W. de~Bakker and D.S. Scott.
\newblock A theory of programs.
\newblock {\em Unpublished Manuscript, IBM, Vienna}, 1969.

\bibitem{DBLP:journals/logcom/Alberucci09}
Luca Alberucci.
\newblock Sequent calculi for the modal $\mu$-calculus over $\mathsf{S}5$.
\newblock {\em J. Log. Comput.}, 19(6):971--985, 2009.

\bibitem{Bezhanishvili20121}
Nick Bezhanishvili and Ian Hodkinson.
\newblock Sahlqvist theorem for modal fixed point logic.
\newblock {\em Theor. Comput. Sci.}, 424(0):1 -- 19, 2012.

\bibitem{safra1988}
S. Safra. 
\newblock On the complexity of $\omega$-automata.
\newblock {\em 29th IEEE FOCS}, 1988.

\end{thebibliography}
\end{document}